\documentclass[11pt,reqno]{amsart}
\usepackage{dsfont, amssymb,amsmath,amscd,latexsym, amsthm, amsxtra,amsfonts,enumerate}
\allowdisplaybreaks[4]

\DeclareMathOperator*{\argmax}{arg\,max}
\usepackage{caption,subcaption}
\usepackage[all]{xy}
\usepackage[active]{srcltx}
\usepackage{mathrsfs}
\usepackage[all]{xy}
\usepackage[active]{srcltx}
\usepackage{graphicx}
\usepackage{bm}
\usepackage{bbm}
\usepackage{graphicx}
\usepackage{tabularx}
\usepackage{caption}
\usepackage{colortbl}
\usepackage{xcolor}
\usepackage{algorithmicx,algorithm}
\usepackage{longtable}
\usepackage[round]{natbib}
\usepackage{ulem}
\usepackage{verbatim}
\usepackage{appendix}
\usepackage{color}
\usepackage{threeparttablex}
\usepackage{verbatim}
\definecolor{dark-gray}{gray}{0.9}
\usepackage[pdfstartview=FitH, bookmarksnumbered=true,bookmarksopen=true, colorlinks=true, pdfborder=001, citecolor=blue, linkcolor=blue,urlcolor=blue]{hyperref}
\usepackage{CJK}
\newtheorem{theorem}{Theorem}[section]
\newtheorem{lemma}[theorem]{Lemma}
\newtheorem{corollary}[theorem]{Corollary}
\newtheorem{proposition}[theorem]{Proposition}

\newtheorem{remark}[theorem]{Remark}

\newtheorem{Def}{Definition}[section]

\begin{document}
    \makeatletter
    \def\@setauthors{%
        \begingroup
        \def\thanks{\protect\thanks@warning}%
        \trivlist \centering\footnotesize \@topsep30\p@\relax
        \advance\@topsep by -\baselineskip
        \item\relax
        \author@andify\authors
        \def\\{\protect\linebreak}%
        {\authors}%
        \ifx\@empty\contribs \else ,\penalty-3 \space \@setcontribs
        \@closetoccontribs \fi
        \endtrivlist
        \endgroup } \makeatother
    \baselineskip 17pt
    \title[{{\tiny MODERN TONTINE WITH TRANSACTION COSTS}}]
    {{\tiny MODERN TONTINE WITH TRANSACTION COSTS }}
    \vskip 10pt\noindent
    \author[{Lin He, Zongxia Liang, Sheng Wang}]
    {\tiny {\tiny  Lin He$^{a,\dag}$, Zongxia Liang$^{b,\ddag}$, Sheng Wang$^{b,*}$ }
        \vskip 10pt\noindent
        {\tiny ${}^a$School of Finance, Renmin University of China, Beijing
            100872, China \vskip 10pt\noindent\tiny ${}^b$Department of
            Mathematical Sciences, Tsinghua University, Beijing 100084, China }
        \footnote{\\
            $ \dag$ email: helin@ruc.edu.cn\\
            $ \ddag$ email: liangzongxia@mail.tsinghua.edu.cn \\
            $*$ Corresponding author, email:   wangs20@mails.tsinghua.edu.cn  }}
    \maketitle
    \noindent
    %
\begin{abstract}
  In this paper, we propose a new type of modern tontine with transaction costs. The wealth of the retiree is allocated into a bequest account and a tontine account. The retiree can freely purchase and redeem the share in the tontine account while bearing the fixed and proportional transaction costs depending on the transaction volume. The bequest account is regarded as liquidity account and the consumption is directly withdrawn from this account. The retiree dynamically controls the allocation policy between the two accounts and the consumption policy to maximize the consumption and bequest utilities. We formulate the optimization problem as a combined stochastic and impulse control problem with infinite time horizon,  and characterize the value function as the unique viscosity solution of a HJBQVI (Hamilton-Jacobi-Bellmen quasi-variational inequality). The numerical results exhibit the V-shaped transaction region which consists of two stages. In the former stage, as longevity credits increase gradually,  the retiree decreases the proportion of wealth in the tontine account to smooth the wealth and reduce the volatility. However,  in the latter stage, the retiree increases the proportion of wealth in the tontine account to gamble for the considerable longevity credits.
\vskip 10pt  \noindent
 Submission Classification: IB81, IB13, IE13, IE43, IE53.
   \vskip 5pt \noindent
\end{abstract}
\setcounter{equation}{0}

\section{{ {\bf Introduction}}}
Tontine, invented by the Italian Lorenzo de Tonti in the 17th century, can simply diversify the longevity risks among a group of retirees. In the tontine plan, the survival member receives an increasing share when some member dies, and the last survival member gets the remaining share (cf. \cite{M2009}).
It then becomes popular in several European countries in the 18th century, but soon goes bankrupt.
With the increasing severity of aging problem, the traditional life annuity  is under pressure because of its high premium and inflexibility. Thus, tontine gets fresh attention
in recent years. Modern tontine schemes are well-designed to meet the diversified  needs of the participants, such as flexible withdrawals and bequest motives.  Comparing with life annuity, modern tontine has the following three advantages. First, the longevity risk can be diversified by the participants themselves and there is no need for solvency reserves. Thus, it reduces the management costs (cf. \cite{HMMS2010}). Second, the longevity credits increase with age, which is more in line with the consumption needs of the elderly (cf. \cite{WG2021}). Third, there are more flexible withdrawal mechanisms to meet the personalized objectives (cf. \cite{HD2006}, \cite{HMS2008} and \cite{MS2015}).



Like most of the longevity insurance products, traditional tontine cannot meet the needs for bequest motive. Accordingly, scholars have designed a variety of modern tontine mechanisms to meet the needs. In \cite{BD2019}, the wealth is  divided into two accounts: bequest account and tontine account. Furthermore, the proportions of wealth in the two accounts are time-invariant and the tontine account receives continuous longevity credits while the retiree is alive. Another interesting setting is that the two accounts are consolidated into one account when the retiree invests and consumes.  \cite{D2021} extends the work of \cite{BD2019} by assuming the proportion be time-varying. In \cite{CR2022}, the death benefit is included. The bequest and the consumption are both control variables, and the retirees choose the optimal policies to maximize the weighted utility of the consumption and the bequest. Inspired by these ideas, we aim to design a modern tontine mechanism which is more flexible to meet the diversified needs for consumption and bequest, and reasonable for practical implementation.

In early studies, there is a lack of the mechanism for freely buying and redeeming tontine shares. 
That is, the retiree's wealth is divided into the tontine and bequest accounts with a fixed proportion, and the investment and the consumption proportions are the same in the two accounts. Besides, the  withdrawals from the two accounts are used for immediate consumption. Though the subsequent studies have relaxed the restrictions on the proportions, there are still some artificial rules. Such as that the proportion of the tontine account gradually approaches $0$ when the retiree is extremely old. These artificial rules
narrow the admissible domain of the optimization problem. Under these settings, the policy is hard to flexibly meet the heterogeneous needs of the retirees in practice.
We guess that the lack of such  mechanism design is mainly caused by two issues. On the one hand, there exists serious moral hazard problem. For example, if a retiree discovers that she/he has a serious illness, she/he prefers to withdraw the tontine immediately. Therefore, the rest are the ones who have longer lifespan and the longevity risk cannot be diversified. On the other hand, the redemption of the tontine also increases the volatility of the following payout and the remaining participants need risk compensation (cf. \cite{Weinert2017}).
In recent practice, we observe that the above two issues can be relieved by charging the transaction fees. Specifically, when the insured wants to withdraw a longevity insurance, she/he can only get back the cash value after deducting the surrender fees. The deducted part can be regarded as the fee charged to limit the greater volatility and moral hazard problem.

In this paper, we propose a new modern tontine mechanism with transaction fees. The retiree's wealth  allocated into the bequest account and tontine account at time $t$ are denoted by $X(t)$ and $Y(t)$. 
Unlike the settings in the literature, the wealth in the bequest account is reserved in a risk-free asset. Therefore, the bequest account can be regarded as a liquidity account and the consumption can only be withdrawn from this account. In practice, frictionless timely consumption can only be obtained in a liquidity account. In addition, retirees may want to retain some liquidity and make appropriate arrangements for subsequent consumption. In the tontine account, we assume that the longevity credits are paid at the rate $\lambda(t)Y(t)$. $\lambda(t)$ is the force of mortality of the retiree at time $t$, which follows a deterministic and non-negative function. Besides, as most tontine plans have investment function, we assume that the tontine wealth is fully invested in a risky asset (usually a stock index). Particularly, we assume that there are no artificial rules on the two accounts and the retiree can freely buy and redeem the tontine shares while bearing the transaction costs. Each transaction between the two accounts incurs fixed and proportional transaction costs depending on the transaction volume (cf. \cite{AS2017} and \cite{BC(2019)}).


The retiree chooses optimal consumption and transaction policies to maximize the expected utilities of the consumption and the bequest. We expect the optimal policies under the new  settings in this paper to meet the following rationalities. First, the preference for tontine account depends on the risk attitude of the retiree. An important finding in \cite{BD2019} is that when the retiree becomes less risk averse, she/he would  increase the  tontine proportion to almost $100\%$. Second,  when the retiree is relatively young, the proportion in tontine account should decrease gradually to smooth the wealth because the longevity credit increases with age. The last, the preference for tontine account should happen both when the retiree is relatively young and extremely old. When the force of mortality rate is extremely large, the retiree will gamble for the longevity credit. In  \cite{D2021}, the author introduces an artificial boundary value $0$ at the infinite time for the proportion of tontine. The condition leads to the results that the retiree puts little wealth in the tontine account when she/he is extremely old. This does not quite match reality. We would like to see what will happen if we relax this condition under the new settings.

According to the special settings of transaction costs and random death time, the above optimal control problem is formed into a combined stochastic and impulse control problem with infinite time horizon. To our best of knowledge,  it has not been discussed in the literature. To solve the problem, we first prove a weak dynamic programming principle (WDPP) for the value function.  Using the WDPP, we characterize the value function as a viscosity solution of a HJBQVI. Then, a comparison principle is given, which guarantees  the uniqueness of the viscosity solution of the HJBQVI in some function classes. It should be noted that the existence of the time-dependent force of mortality rate leads to the difficulties in the proof and the three main theorems require special treatments. For the proof of WDPP, different from \cite{AS2017}, we need to discretize and define a different countable cover. To prove that the value function is the viscosity subsolution of the HJBQVI, we extend the technique in \cite{VMP(2007)} and \cite{Seydel2009}. Besides, a  technique based on \cite{AS2017} and some estimations are needed to get the contradiction. Because our problem has infinite  time horizon, the operator from the stochastic control of the consumption causes difficulties in the comparison theorem.  We solve the difficulties by using  concavity methods and doing some estimations for the first order partial derivatives coming from the Ishii's lemma.

In order to accurately display the optimal consumption policies and transaction regions, we add some artificial boundary conditions and use the numerical algorithm proposed in \cite{AFP2016} to obtain the numerical solution of the HJBQVI in a bounded area. We have two main findings: First, like in \cite{OS2002} and \cite{BS2021}, the transaction region is approximately V-shaped at every age. Second, the transaction region exhibits two stages with respect to age. In the former stage, the retiree decreases the proportion of wealth in the tontine account gradually to smooth the wealth. In the latter stage,  the retiree  increases the proportion of wealth in the tontine account to gamble for the considerable longevity credits. Interestingly, in the latter stage,  the transaction policy is  insensitive to the magnitude of the bequest motive, but sensitive to the risk aversion attitude and the force of mortality rate. These findings exactly meet the aforementioned rationalities. 

The main contribution of this paper is threefold: First, we propose a new modern tontine mechanism for flexible purchase and redemption, and formulate the problem as a combined stochastic and impulse control problem with infinite  time horizon which has not been studied in the literature. We then characterize the value function as the unique viscosity solution of a HJBQVI.  Second, we establish the optimal consumption rate and transaction regions numerically by adding appropriate boundary conditions and extending the numerical methods to solve the HJBQVI. The last, the transaction region is divided into two stages with respect to age. The preference for tontine account decreases in the former stage and increases in the latter stage. It is insensitive to the bequest motive, but sensitive to the risk aversion attitude in the latter stage. These results show that the rationalities mentioned above are achieved under the proposed mechanism.


The remainder of the paper is organized as follows. Section \ref{formulation} provides the mathematical formulation of the retiree's optimization problem with the  reconstruction of modern tontine mechanism. In Section \ref{math}, we establish three main theorems on WDPP, viscosity characterization and comparison theorem for the value function. The proofs of the theorems are given in  Appendix. Section \ref{numerical} exhibits the numerical results on the optimal consumption rate and transaction regions. Finally,  we conclude the paper in Section \ref{conclude}.

\vskip 15pt
\setcounter{equation}{0}
\section{{ {\bf Optimization Problem for Modern Tontine}}}\label{formulation}
\vskip 5pt
In this section, we first establish a new model for modern tontine with explicit longevity credit distribution mechanism and transaction costs. Then, we formulate the optimization problem of the retiree with the consideration of both the consumption and bequest utilities.

In order to consider the  motivation for bequest, two mechanisms have been provided in existing literature. In  \cite{BD2019} and \cite{D2021}, the authors assume the coexistence of a bequest account and a tontine account. Besides, the proportions of the investment and the consumption are the same in the two accounts. In \cite{CR2022}, the authors regard the death benefit as a control variable in the modified tontine. Basically, we extend the former model settings to be more flexible and practical. We believe that the artificial rules on the consumption and the allocation of wealth narrow the admissible domains and are hard to be practically implemented. For example, the retiree may want to withdraw some wealth from the tontine account for a liquidity reserve instead of immediate consumption. Besides, if the consumption is withdrawn from an account of risky investment, the redemption usually leads to a liquidity discount which has not been considered. Importantly, retirees want to purchase and redeem tontine shares based on their own considerations. Someone may want to purchase more tontine share to gamble for mortality credits when she/he is extremely old. Thus, artificial rules are not welcomed and the retirees need more flexible and practical mechanisms of the tontine.

In this paper, we suppose that the wealth of the retiree is allocated to the bequest account and the tontine account. In the tontine account, the wealth is invested in a risky asset and the dynamics follows the geometric Brownian motion. Meanwhile, the bequest account is regarded as a liquidity account. The wealth is reserved in a risk-free asset and the consumption is directly withdrawn from this account. Naturally, consumption withdrawn from a liquidity account will not cause any discount and it is a reasonable setting. If the wealth in the bequest account is insufficient, the retiree needs to transfer the wealth from the tontine account before consumption. If the retiree wants to hedge the longevity risk and earn mortality credits, she/he will allocate more wealth in the tontine account. Otherwise, she/he will allocate more wealth in the bequest account to obtain higher consumption and leave more heritage. Based on this reconstructed mechanism, the retiree can buy and redeem the tontine shares freely while bearing the transaction costs.

First, we consider the scenario without transaction costs and longevity credits. Let $(\Omega,\mathcal{F},\{\mathcal{F}_{t}\}_{t\in\lbrack0,\infty)},\mathbb{P})$
be a filtered complete probability space satisfying the usual conditions, and $\{W_t, t\geq 0\} $ is
a standard Brownian motion defined on
$(\Omega,\mathcal{F}, \{\mathcal{F}_{t}\}_{t\in\lbrack0,+\infty)},\mathbb{P})$. The dynamics of the two accounts satisfies the following stochastic differential equations:
\begin{equation*}
	\left \{
	\begin{array}{ll}
		dP_0(t)=r P_0(t)dt, \ \  P_0(0)=p_0,\nonumber\\\\
		dP_1(t)=\mu P_1(t)dt+\sigma P_1(t)dW_t,\ \  P_1(0)=p_1,
	\end{array}
	\right.
\end{equation*}
where $P_0(t)$ and $P_1(t)$ are the values of the bequest account and the
tontine account at time $t$, respectively. $r$ is the risk-free interest rate, i.e., the expected return of the bequest account. $\mu$ and $\sigma$ are the expected return and the
volatility of the tontine account, respectively.  $p_0$ and $p_1$ are constants.

Second, we study the scenario with the consideration of longevity credits and transaction costs. Inspired by the settings in \cite{BD2019}, the longevity credits are paid into the tontine account at the rate $\lambda(t)Y(t)$ at time $t$. The force of mortality of the retiree is represented by $\lambda(t)$ at time $t$. We assume that $\lambda(\cdot)$ is continuous, non-negative and  non-decreasing in $[0,\infty)$. As discussed above, the flexible purchase and redemption of longevity insurance product is usually carried out by charging the transaction fees. For the first reason, the redemption of tontine share reduces the number of participants and increases the payout volatility.  Thus, the surrender premium is usually charged (cf. \cite{Weinert2017}). For the second reason, when a retiree notices the deterioration (improvement) of health, she/he tends to redeem (purchase) the tontine share. In the extreme scenario, the retiree will redeem all the wealth in the tontine account immediately before death. To mitigate the moral hazard problem, a transaction fee is needed. For the last reason, as the tontine wealth is invested in the risky asset, it may lead to a certain discount when selling the index in response to redemption. For the above reasons, we assume that there are transaction costs.  Following the ideas of \cite{FOS2001},  \cite{AS2017} and \cite{BC(2019)}, we assume that the transaction costs are the same when purchasing and withdrawing tontine. Particularly, each transaction between the two accounts incurs a fixed cost $C_{\textup{min}}>0$ and a proportional cost $\xi |\Delta|$, where $\Delta$ is the transaction volume and $0<\xi<1$. A transaction policy is represented by a sequence $(\tau,\Delta)=\{(\tau_k,\Delta_k)\}_{k\in \mathbb{N}}$, where the transaction time  $\{\tau_k\}$ is an increasing sequence of $\{\mathcal{F}_{t},t\geq 0\}$-stopping times satisfying $\lim\limits_{k\to \infty} \tau_k =\infty , \mathbb{P} \, a.s.$ , and the $k^{th}$ trading volume $\Delta_k$ is  $\mathcal{F}_{\tau_k}$-measurable. The wealth in the two accounts is respectively denoted by $X(t)$ and $Y(t)$  at time $t$.
\begin{equation*}
	\left \{
	\begin{array}{ll}
		X(s)=x+\int\limits_t^s[rX(u)-c(u)]\mathrm{d}u\!-\!\sum\limits_{k=1}^{+\infty}(\Delta_k+\xi |\Delta_k|+C_{\textup{min}})1_{\tau_k\leq s},\vspace{1ex}\\
		Y(s)=y+\int\limits_t^s(\mu+\lambda(u)) Y(u)\mathrm{d}u+\int\limits_t^s\sigma  Y(u)\mathrm{d}W(u)+\sum\limits_{k=1}^{+\infty}\Delta_k1_{\tau_k\leq s},
	\end{array}
	\right.
\end{equation*}
where $c(t)\geq 0$ represents the consumption rate at time $t$. Besides, we require the wealth in the two accounts be non-negative. Thus,
a policy $\nu=(c^{\nu},\tau,\Delta)$ is called admissible if
\begin{equation*}
	\tau_1 \geq t \,\ \text{and} \, \ 	(X^{\nu}(s),Y^{\nu}(s))\in \mathbb{R}_+^2, \,\mathbb{P}-a.s. \, , \forall s\in[t,\infty).
\end{equation*}
Let $\mathcal{A}(t,x,y)$ denote the set of all admissible polices with the initial state $(t,x,y)\in \mathbb{R}_+^3$. The objective of the retiree is to maximize the expected utilities of the consumption and the bequest. Then, similar to  \cite{BD2019}  and \cite{D2021}, the value function can be defined as follows:
\begin{equation}
	V(t,x,y)=\!\!\!\sup_{\nu\in\mathcal{A}(t,x,y) }\!\!\!\mathbb{E}\left.\left[\!\int\limits_t^{\tau_D}\mathrm{e}^{-\rho(s-t)}U(c^{\nu}(s)) \mathrm{d}s\!+\!b\mathrm{e}^{-\rho(\tau_D-t)}U(X^{\nu}({\tau_D}))\right| \tau_{D}>t\!\right],\label{v111}
\end{equation}
where $U(z)=\frac{z^p}{p}(0<p<1)$ is the CRRA (Constant Relative Risk Aversion) utility function.  $b>0$ is a constant and it is a weight parameter measuring the importance of the bequest utility.
$\tau_D$ is the death time of the retiree which is independent of $\{\mathcal{F}_t\}_{t\geq 0}$ and satisfies $\mathbb{P}(\tau_D>t)=\mathrm{e}^{-\int_0^{t}\lambda(s)\mathrm{d}s}$, we can then rewrite $(\ref{v111})$ as
\begin{equation}
	V(t,x,y)=\sup_{\nu\in\mathcal{A}(t,x,y) }\mathbb{E}\!\left[\int\limits_t^{+\infty}\!\!\!\mathrm{e}^{-\int_t^s(\lambda(u)+\rho)\mathrm{d}u}\left\{U(c^{\nu}(s))\!+\!b\lambda(s)U(X^{\nu}(s))\right\}\mathrm{d}s\right].\label{vv}
\end{equation}
For  simplicity, we define
\begin{equation*}
	D(t,s)=\mathrm{e}^{-\int_t^s(\lambda(u)+\rho)\mathrm{d}u},\, \bar{U}(s,c,x)=U(c)+b\lambda(s)U(x),
\end{equation*}
then Eq. \eqref{vv} becomes
\begin{equation}
	V(t,x,y)=	\sup_{\nu\in\mathcal{A}(t,x,y) }\mathbb{E}\!\left[\int\limits_t^{+\infty}D(t,s)\bar{U}(s,c^{\nu}(s),X^{\nu}(s))\mathrm{d}s\right]. \label{vf}
\end{equation}

\begin{remark}
	Similar to Lemma A.5 in \cite{BC(2019)}, we can prove that for each $T>t\geq0$, there is a constant $C(T)$ such that
	\begin{equation}
		\sup_{\nu\in\mathcal{A}(t,x,y)}\mathbb{E}\left[\sup_{t\leq s\leq T}(X^{\nu}(s)+Y^{\nu}(s))^2\right]\leq C(T)(1+x+y)^2.\label{ineq}
	\end{equation}
\end{remark}
\vskip 10pt
\setcounter{equation}{0}
\section{{ {\bf Solution of the Optimization Problem}}}\label{math}
\vskip 5pt
In this section, we solve the  combined stochastic and impulse control problem with infinite time  horizon. First,  we prove a WDPP for the value function \eqref{vf}.  Second, using the WDPP, we characterize the value function as a viscosity solution of a HJBQVI. The last, a comparison principle is given, which characterizes the value function as the unique continuous   viscosity solution of the HJBQVI.

\subsection{Upper Bound of the Value Function}
Because there are fixed and proportional transaction costs between the two accounts, we cannot prove the smoothness of the value function\footnote{ \cite{SS(1994)} gets the smoothness of the value function when there are only proportional transaction costs.}. Thus, we characterize the value function as the viscosity solution of a HJBQVI, and we start by giving an upper bound of the value function. Note that if $\varphi\in C^{1,2}(R^3_+)$, then for any finite $\{\mathcal{F}_{z},z\geq 0\}$-stopping time $\gamma\geq t\geq 0$ and $\nu=(c^{\nu},\tau,\Delta)\in \mathcal{A}(t,x,y)$, using It\^{o}'s formula for jump process (cf. \cite{JS2013} ), we get
\begin{align*}
	&D(t,\gamma)\varphi(\gamma,X^{\nu}(\gamma),Y^{\nu}(\gamma))\nonumber-\varphi(t,x,y)\\&=\int_t^{\gamma}D(t,s)\left[-\mathcal{L}[\varphi](s,X^{\nu}(s),Y^{\nu}(s))-c^{\nu}(s)\varphi_x\right]\mathrm{d}s
	+\int_t^{\gamma}\sigma Y^{\nu}(s)D(t,s)\varphi_y\mathrm{d}W(s)\\&\phantom{ee}+\sum_{k=1}^{+\infty}D(t,\tau_k)\left[\varphi(\tau_k,X^{\nu}(\tau_k-)\!-\!C_{\text{min}}\!-\!\xi|\Delta_k|\!-\!\Delta_k,Y^{\nu}(\tau_k-)\!+\!\Delta_k)\right.\\&\phantom{eeeeeeeeeeeeeeeeeeeeeeeeeeeeeeeeeeiee}\left.-\varphi(\tau_k,X^{\nu}(\tau_k-),Y^{\nu}(\tau_k-))\right]1_{\tau_k\leq \gamma},
\end{align*}
where
\begin{equation*}
	\mathcal{L}[\varphi](t,x,y)=(\lambda(t)+\rho)\varphi-[\varphi_t+rx\varphi_x+(\lambda(t)+\mu)y\varphi_y+\frac 12\sigma^2y^2\varphi_{yy}].
\end{equation*}
Combining the WDPP of the value function detailed in the next subsection, we establish the following HJBQVI that the value function should satisfy in some sense
\begin{equation}
	\min\left\{\mathcal{L}[\varphi](t,x,y)\!-\!b\lambda(t)U(x)\!-\!\sup_{c\geq 0}\{U(c)\!-\!c\varphi_x(t,x,y)\},\varphi(t,x,y)\!-\!\mathcal{M}[\varphi](t,x,y)\right\}\!\!=\!0,\label{HJBQVI}
\end{equation}
where
\begin{equation}
	\mathcal{M}[\varphi](t,x,y)=\begin{cases}\sup\limits_{\Delta\in S(x,y)}\varphi(t,x-\Delta-C_{\text{min}}-\xi|\Delta|,y+\Delta),&S(x,y)\neq \emptyset,\\
		-\infty,&S(x,y)= \emptyset,
	\end{cases}
	\label{operator}
\end{equation}
and
\begin{equation*}
	S(x,y)=\{\Delta\in \mathbb{R}|(x-\Delta-C_{\text{min}}-\xi|\Delta|,y+\Delta)\in \mathbb{R}^2_+\}.
\end{equation*}
For  simplicity, we define
\begin{equation*}
	S_{\emptyset}=\{(x,y)\in\mathbb{R}^2_+|S(x,y)=\emptyset\}=\{(x,y)\in\mathbb{R}^2_+|x+(1-\xi)y<C_{\text{min}}\},
\end{equation*}
and a convex function
\begin{align*}
	f(x)=\sup_{c\geq 0}\left\{U(c)-cx\right\}=\begin{cases}
		\frac{1-p}{p}x^{\frac{p}{p-1}},&x>0,\\
		+\infty, &x\leq 0.
	\end{cases}
\end{align*}
Inspired by \cite{AMS1996} and \cite{BS2021}, we start with constructing a classical supersolution.
\begin{lemma}\label{supsolution}
	For any $p\leq q<1$ such that
	\begin{equation}
		\frac{\rho-qr}{1-q}-\frac{q}{2}\frac{(\mu-r)^2}{\sigma^2}\frac{1}{(1-q)^2}>0,\label{condition2}	
	\end{equation}
	we define
	\begin{equation*}
		\Psi(t,x,y) \triangleq C(1+x+y)^q.
	\end{equation*}
	Then, when $C>0$ is large enough,  there is  a positive and continuous function $K(x,y)$ on $(\mathbb{R}^2_+\setminus(0,0))$ such that
	\begin{equation*}		\min\{\mathcal{L}[\Psi
		](t,x,y)-b\lambda(t)U(x)-f(\Psi_x(t,x,y)),\Psi(t,x,y)-\mathcal{M}[\Psi](t,x,y)\}\geq K(x,y)
	\end{equation*}
	for all $(t,x,y)\in \mathbb{R}_+\times(\mathbb{R}^2_+\setminus(0,0))$.
\end{lemma}
\begin{proof}
	Computation shows
	\begin{align*}
		&\mathcal{L}[\Psi
		](t,x,y)\\
		&=C(x+y+1)^{q-1}\!\left[(\lambda(t)+\rho)(x+y+1)-q(rx+\mu y)-\lambda(t)qy-\frac{q(q-1)\sigma^2y^2}{2(x+y+1)}\right]\\
		&\geq C(x+y+1)^{q-1}\!\left[\rho(x+y+1)\!+\!(1\!-\!q)\lambda(t)(x+y+1)\!-\!q(rx+\mu y)\!-\!\frac{q(q\!-\!1)\sigma^2y^2}{2(x+y+1)}\right]\\
		&\geq C(x+y+1)^{q-2}\!\left[(\rho\!-\!qr\!+\!(1\!-\!q)\lambda(t))\!(x\!+\!y\!+\!1)^2\!\!-\!q(\mu\!-\!r) y(x\!+\!y\!+\!1)\!-\!\!\frac{q(q\!-\!1)\sigma^2y^2}{2}\!\right]\\
		&\geq C(x+y+1)^{q-2}\left[\eta (x+y+1)^2+(1-q)\lambda(t)(x+y+1)^2\right]\\
		&=C\eta(x+y+1)^q+(1-q)\lambda(t)C(x+y+1)^q,
	\end{align*}
	where $\eta>0$ satisfies
	\begin{equation*}
		\frac{\rho-qr-\eta}{1-q}-\frac{q}{2}\frac{(\mu-r)^2}{\sigma^2}\frac{1}{(1-q)^2}\geq 0.
	\end{equation*}
	We also note that
	\begin{equation*}
		f(\Psi_x(t,x,y))=\frac{1-p}{p}C^{\frac{p}{p-1}}q^{\frac{p}{p-1}}(x+y+1)^{\frac{1-q}{1-p}p}\leq \frac{1-p}{p}C^{\frac{p}{p-1}}q^{\frac{p}{p-1}}(x+y+1)^{p}.
	\end{equation*}
	Thus,
	\begin{align*}
		&\mathcal{L}[\Psi
		](t,x,y)-b\lambda(t)U(x)-f(\Psi_x(t,x,y))\\
		&\geq \left[C\eta+C(1-q)\lambda(t)-\frac{b}{p}\lambda(t)-\frac{1-p}{p}C^{\frac{p}{p-1}}q^{\frac{p}{p-1}}\right](x+y+1)^q.
	\end{align*}
	Since $\lim\limits_{C\to\infty}C^{\frac{p}{p-1}}=0$, there is a $\tilde{C}>0$ such that
	\begin{equation*}
		\mathcal{L}[\Psi
		](t,x,y)-b\lambda(t)U(x)-\sup_{c\geq 0}\{U(c)-c\Psi_x(t,x,y)\}
		\geq \tilde{C}(x+y+1)^q.
	\end{equation*}
	Moreover, when $(x,y)\notin S_{\emptyset}$, $x+(1-\xi)y\geq C_{\text{min}}$, as such
	\begin{align*}
		\Psi(t,x,y)-\mathcal{M}[\Psi](t,x,y)&=C\inf_{\Delta\in S(x,y)}[(1+x+y)^q-(1+x+y-\xi|\Delta|-C_{\text{min}})^q]\\
		&\geq C[(1+x+y)^q-(1+x+y-C_{\text{min}})^q].
	\end{align*}
	When $(x,y)\in S_{\emptyset}$,
	\begin{equation*}
		\Psi(t,x,y)-\mathcal{M}[\Psi](t,x,y)=+\infty.	
	\end{equation*}
	This completes the proof.
\end{proof}	
\begin{remark}
	It is easy to verify that if \eqref{condition2} holds for $q=p$, then  $\Psi=C(x+y)^p$  also satisfies
	\begin{equation*}		\min\{\mathcal{L}[\Psi
		](t,x,y)-b\lambda(t)U(x)-f(\Psi_x(t,x,y)),\Psi(t,x,y)-\mathcal{M}[\Psi](t,x,y)\}\geq \tilde{K}(x,y),	
	\end{equation*}
	for some positive and continuous function $\tilde{K}$ on $\mathbb{R}^2_+\setminus(0,0)$  and	for all $(t,x,y)\in \mathbb{R}_+\times(\mathbb{R}^2_+\setminus(0,0))$. In the following of this paper, we assume that \eqref{condition2} always holds for $q=p$. The condition is commonly used in the literature, see \cite{M1971}, \cite{AMS1996}, \cite{FOS2001} and  \\   \cite{AS2017} for the discussions.
\end{remark}
\begin{theorem}
	Suppose that $0\leq \Psi \in C^{1,2}(\mathbb{R}_+\times(\mathbb{R}^2_+\setminus(0,0)))$ satisfies
	\begin{equation*}		\min\left\{\mathcal{L}[\Psi
		](t,x,y)-b\lambda(t)U(x)-f(\Psi_x(t,x,y)),\Psi(t,x,y)-\mathcal{M}[\Psi](t,x,y)\right\}\geq 0,	
	\end{equation*}
	for all  $(t,x,y)\in \mathbb{R}_+\times(\mathbb{R}^2_+\setminus(0,0))$. Then $V(t,x,y)\leq \Psi(t,x,y)$ for all $(t,x,y)\in \mathbb{R}_+\times(\mathbb{R}^2_+\setminus(0,0))$. Particularly, $V(t,x,y)\leq C(x+y)^p$ for some
	$C>0$.
\end{theorem}
\begin{proof}
	The proof is similar to the proof of Proposition 5.1 in \cite{SS(1994)}, and we omit it here.
\end{proof}
\subsection{The Weak Dynamic Programming Principle for the Value Function}
Proof of classical  dynamic programming principle (cf. \cite{KT2003}) needs the deep measurable selection theorem in probability theory. To avoid discussing a lot of measurability issues, we follow the techniques developed in \cite{KK2004}, \cite{BT(2011)} and \cite{AS2017} to obtain a WDPP.  However,  as we  are dealing with a control problem with infinite time horizon, we cannot draw conclusions directly from the aforementioned three papers. The main difficulty is to prove \eqref{DPP2}. Thus, we discretize the stopping times first, and then use similar methods as  in \cite{AS2017} to prove that \eqref{DPP2} holds in the discrete case. Finally, we extend the result to the continuous case.
\begin{theorem}\label{WDPP}
	Suppose $(t,x,y)\in \mathbb{R}_+\times\mathbb{R}_+^2\setminus{(0,0)}$, and let $\{\theta^{\nu},\nu\in \mathcal{A}(t,x,y)\}$ be a family of finite stopping times with values in $[t,T)$ for some constant $T>t$\footnote{Combining  \eqref{ineq} with Theorem 3.1, the second row of \eqref{DPP2} is finite.
	}, then we have the following inequalities
	\begin{align}
		&V(t,x,y)\nonumber\\
		&\leq\!\!\!\!\!\! \sup_{\nu\in\mathcal{A}(t,x,y)}\!\!\!\!\!\!\mathbb{E}\left\{\int_t^{\theta^{\nu}}\!\!\!\!D(t,s)\bar{U}(s,c^{\nu}(s),X^{\nu}(s))\mathrm{d}s\!\!+\!\!D(t,\theta^{\nu})V^*(\theta^{\nu},X^{\nu}(\theta^{\nu}-),Y^{\nu}(\theta^{\nu}-))\right\},\label{DPP1}\\and\nonumber\\
		&V(t,x,y)\nonumber\\&\geq\!\!\!\! \sup_{\nu\in\mathcal{A}(t,x,y)}\!\!\!\!\!\mathbb{E}\left\{\int_t^{\theta^{\nu}}D(t,s)\bar{U}(s,c^{\nu}(s),X^{\nu}(s))\mathrm{d}s+D(t,\theta^{\nu})\varphi(\theta^{\nu},X^{\nu}(\theta^{\nu}),Y^{\nu}(\theta^{\nu}))\right\}.\label{DPP2}	
	\end{align}
	where $V^*$ \footnote{Henceforth, for any function $v$, $v^*$ and $v_*$ denote the upper semi-continuous envelope and lower semi-continuous envelope of $v$, respectively. } is the upper semi-continuous envelope of $V$, and $\varphi\leq V$ is a continuously differentiable
	function in $[t,+\infty)\times\mathbb{R}_+^2$.
	\begin{proof}
		See Appendix \ref{proof WDPP}.
	\end{proof}
\end{theorem}

\subsection{The Value Function as a Viscosity Solution of a HJBQVI}
Before introducing the definition of viscosity solution, we first give some properties of the value function and non-local operator $\mathcal{M}$ ( see \eqref{operator}).
\begin{proposition}\label{VM}
	For any locally bounded functions $\varphi,\phi_1,\phi_2:\mathbb{R}^3_+\rightarrow \mathbb{R}$, the following properties hold:
	
	1.  If $\varphi\in$ USC (upper semi-continuous), then $\mathcal{M}[\varphi]\in$ USC;
	
	2. If $\varphi\in$ LSC (lower semi-continuous), then $\mathcal{M}[\varphi]\in$ LSC in $\mathbb{R}_+\times (\mathbb{R}^2_+\setminus\bar{S_{\emptyset}})$;
	
	3. For fixed $(t,x)\in \mathbb{R}^2_+$, $V(t,x,\cdot)$ is non-decreasing. For fixed $(t,y)\in \mathbb{R}^2_+$, $V(t,\cdot,y)$ is non-decreasing;
	
	4. For fixed $(t,x)\in \mathbb{R}^2_+$, $V^*(V_*)(t,x,\cdot)$ is non-decreasing. For fixed $(t,y)\in \mathbb{R}^2_+$, $V^*(V_*)(t,\cdot,y)$ is non-decreasing;
	
	5. $V\geq \mathcal{M}[V]$;
	
	6. $\mathcal{M}[\phi_1+\phi_2]\leq\mathcal{M}[\phi_1]+\mathcal{M}[\phi_2]$;
	
	7. If $\phi_1\in$ LSC and $\phi_2\in$ USC, then $\mathcal{M}[\phi_1+\phi_2]_*\leq\mathcal{M}[\phi_1]_*+\mathcal{M}[\phi_2]$;
	
	8.  $(\mathcal{M}[V])^*\leq \mathcal{M}[V^*]$;
	
	9.  $V_*\geq \mathcal{M}[V_*]$.
\end{proposition}
\begin{proof}
	See \cite{BS2021} for the proof of Points 1 and 2. Points
	3, 5 and 6 are directly proved from the definition of $V$ and $\mathcal{M}$. Point 7 is directly proved from points 6 and 1. Thus, we only give the proof of Points 4, 8 and 9.
	
	The proof of Point 4.   Without loss of generality, we fix $(t,x,y)\in\mathbb{R}^3_+$ and $\delta>0$ to prove $V_*(t,x,y)\leq V_*(t,x+\delta,y)$. We choose $(t_k,x_k,y_k)\in\mathbb{R}^3_+ $ such that $\lim\limits_{k\to \infty}(t_k,x_k,y_k,V(t_k,x_k,y_k))=(t,x+\delta,y,V_*(t,x+\delta,y))$, then
	\begin{equation*}
		V_*(t,x+\delta,y)=\lim\limits_{k\to \infty}V(t_k,x_k,y_k)\geq\varliminf_{k\to\infty}V(t_k,|x_k-\delta|\wedge x_k,y_k)\geq V_*(t,x,y).
	\end{equation*}
	The proof of Point  8.  Because $ \mathcal{M}[V]\leq \mathcal{M}[V^*]$ and  $\mathcal{M}[V^*]$ is upper semi-continuous,  $(\mathcal{M}[V])^*\leq \mathcal{M}[V^*]$.

	The proof of Point  9. Because $V\geq \mathcal{M}[V]$, it suffices to show that $V_*(t,x,y)\geq \mathcal{M}[V_*](t,x,y)$ for $(x,y)\in \partial S_{\emptyset}$. For this, since $S(x,y)=\{-y\}$,  $\mathcal{M}[V_*](t,x,y)=V_*(t,0,0)\leq V_*(t,x,y)$.
\end{proof}

Now we introduce the definition of viscosity solution of $\eqref{HJBQVI}$, which is close to the definition in  \cite{I1993}, \cite{OS2002} and \cite{BS2021}.
\begin{Def}
	1. A function $v\in $ LSC is called a viscosity supersolution of $\eqref{HJBQVI}$ if, for all $\varphi\in C^{1,2}(\mathbb{R}_+\times (\mathbb{R}^2_+\setminus{(0,0)}))$ such that $\varphi-v$ attains  maximum $0$ at $(t,x,y)$, we have	
	\begin{align*}
		\min\left\{\mathcal{L}[\varphi](t,x,y)-b\lambda(t)U(x)-f(\varphi_x(t,x,y)),v(t,x,y)-\mathcal{M}[v]_*(t,x,y)\right\}\geq0,	
	\end{align*}
	when $x\neq0$ and
	\begin{equation*}
		\min\left\{\mathcal{L}[\varphi](t,x,y)-b\lambda(t)U(x),v(t,x,y)-\mathcal{M}[v]_*(t,x,y)\right\}\geq0,	
	\end{equation*}
	when $x=0$.
	
	2. A function $u\in $ USC is called a viscosity subsolution of $\eqref{HJBQVI}$ if, for all $\varphi\in C^{1,2}(\mathbb{R}_+\times (\mathbb{R}^2_+\setminus{(0,0)}))$ such that $\varphi-u$ attains minimum $0$ at $(t,x,y)$,  we have	
	\begin{equation*}
		\min\{\mathcal{L}[\varphi](t,x,y)-b\lambda(t)U(x)-f(\varphi_x(t,x,y)),u(t,x,y)-\mathcal{M}[u](t,x,y)\}\leq0.
	\end{equation*}
	Furthermore, a function $v$ is called a viscosity solution of \eqref{HJBQVI} if $v_*$ is a viscosity supersolution of \eqref{HJBQVI} and $v^*$  is a viscosity subsolution of \eqref{HJBQVI}.
	
\end{Def}
Next, we are going to prove that  the value function is a viscosity solution of \eqref{HJBQVI}. The difficulty is to prove that
$V^*$ is a viscosity subsolution. \cite{OS2002} provides a method using $0-1$ law of the first impulse. However, it does not work here because the  value function of this paper is related to time $t$. Thus, we follow the technique proposed in \cite{VMP(2007)}, which has been extended to a general setting with jumps by \cite{Seydel2009}. However, stochastic control is not considered in \cite{VMP(2007)}, and  \cite{Seydel2009}  assumes that the stochastic control is bounded and the jump process satisfies some conditions as ``uniform continuity''.  Particularly, as we explore an unbounded stochastic control, which results in no prior ``uniform continuity'',  the setting in this paper does not satisfy the assumptions in \cite{Seydel2009}.  To overcome the difficulties, we extend the technique in \cite{AS2017} to avoid using boundedness and make some estimates to obtain the ``uniform continuity''.
\begin{theorem}\label{vis char}
	The value function $V$ is a viscosity solution of $\eqref{HJBQVI}$ .
\end{theorem}
\begin{proof}
	See Appendix \ref{proof of vis char }.
\end{proof}
\subsection{Comparison Principle for the HJBQVI}

We do not prove a comparison principle for $\eqref{HJBQVI}$ directly.  Instead, we  consider an equivalent form of \eqref{HJBQVI}, which is more tractable. We start with the following lemma and omit its proof as it is easy to be proved.
\begin{lemma}
	$u(v)\in$ USC(LSC) is a viscosity  subsolution (supersolution) of \eqref{HJBQVI} if and only if $D(0,t)u(D(0,t)v)$ is a viscosity  subsolution (supersolution) of the following HJBQVI
	\begin{align}
		\min\{\mathcal{\tilde{L}}&[\phi](t,x,y)-bD(0,t)\lambda(t)U(x)-[D(0,t)]^{\frac{1}{1-p}}f(\phi_x(t,x,y)),\nonumber\\ &\phantom{eeeeeeeeeeeeeeeeeeeeeeeeee}\phi(t,x,y)-\mathcal{M}[\phi](t,x,y)\}=0,\label{HJBQVI2}
	\end{align}
	where
	\begin{equation*}
		\mathcal{\tilde{L}}[\phi](t,x,y)=-[\phi_t+rx\phi_x+(\lambda(t)+\mu)y\phi_y+\frac 12\sigma^2y^2\phi_{yy}].
	\end{equation*}
\end{lemma}
Then, we prove the comparison principle.  We introduce some notations and  define
\begin{align*}
	&F(t,z,q,p,X)=\\&-(q+rz_1p_1+(\lambda(t)+\mu)z_2p_2+\frac12\sigma^2z_2^2X_{22})-bD(0,t)\lambda(t)U(z_1)-[D(0,t)]^{\frac{1}{1-p}}f(p_1).
\end{align*}
Thus, $F:\mathbb{R}_+\times \mathbb{R}_+^2\times\mathbb{R}\times\mathbb{R}^2\times\mathbb{S}^2\rightarrow \mathbb{R}\cup\{-\infty\}$ is continuous in $(t,z,q,p,X)$ and concave in $p_1$. For $\psi\in C^{1,2}(S)$($S=\mathbb{R}_+\times(\mathbb{R}_+^2\setminus(0,0))$), we define
\begin{equation*}
	F[\psi](t,z)=F(t,z,\psi_t(t,z),\psi_z(t,z),\mathrm{D}^2_z\psi(t,z)).
\end{equation*}
Note that for the  definition of the viscosity supersolution for \eqref{HJBQVI}, the validity of the supersolution inequality depends on whether $x$ equals 0. This prevents us from obtaining the comparison theorem and the continuity of the value function over the entire region $S$. After some common preparations (see lemma 3.1 in  \cite{GIL(1992)} and step 1 of the proof of Theorem 3.4 in \cite{BS2021}),
the proof has been done by considering two cases: the interior points ($S\cap\{x>0\}$),  and the boundary points ($S\cap\{x=0\}$). We follow \cite{OS2002},  \cite{BC(2019)} and \cite{BS2021} to deal with the interior points, and  \cite{B1994}, \cite{AST2001}, \cite{OS2002} and  \cite{VMP(2007)} to deal with the boundary points. A common difficulty of the two cases comes from the term $[D(0,t)]^{\frac{1}{1-p}}f(\phi_x(t,x,y)) $ in \eqref{HJBQVI2}.  When we apply Ishii's lemma to derive some contradictions, this term always causes some troubles. Fortunately, the troubles are well solved by doing some careful estimates for the first order partial derivatives.

\begin{theorem}\label{compare}
	Let $u\in $ USC be a viscosity subsolution and $v\in$ LSC be a viscosity supersolution of \eqref{HJBQVI2}. Suppose  that $\lambda(\cdot)$ is locally Lipschitz continuous, $v$ is non-decreasing in $x$ and
	\begin{equation}
		u(t,0,0)=v(t,0,0)=0 \,\text{and}\,\,0\leq u(t,x,y),v(t,x,y)\leq KD(0,t)(1+x+y)^p,\label{uv}
	\end{equation}
	where $K>0$ is constant. Then
	\begin{equation*}
		u(t,x,y)\leq v(t,x,y) ,\quad \forall(t,x,y)\in\mathbb{R}_+\times(0,+\infty)\times\mathbb{R}_+.
	\end{equation*}
	In addition, if
	\begin{equation}
		v(t,0,y)=\varliminf_{x'>0,(t',x',y')\to (t,0,y)}v(t',x',y'), \quad \forall (t,y)\in\mathbb{R}^2_+,\label{v}
	\end{equation}
	then $v$ dominates $u$ everywhere.
\end{theorem}
\begin{proof}
	See Appendix \ref{proof of compare}.
\end{proof}

\begin{corollary}
	The value function is  continuous in $\mathbb{R}_+\times(0,+\infty)\times\mathbb{R}_+$. The viscosity solution of the HJBQVI which satisfies the conditions in Theorem  \ref{compare}  is unique.
\end{corollary}
\vskip 10pt
\setcounter{equation}{0}
\section{{ {\bf Numerical Analysis}}}\label{numerical}
\vskip 5pt
In this section, we first establish the  parameter settings according to practice and reasonable mechanism design. Then, we exhibit the behaviors of  the optimal transaction regions and consumption rate. The next, we study the impacts of the bequest motive, risk averse attitude, force of mortality rate and the transaction cost on the optimal control policies. The last, we study the utility improvement brought by participation in the tontine account.
\begin{figure}[]
	\centering
	\begin{subfigure}[t]{0.4\textwidth}
		\centering
		\includegraphics[width=1\textwidth]{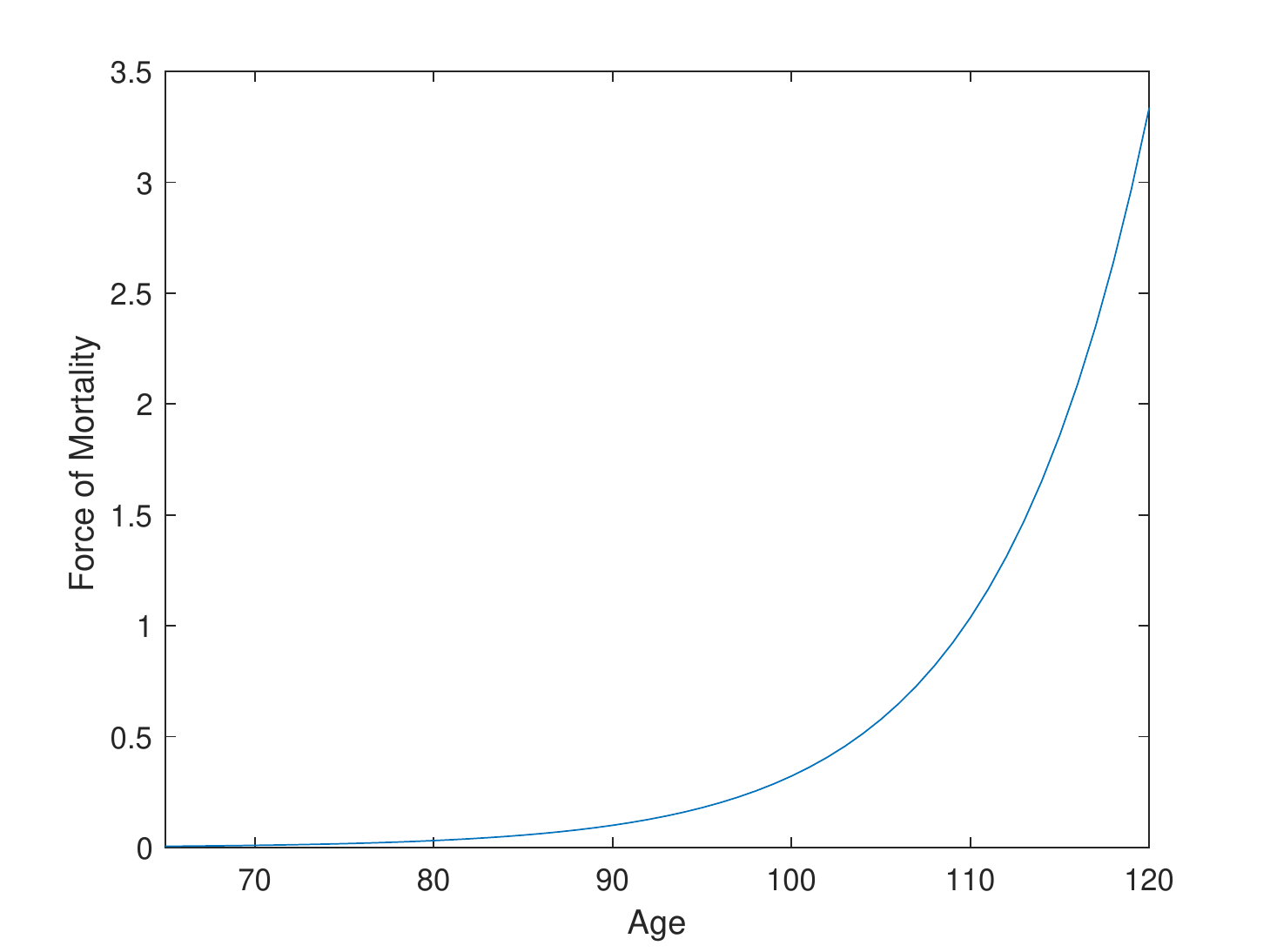}
	\end{subfigure}
	\begin{subfigure}[t]{0.4\textwidth}
		\centering
		\includegraphics[width=1\textwidth]{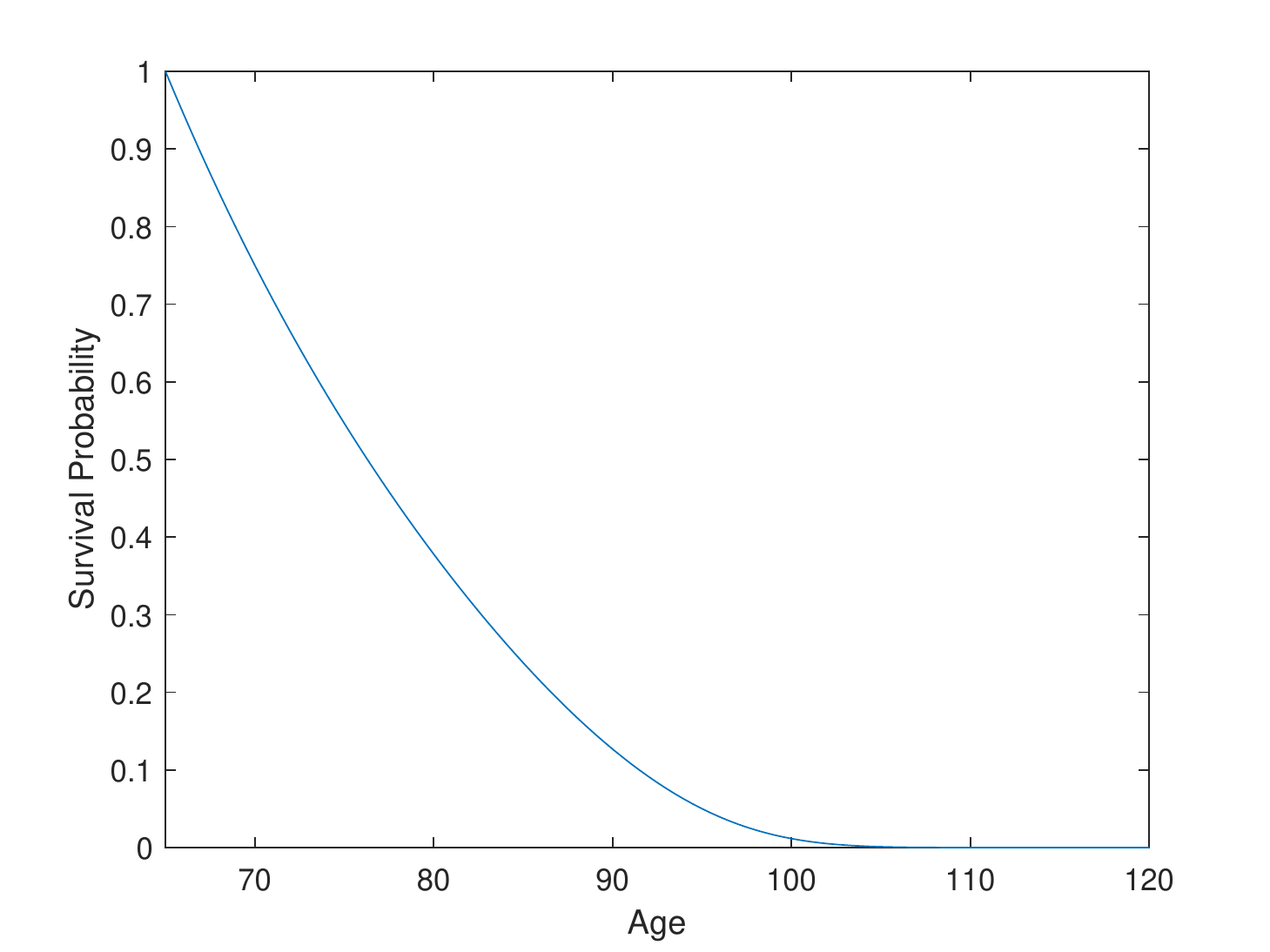}
	\end{subfigure}
	\caption{The force of mortality and survival  probability of the retiree}
	\label{mortality and survival}
\end{figure}

First, we explore Makeham's law to depict the force of mortality, that is,
\begin{equation*}
	\lambda(t)= A+B C^{65+t},
\end{equation*}
where $A=2.2\times10^{-4}$, $B=2.7\times 10^{-6}$ and $C=1.124$. The evolutions of the force of mortality rate and the survival probability of the retiree are shown in Fig. \ref{mortality and survival}. We can observe that the mortality rate rises rapidly, and the survival probability approaches $0$ around the age of $100$. This will lead to the gambling behaviors of the elderly retirees.
For the market parameters, we choose $r=\rho=0.05$, $\mu=0.085$ and $\sigma=0.25$. The magnitude of bequest motive and the relative risk aversion parameter are $b=3$ and $1-p=0.7$, respectively.   The above parameter settings are consistent with the ones in \cite{BD2019}. Because of the special longevity credit distribution mechanism of the tontine, the moral hazard problem is less serious than that in the life annuity. Accordingly, we suppose that the transaction cost is lower. The fixed and proportional cost parameters are $C_{\textup{min}}=0.5$ and $\xi=0.05$, respectively.

Using \eqref{uv}, we  introduce the artificial boundary conditions $V(T,\cdot,\cdot)=0$ into \eqref{HJBQVI2}. We choose $T=55$ because $D(0,55)= 2.6516\times 10^{-14}
$ is really small. The numerical results are then obtained by solving \eqref{HJBQVI2} and using a penalty method detailed in \cite{AFP2016}. Inspired by \cite{BS2021}, we perform computations on a triangular grid ($x+y\leq 2000$),  and the transaction regions and the consumption rate are reported on a square grid ($0\leq x,y\leq 900$). Moreover, the maximal consumption rate is assumed to be $1000$. For all  the  transaction regions in the later subsections, at time $t$, the square  grid is divided into two parts, the blue part and the other part. The blue part is the transaction region, which can be denoted by $S_1=\{(x,y)\in [0,900]^2|V(t,x,y)=\mathcal{M}[V](t,x,y)\}$, the other region is the No Transaction region, which can be denoted by $S_2= [0,900]^2\setminus S_1=\{(x,y) \in[0,900]^2|V(t,x,y)>\mathcal{M}[V](t,x,y)\}$.

\subsection{Behaviors of the Optimal Transaction Regions and Consumption Rate}
In Fig. \ref{base_region}, we study the optimal transaction regions for the $65$, $80$, $95$ and $110$ years old retiree. The results exhibit the V-shaped domain consists of three regions: Buy tontine region (the lower blue region), Sell tontine region (the upper blue region) and No Transaction region (the other region). Once the retiree's wealth level hits the Buy (Sell) region, a transaction from the bequest (tontine) account to the tontine (bequest) account occurs immediately and the new wealth level reaches the lower (upper) red line. The exact volume of the transaction is  the maximal point $\Delta^*$ of $\mathcal{M}$ in \eqref{operator}. No transaction is generated in the No Transaction region.
\begin{figure}[]
	\centering
	\begin{subfigure}[t]{0.35\textwidth}
		\centering
		\includegraphics[width=1\textwidth]{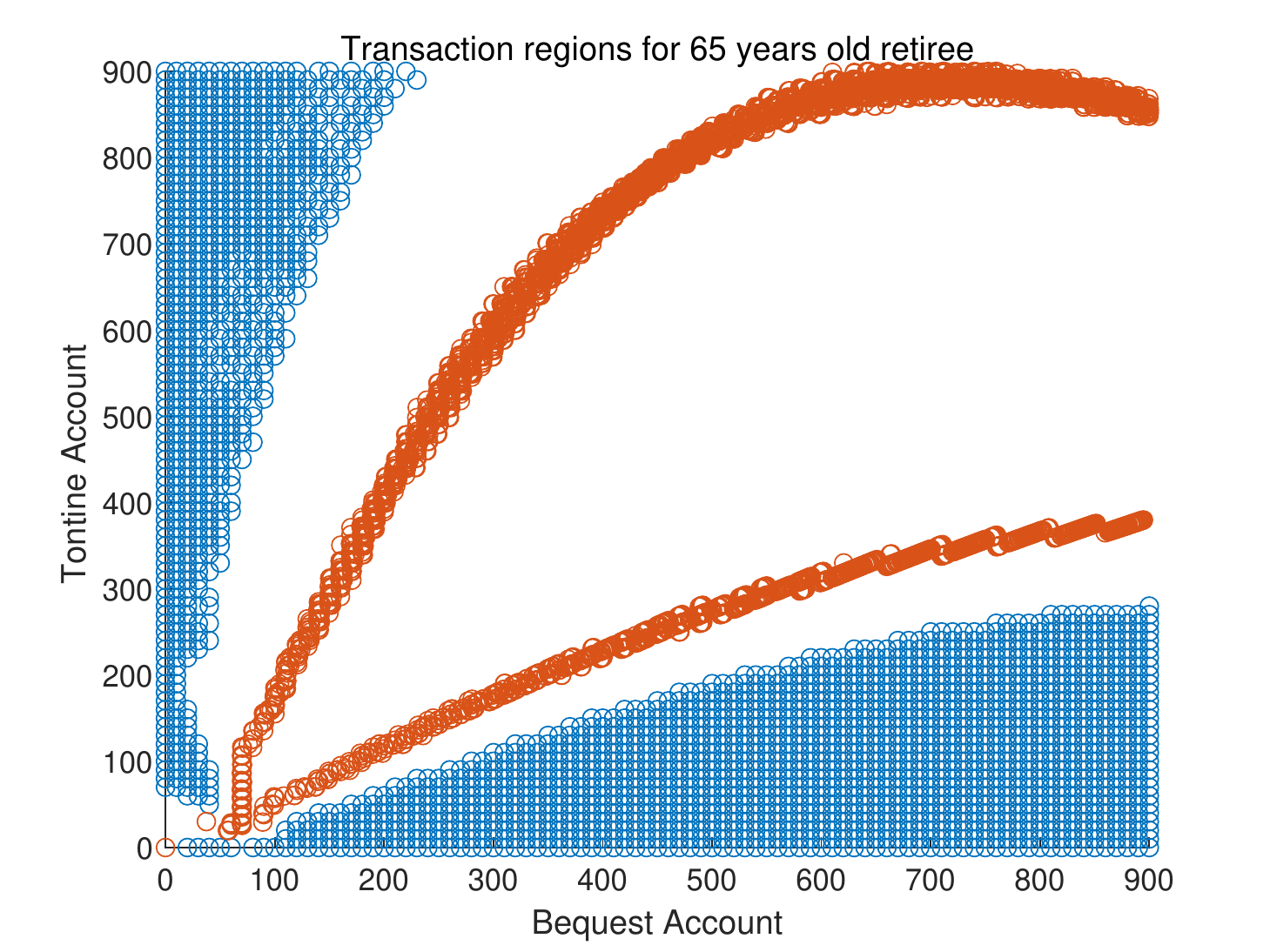}
	\end{subfigure}
	\begin{subfigure}[t]{0.35\textwidth}
		\centering
		\includegraphics[width=1\textwidth]{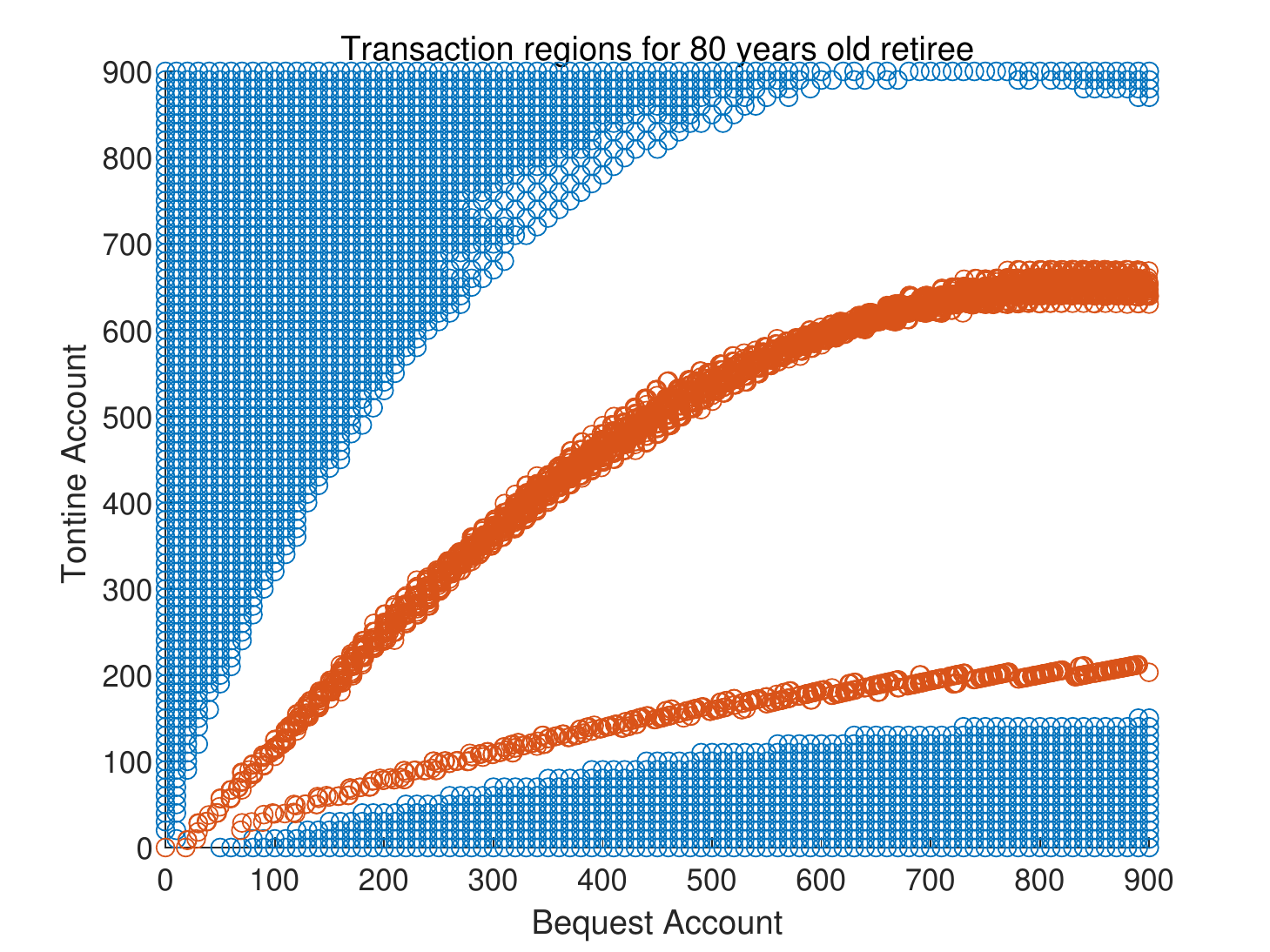}
	\end{subfigure}
	\begin{subfigure}[t]{0.35\textwidth}
		\centering
		\includegraphics[width=1\textwidth]{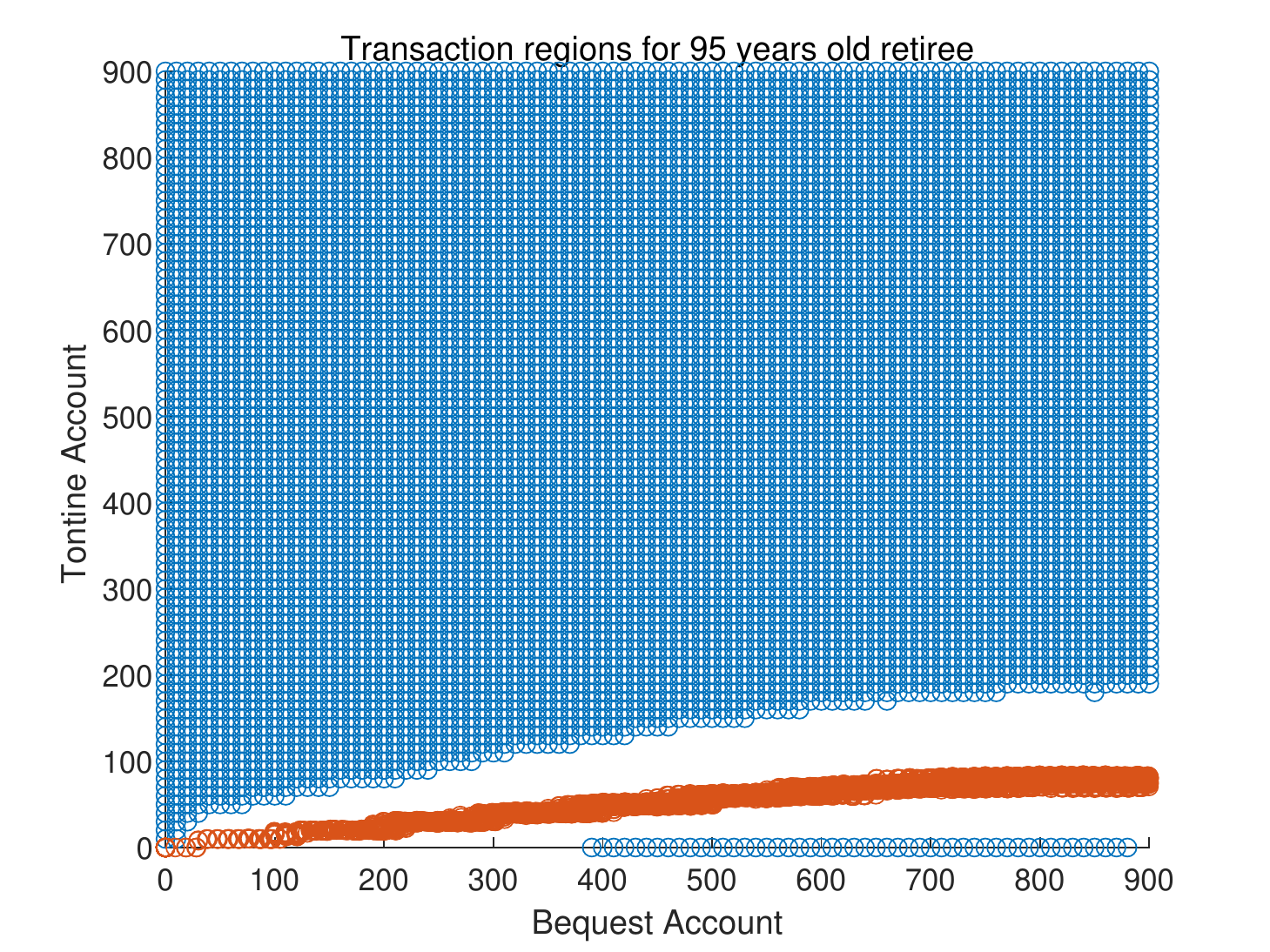}
	\end{subfigure}
	\begin{subfigure}[t]{0.35\textwidth}
		\centering
		\includegraphics[width=1\textwidth]{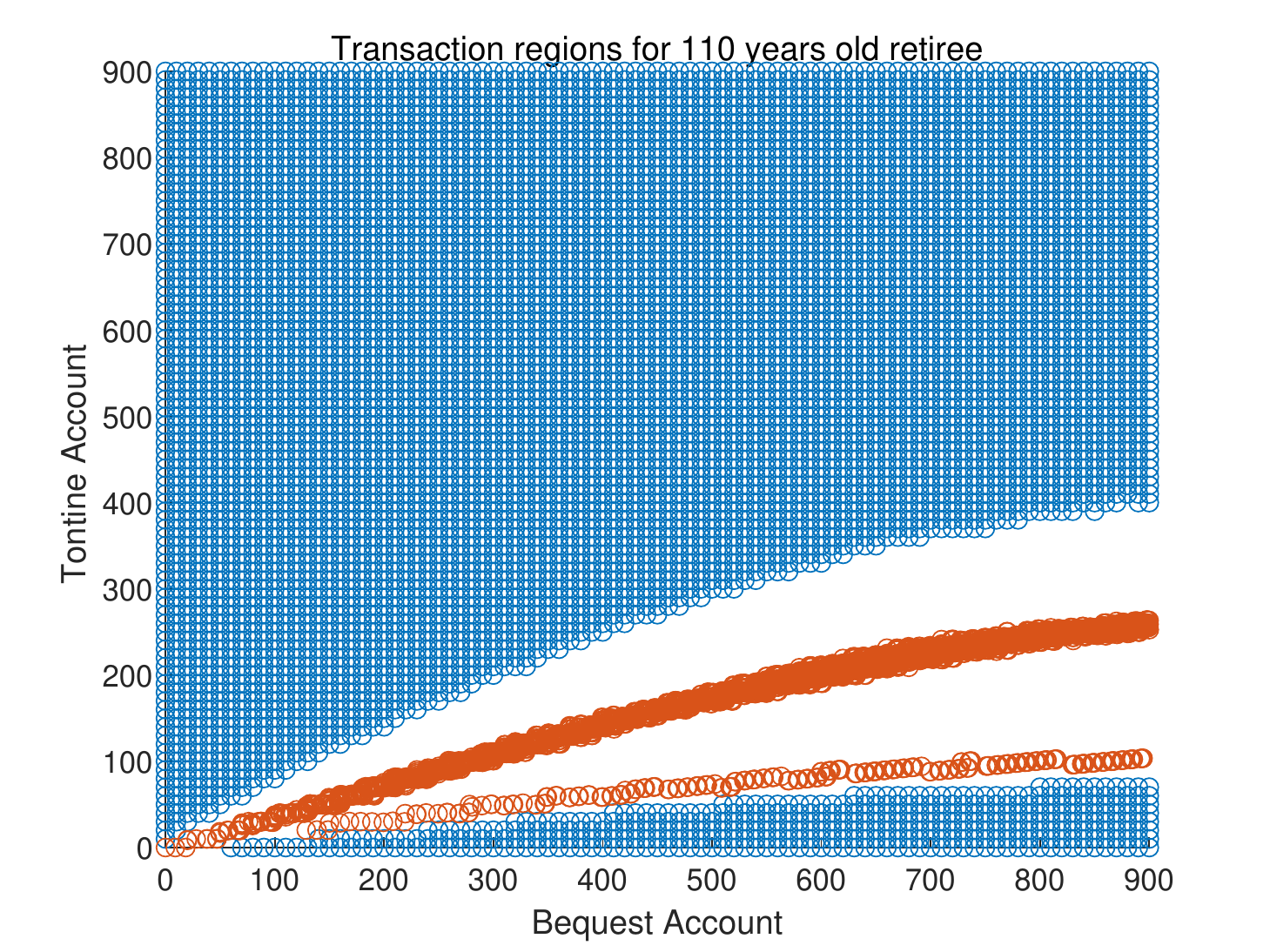}
	\end{subfigure}
	\caption{The optimal transaction regions}
	\label{base_region}
\end{figure}
\begin{figure}[]
	\centering
	\begin{subfigure}[t]{0.35\textwidth}
		\centering
		\includegraphics[width=1\textwidth]{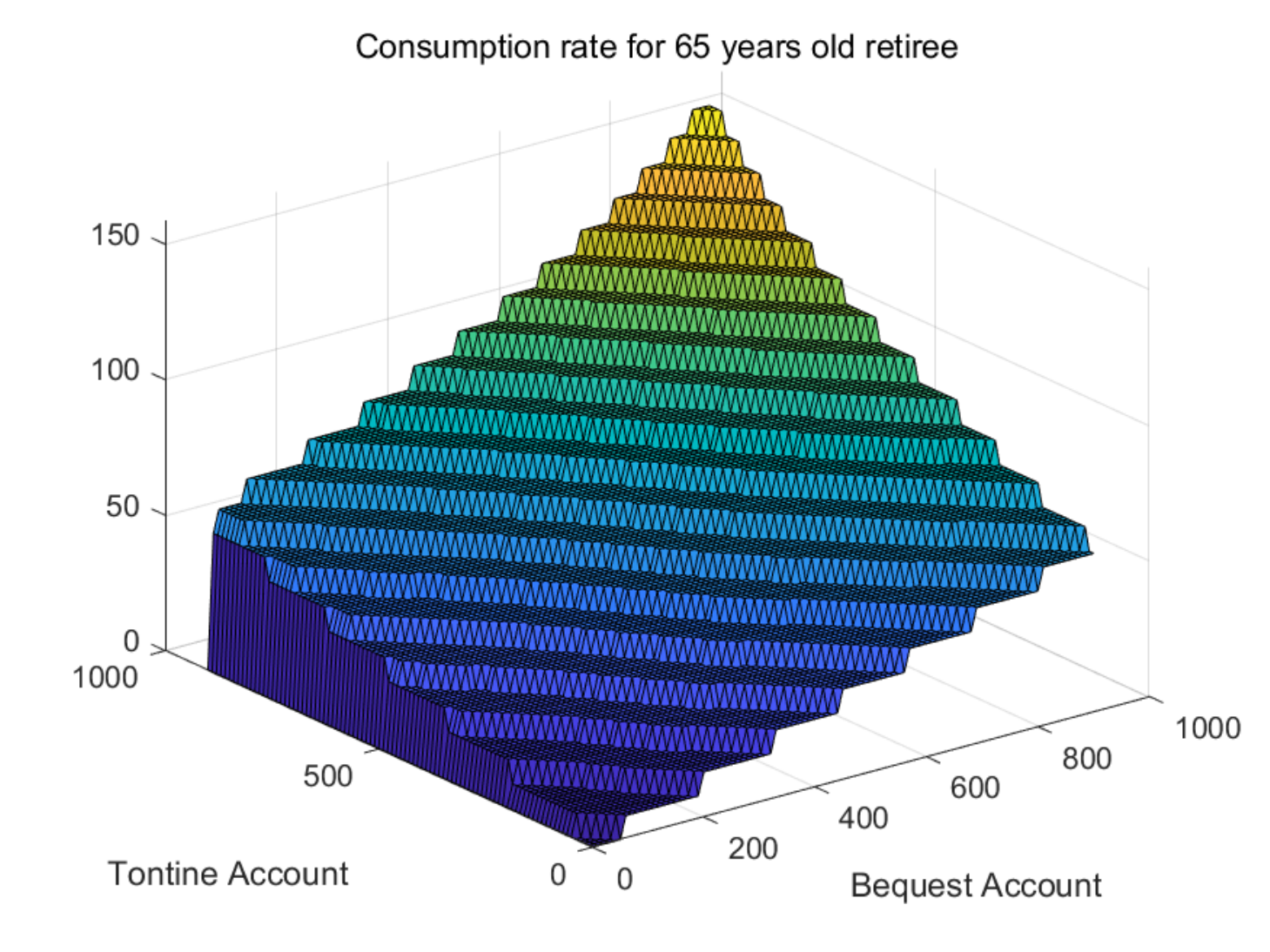}
	\end{subfigure}
	\begin{subfigure}[t]{0.35\textwidth}
		\centering
		\includegraphics[width=1\textwidth]{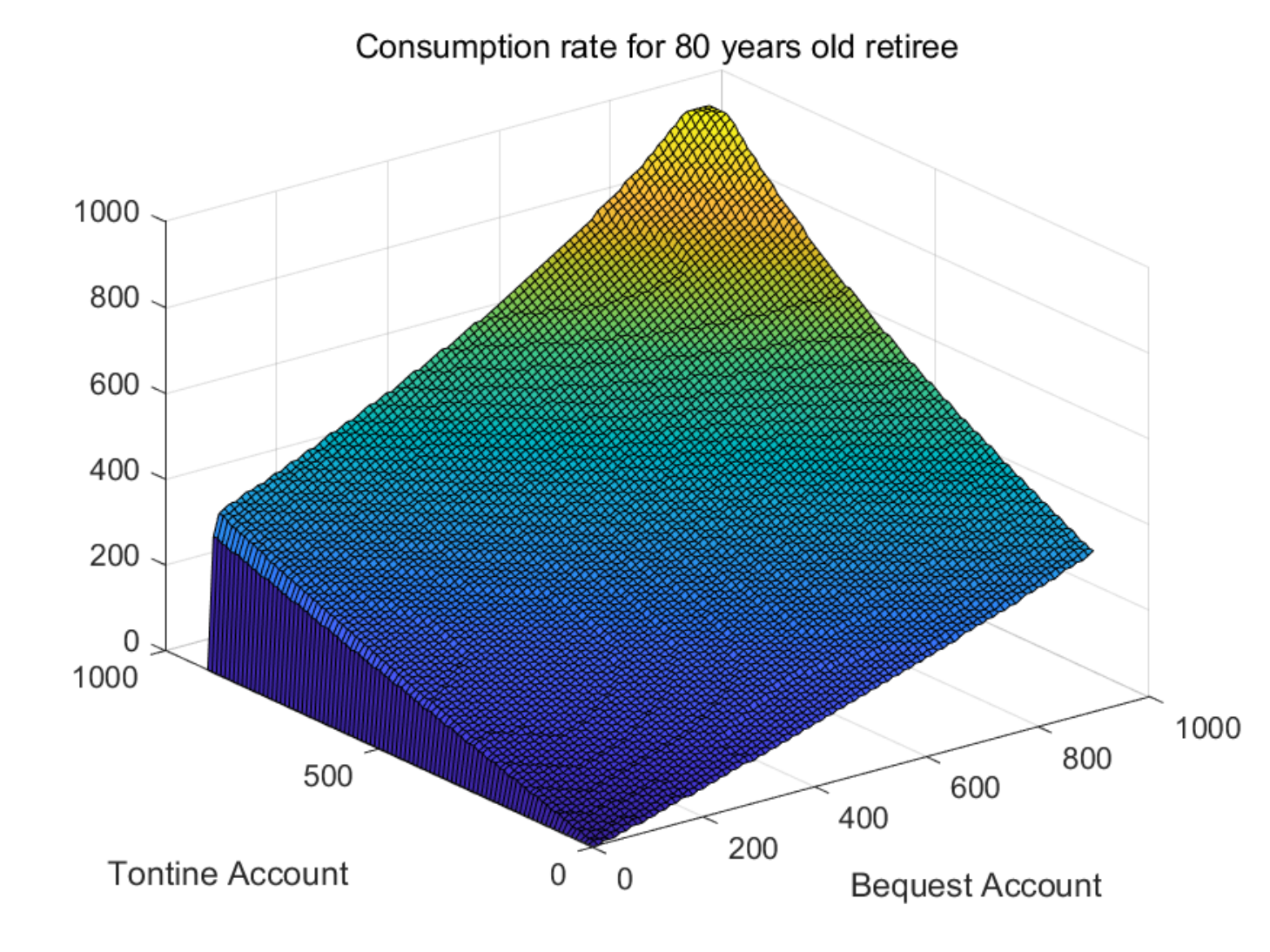}
	\end{subfigure}
	\begin{subfigure}[t]{0.35\textwidth}
		\centering
		\includegraphics[width=1\textwidth]{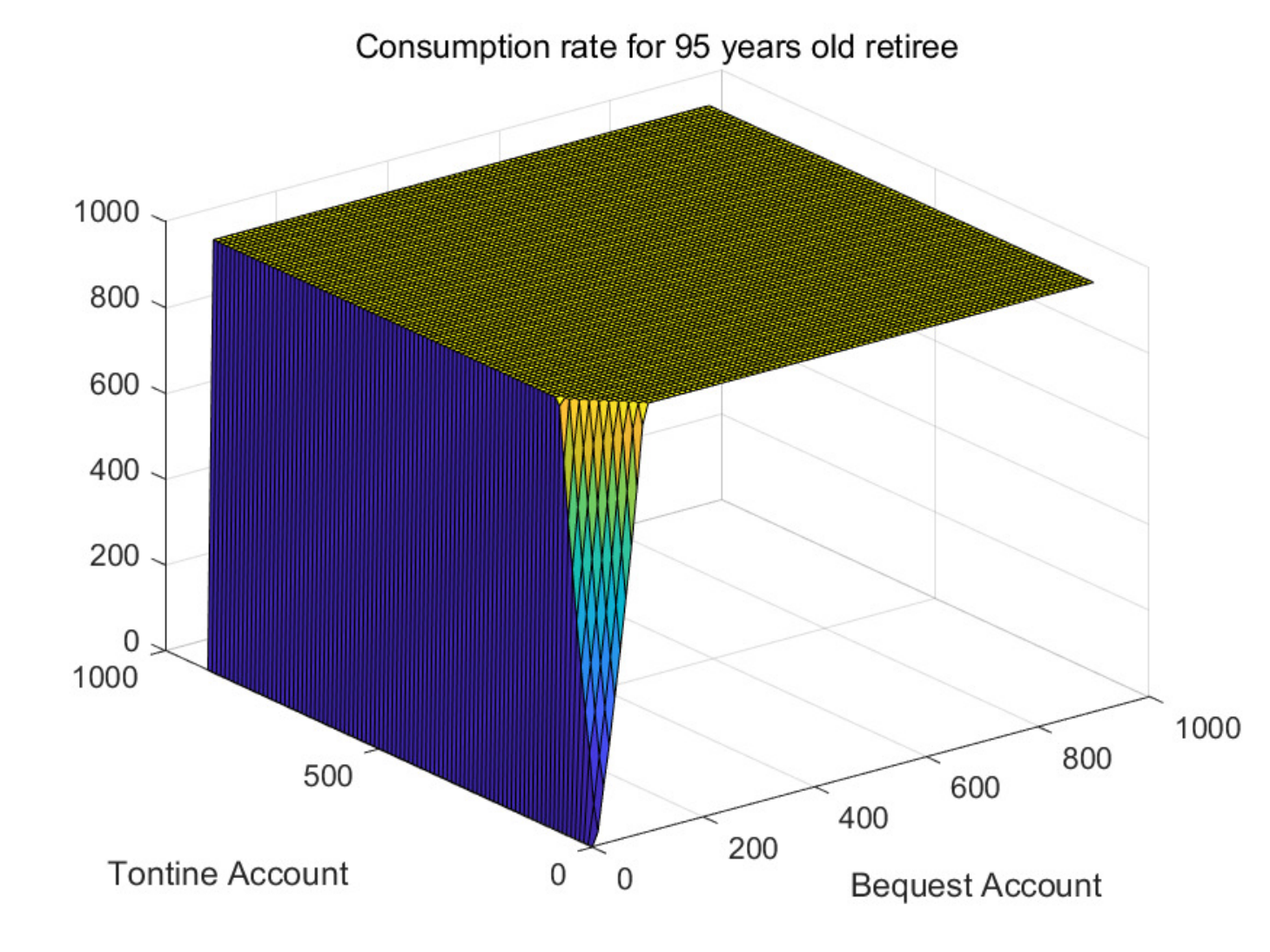}
	\end{subfigure}
	\begin{subfigure}[t]{0.35\textwidth}
		\centering
		\includegraphics[width=1\textwidth]{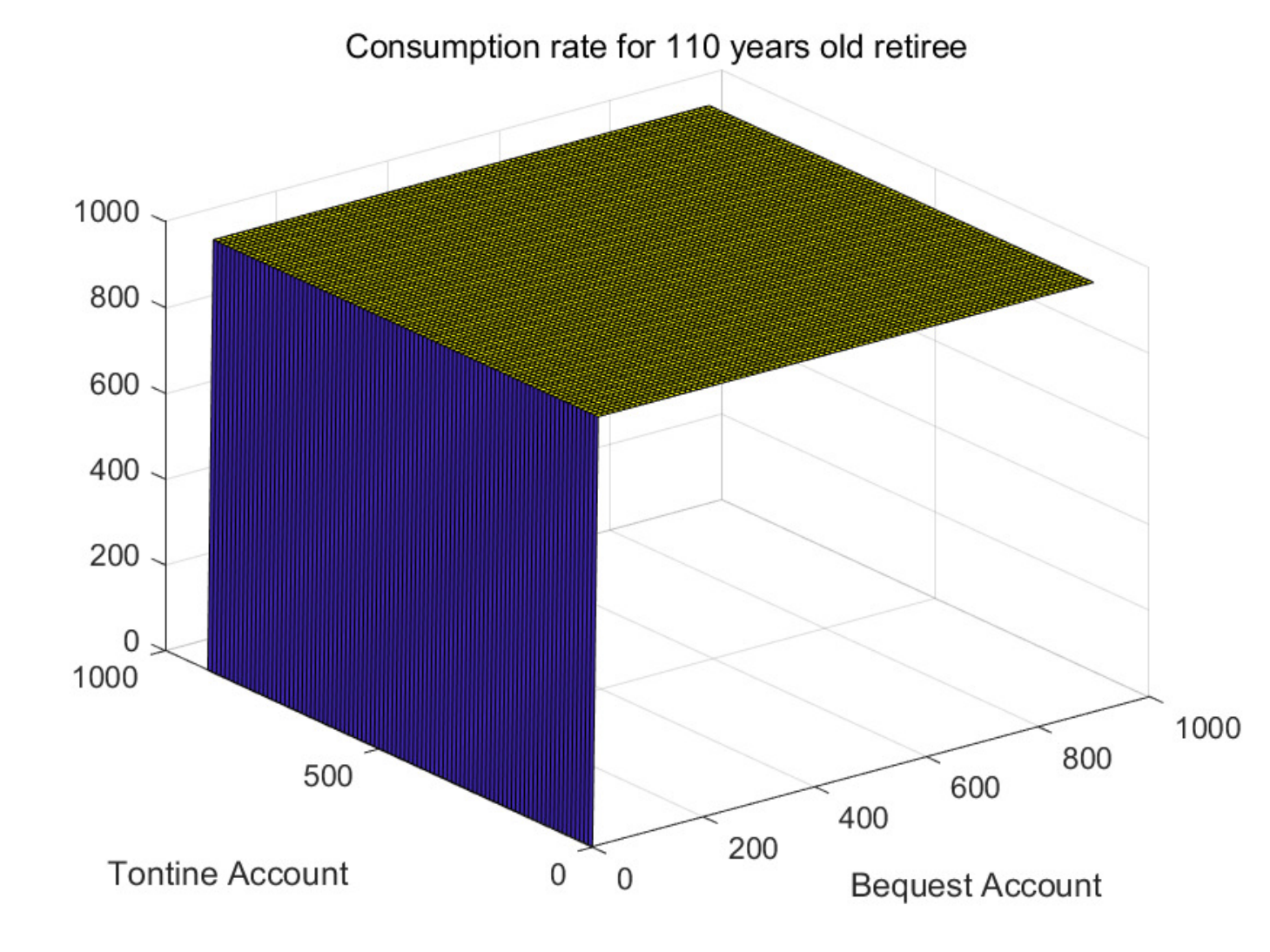}
	\end{subfigure}
	\caption{The optimal consumption rate}
	\label{base_c}
\end{figure}
Interestingly, the evolutions of the  transaction regions are divided into two stages. In the former stage (65-95 years old cases), the Sell region of tontine enlarges and the Buy region shrinks with respect to the retiree's age.  Thus,
the retiree prefers to decrease the wealth in the tontine account when she/he grows older. Because the force of mortality rate increases with age, the retiree can receive increasing longevity credit from the tontine account. Thus, reducing the tontine wealth gradually is an optimal choice for smoothing the wealth and achieving stability. In the latter stage (110 years old case), the Sell region of tontine shrinks and the Buy region enlarges in some extent. In this stage, because the force of mortality rate becomes extremely large, the longevity credit is so attractive. Such that, the retiree prefers to allocate more wealth in the tontine account to gamble for the huge longevity credit. The heterogeneous results in the two stages confirm the rationality requirements aforementioned.

In Fig. \ref{base_c}, we study the optimal consumption rate for the  65, 80, 95 and 110 years old retiree. Particularly, the consumption rate is $0$ when the wealth in the bequest account is $0$. Basically, the consumption rate is larger when the total wealth is greater, which is consistent with Merton's model of portfolio selection (cf. \cite{M1969}). Moreover, the retiree consumes at a higher rate when she/he grows older. Due to the decreasing survival probability, the elderly retiree prefers to obtain higher utility via higher consumption rate.

\subsection{The Impacts of the magnitude of Bequest Motive on Transaction Regions}
In the literature, the magnitude of bequest motive  plays an important role in determining the tontine allocation policies. In Fig. \ref{figure_b}, we study the impacts of bequest motive on the transaction regions for the 65 and 110 years old retiree. At age 65, as bequest motive increases, the Sell region of tontine enlarges and the Buy region of tontine shrinks. The retiree prefers to allocate more wealth in the bequest account when she/he has higher bequest motive. Interestingly, at age 110, the transaction regions are almost the same for different bequest motive magnitudes.  When the bequest motive is large, the retiree wants to leave more bequest. However, the retiree will also allocate in the tontine account to gamble for the abundant longevity credit at the extremely old age.  When the bequest motive is small or even approaches $0$, the retiree will also allocate in the bequest account, because that the consumption withdrawn from this account dominates the bequest in the utility function. Overall, we observe a relatively stable transaction policy for the extremely old retirees regardless of their bequest motive. This is a new result obtained by relaxing the restrictions in \cite{D2021}.

%
%
%
%
%

\begin{figure}[]
	\centering
	\begin{subfigure}[t]{0.4\textwidth}
		\centering
		\includegraphics[width=1\textwidth]{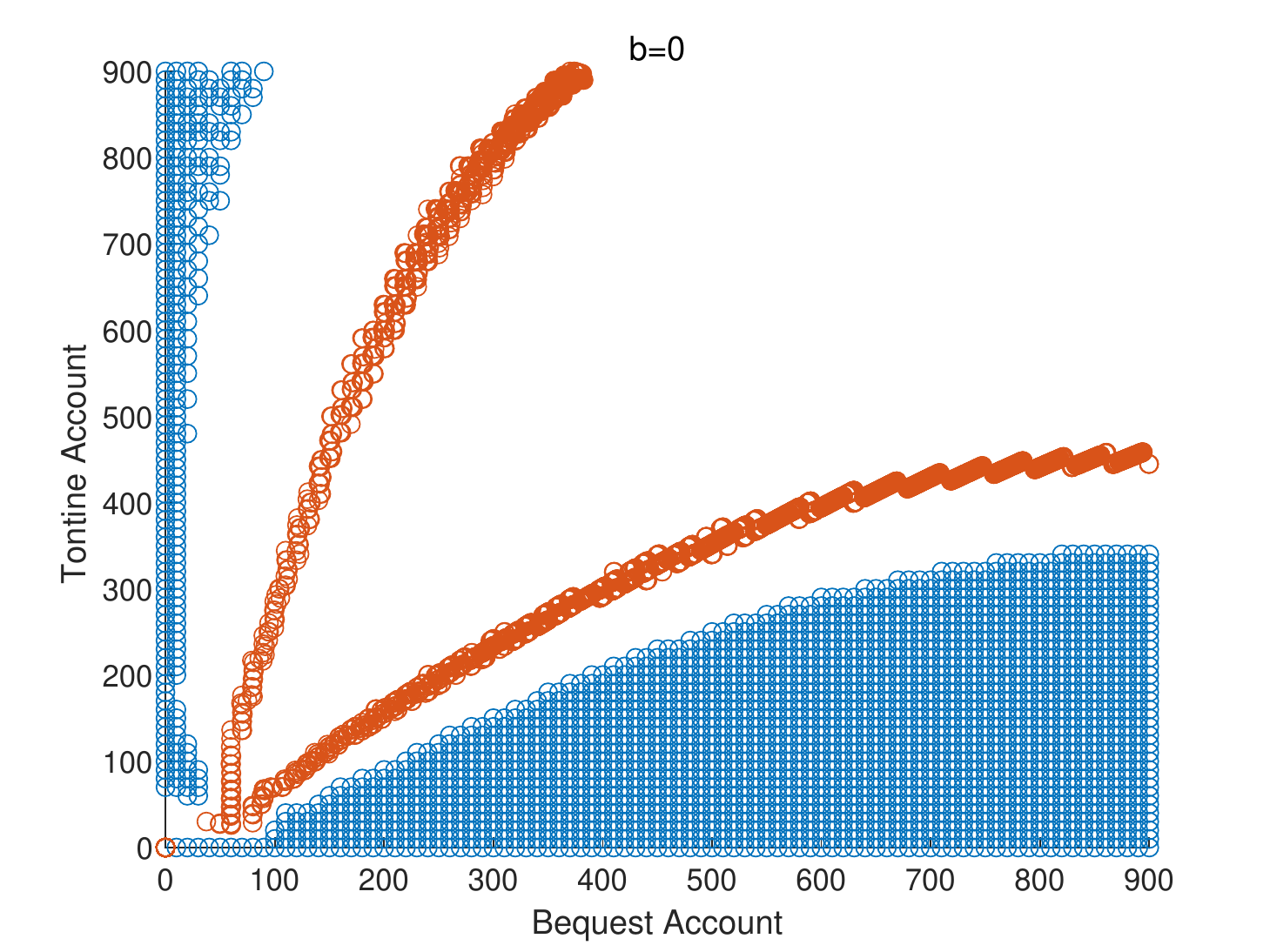}
	\end{subfigure}
	\begin{subfigure}[t]{0.4\textwidth}
		\centering
		\includegraphics[width=1\textwidth]{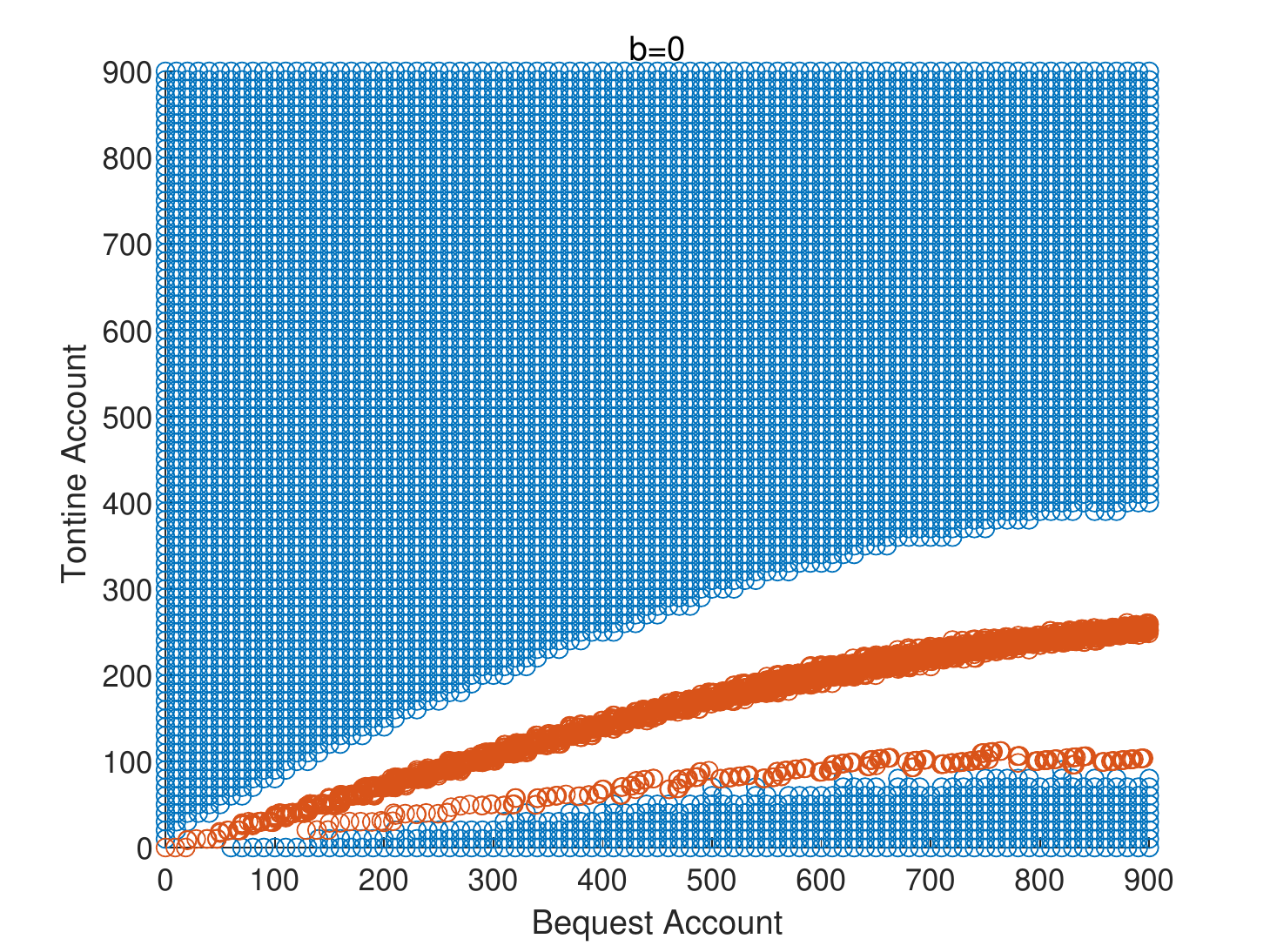}
	\end{subfigure}
	\begin{subfigure}[t]{0.4\textwidth}
		\centering
		\includegraphics[width=1\textwidth]{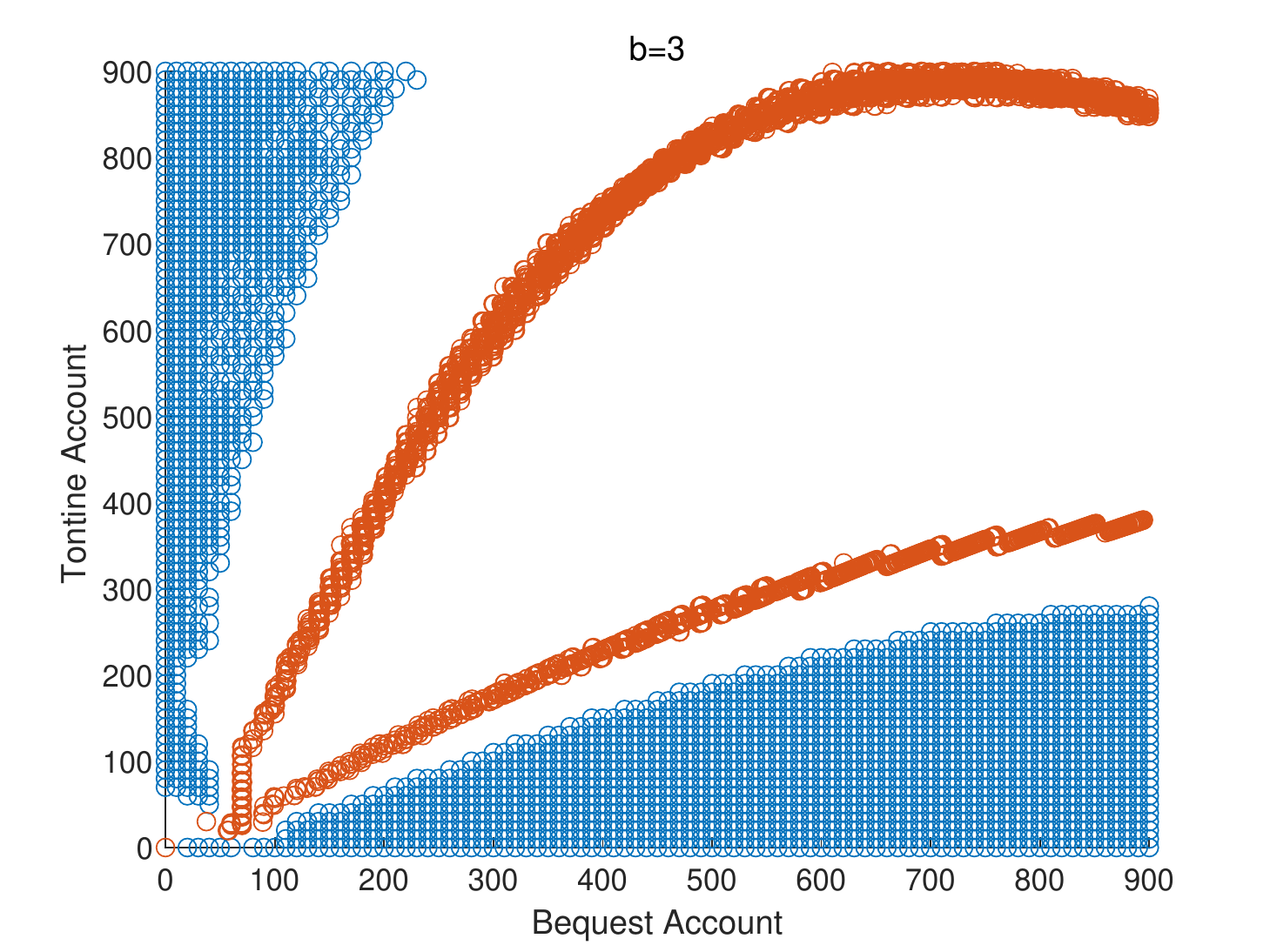}
	\end{subfigure}
	\begin{subfigure}[t]{0.4\textwidth}
		\centering
		\includegraphics[width=1\textwidth]{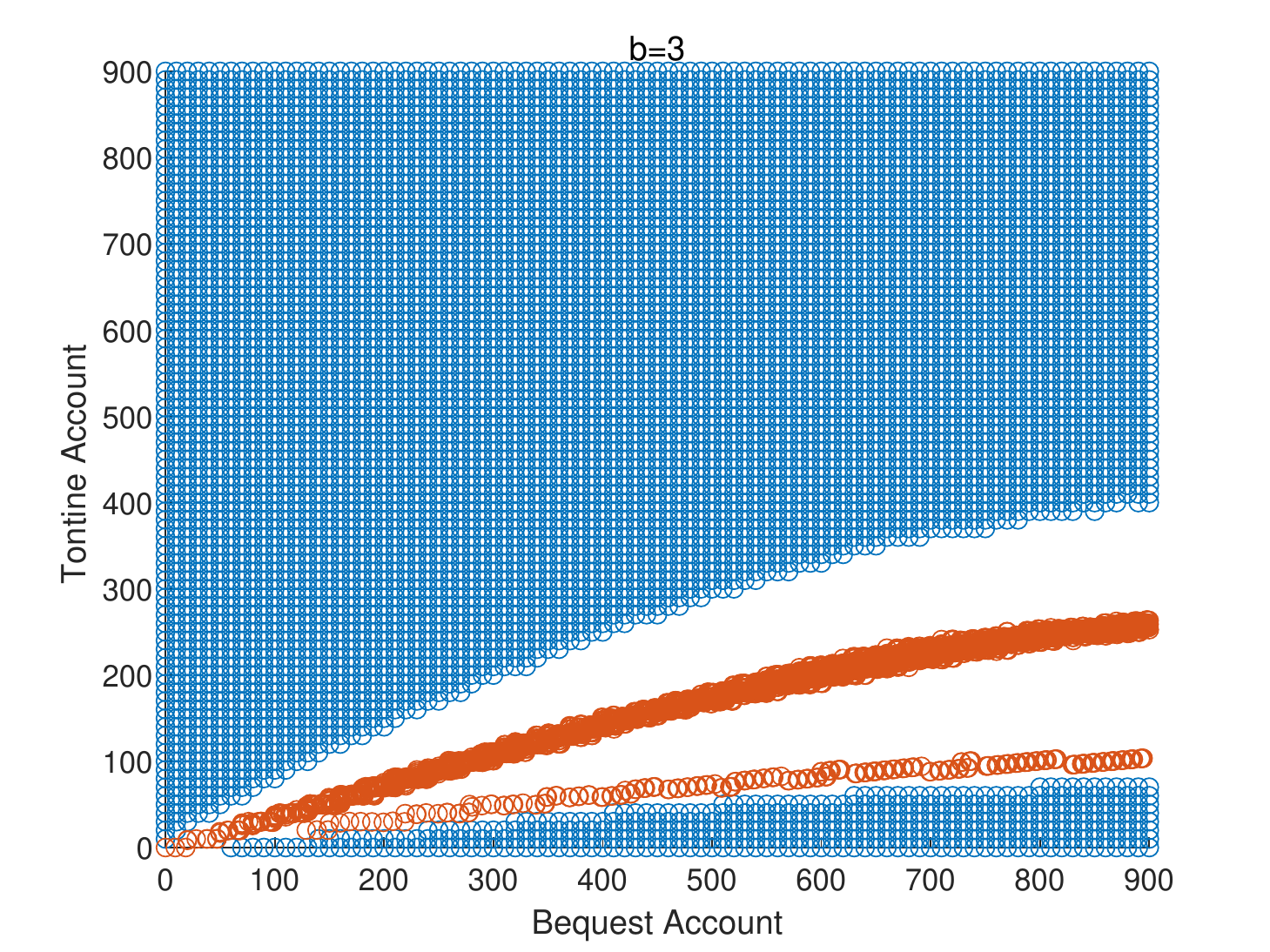}
	\end{subfigure}
	
	\begin{subfigure}[t]{0.4\textwidth}
		\centering
		\includegraphics[width=1\textwidth]{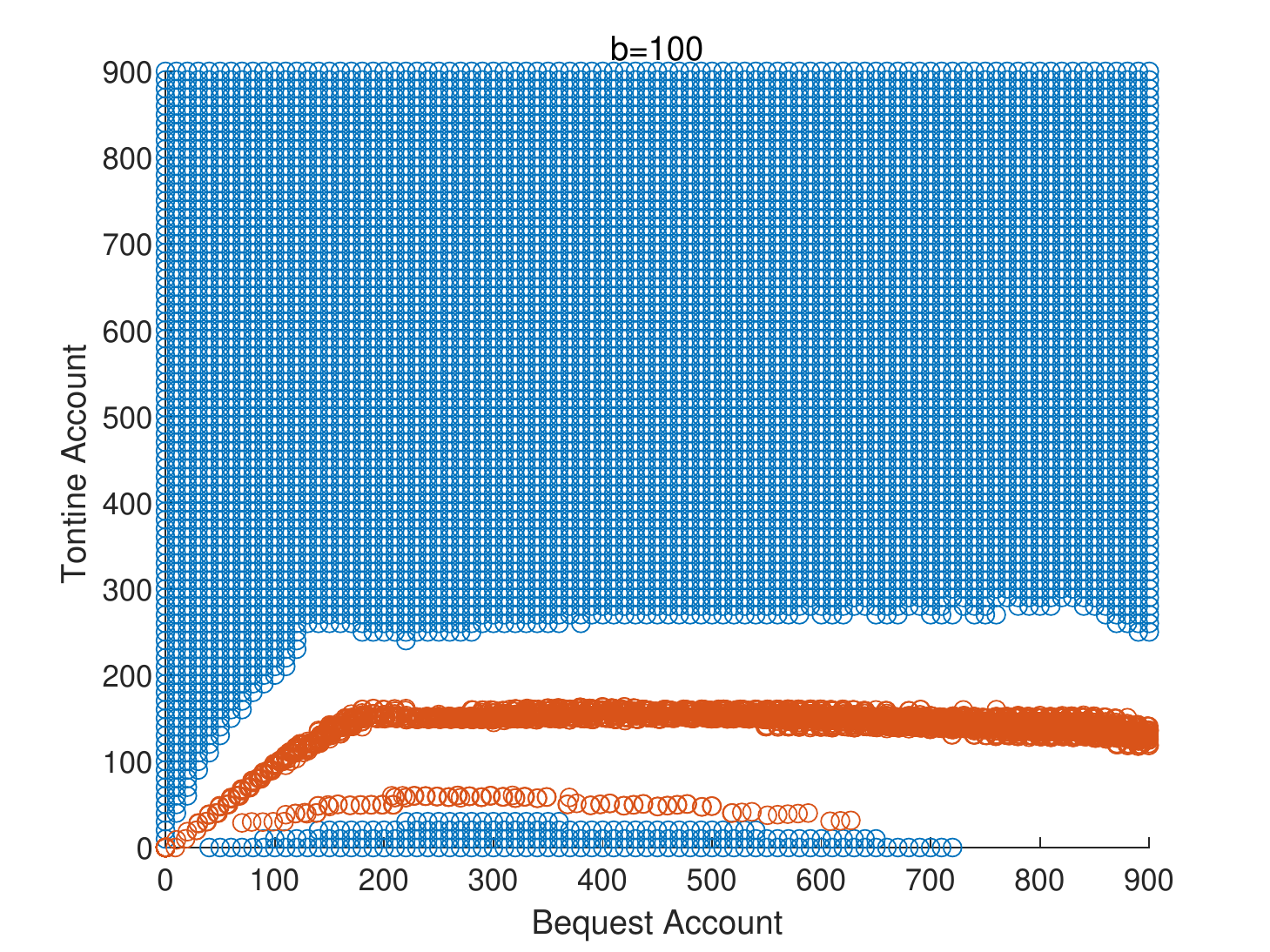}
		\caption*{ 65 years old}
	\end{subfigure}
	\begin{subfigure}[t]{0.4\textwidth}
		\centering
		\includegraphics[width=1\textwidth]{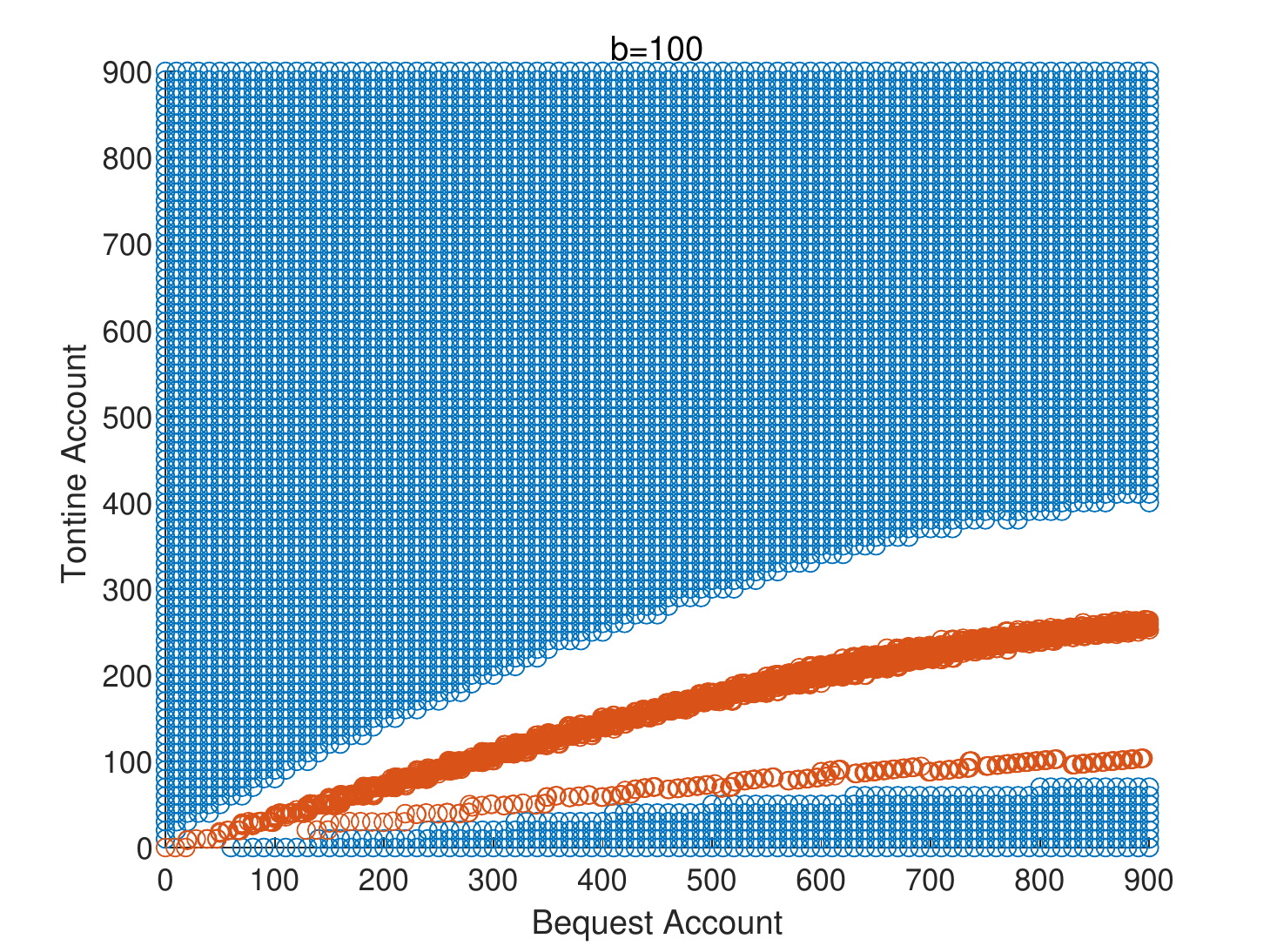}
		\caption*{ 110 years old}
	\end{subfigure}
	
	\caption{The impacts of the magnitude of bequest motive on transaction regions }  \label{figure_b}
\end{figure}

\subsection{The Impacts of the Relative Risk Aversion Parameter and Force of Mortality on Transaction Regions}
In this subsection, we further study the factors that have  important impacts on the optimal transaction policies, especially for the extremely old retirees.
%
%
%

In Fig. \ref{figure_p}, we study the impacts of relative risk aversion parameter on the transaction regions. When the relative risk aversion parameter $1-p$ decreases from $0.9$ to $0.3$,  the retiree becomes less risk averse. Accordingly, the Sell tontine region shrinks. At age 65, less risk averse retiree prefers to allocate more wealth in the tontine account to obtain higher utility via bearing higher return volatility. Similar to  the results in Fig. \ref{base_region}, the proportion allocated in the tontine account gradually shrinks when the retiree grows older. Naturally, the retiree with less risk averse attitude keeps higher tontine proportion. Interestingly, at the extremely old age of $110$, risk averse attitude  becomes the decisive factor in determining  the optimal tontine allocation policies. The less risk averse retiree allocates significantly higher proportion in the tontine account. In this circumstance, she/he is more willing to gamble for the longevity credit in the tontine account. These results also confirm the rationalities aforementioned.

%
\begin{figure}[]
	\centering
	\begin{subfigure}[t]{0.22\textwidth}
		\centering
		\includegraphics[width=1\textwidth]{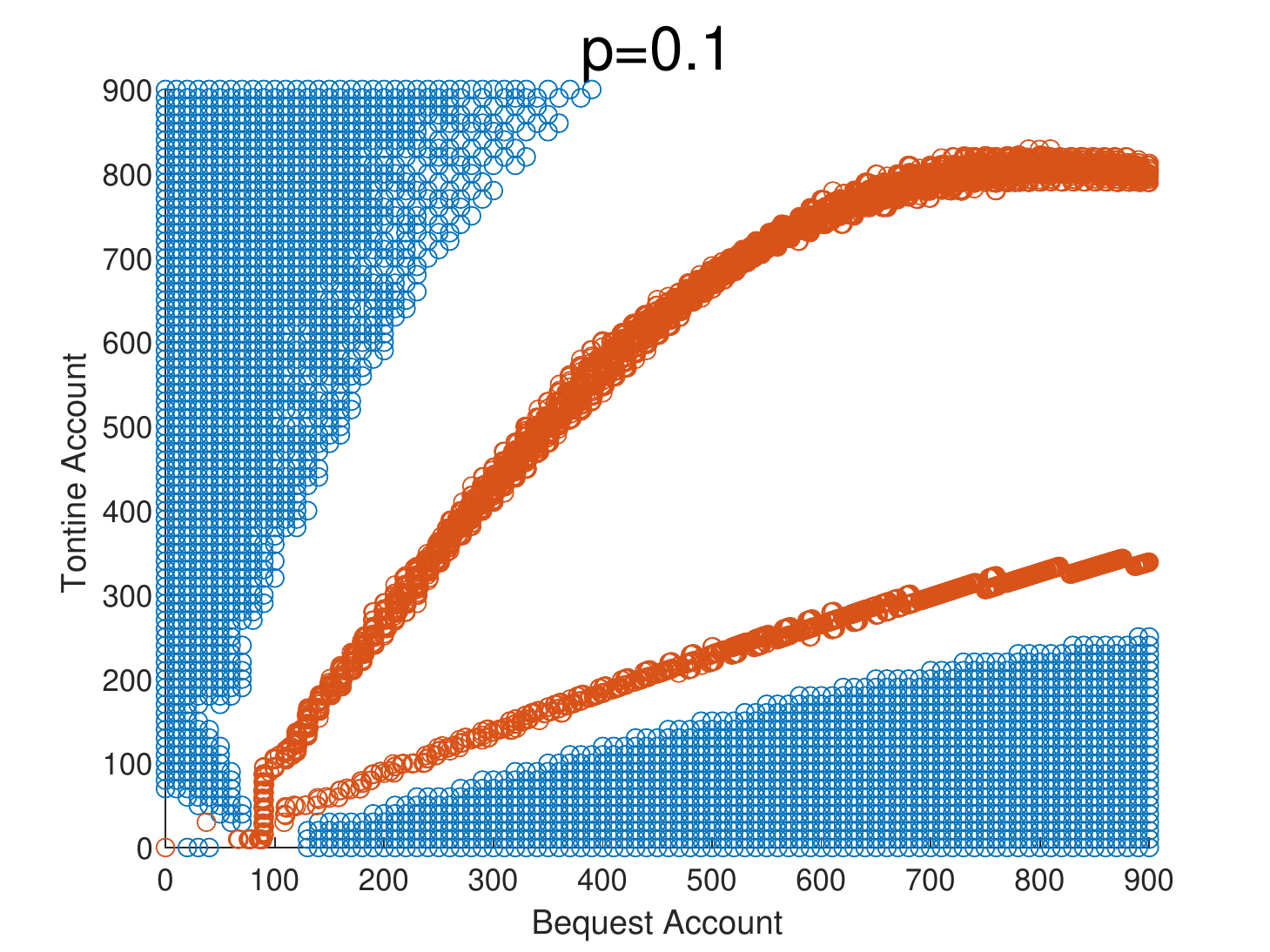}
	\end{subfigure}
	\begin{subfigure}[t]{0.22\textwidth}
		\centering
		\includegraphics[width=1\textwidth]{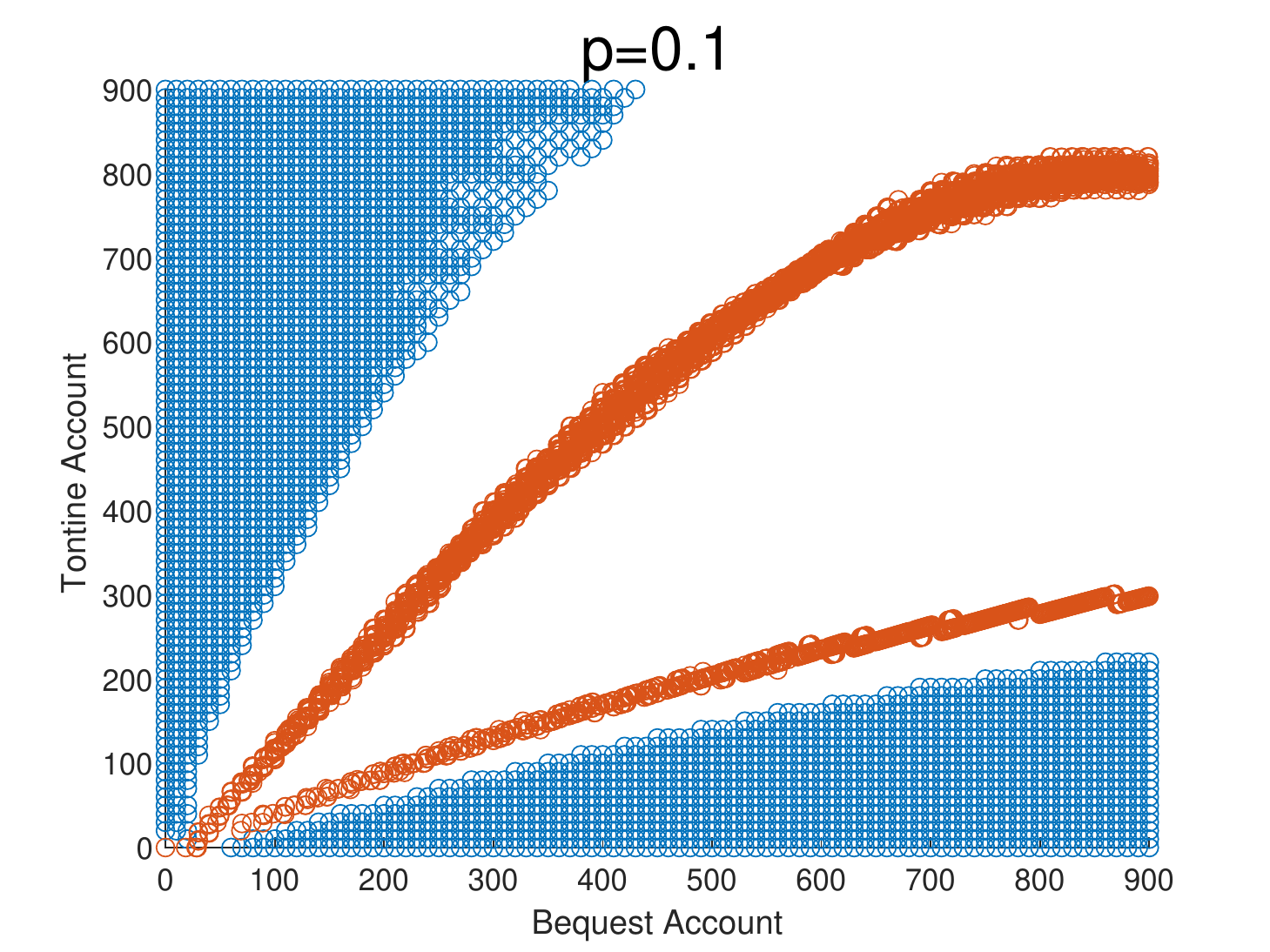}
	\end{subfigure}
	\begin{subfigure}[t]{0.22\textwidth}
		\centering
		\includegraphics[width=1\textwidth]{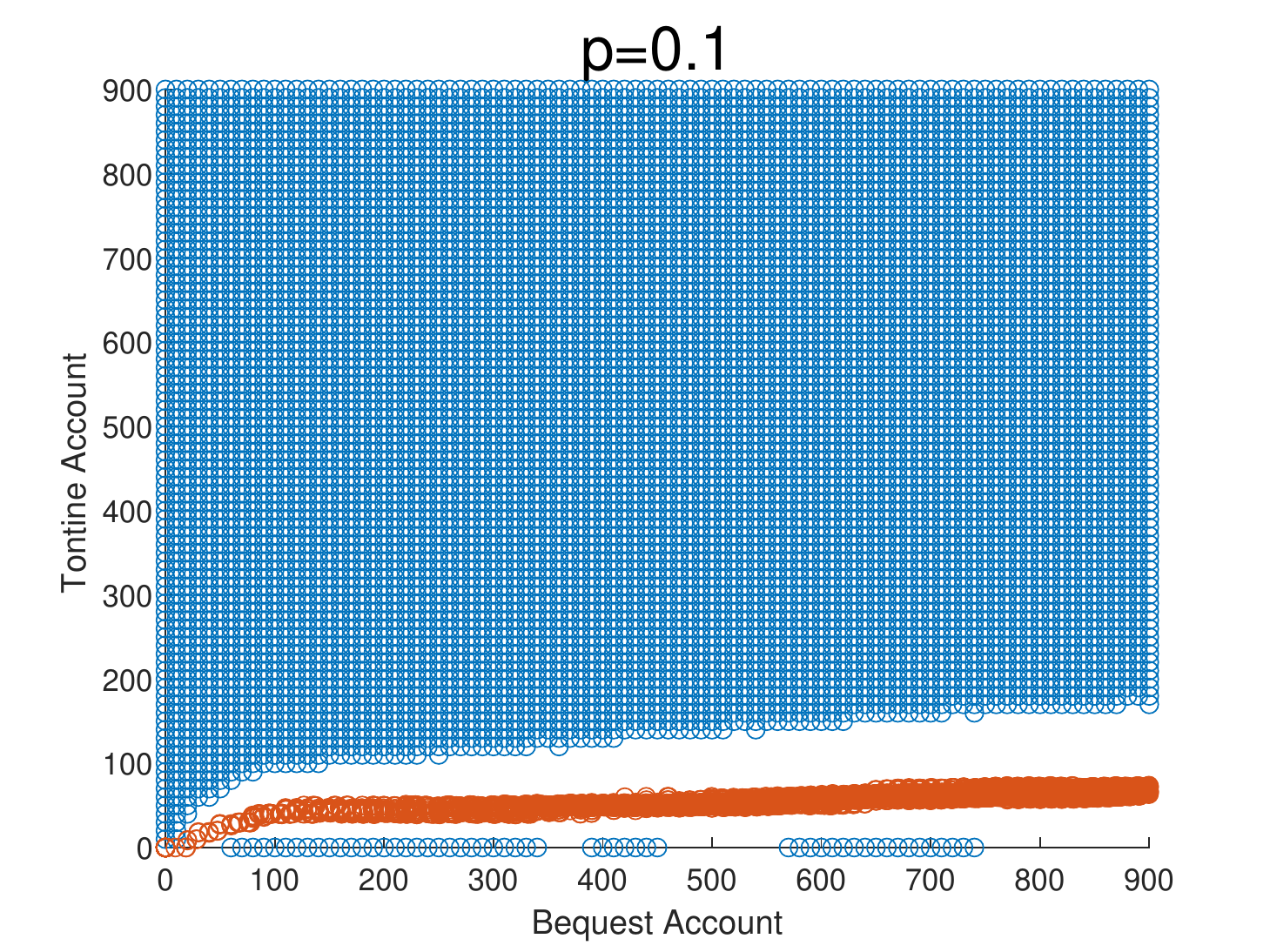}
	\end{subfigure}
	\begin{subfigure}[t]{0.22\textwidth}
		\centering
		\includegraphics[width=1\textwidth]{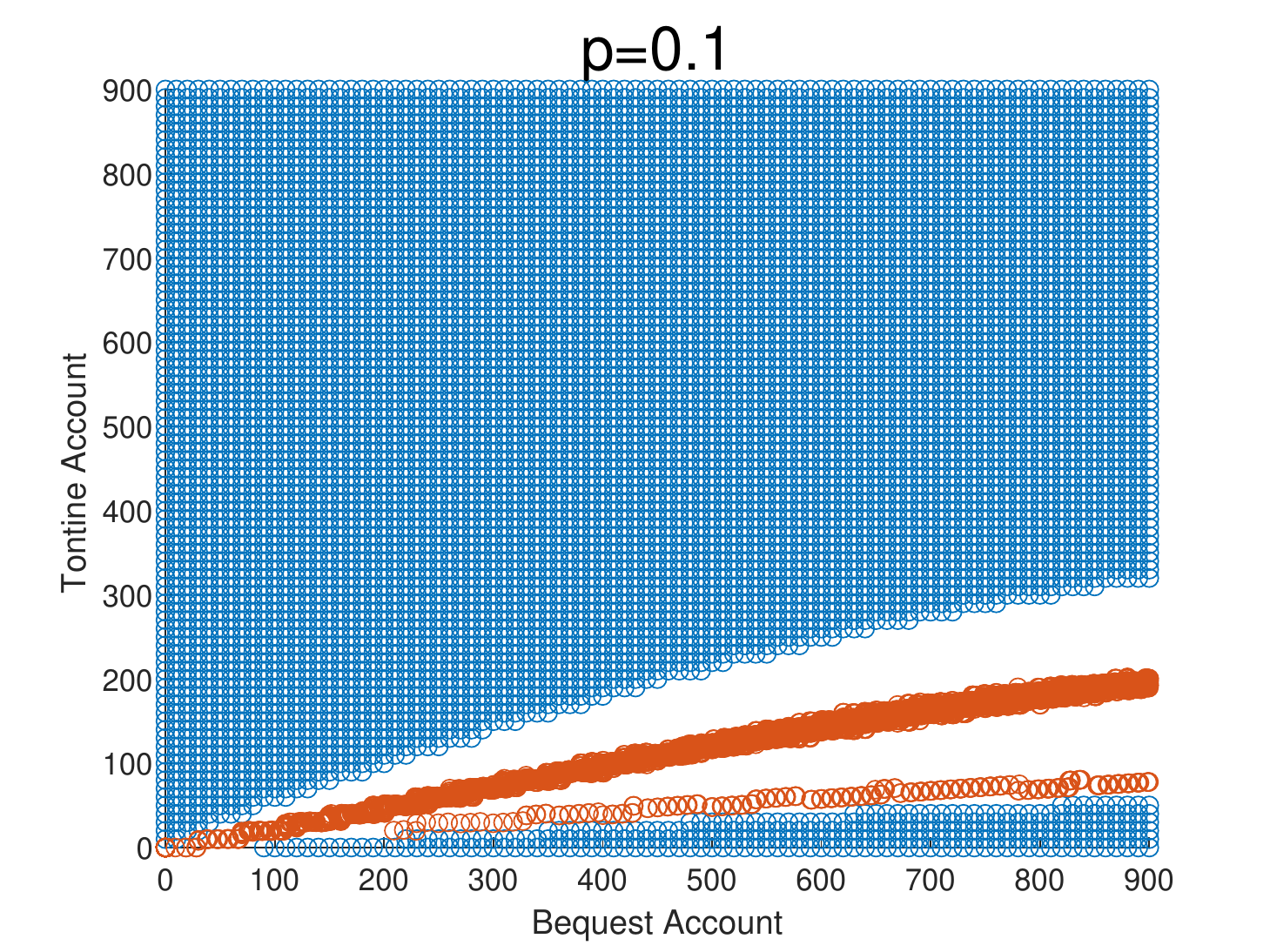}
	\end{subfigure}
	\begin{subfigure}[t]{0.22\textwidth}
		\centering
		\includegraphics[width=1\textwidth]{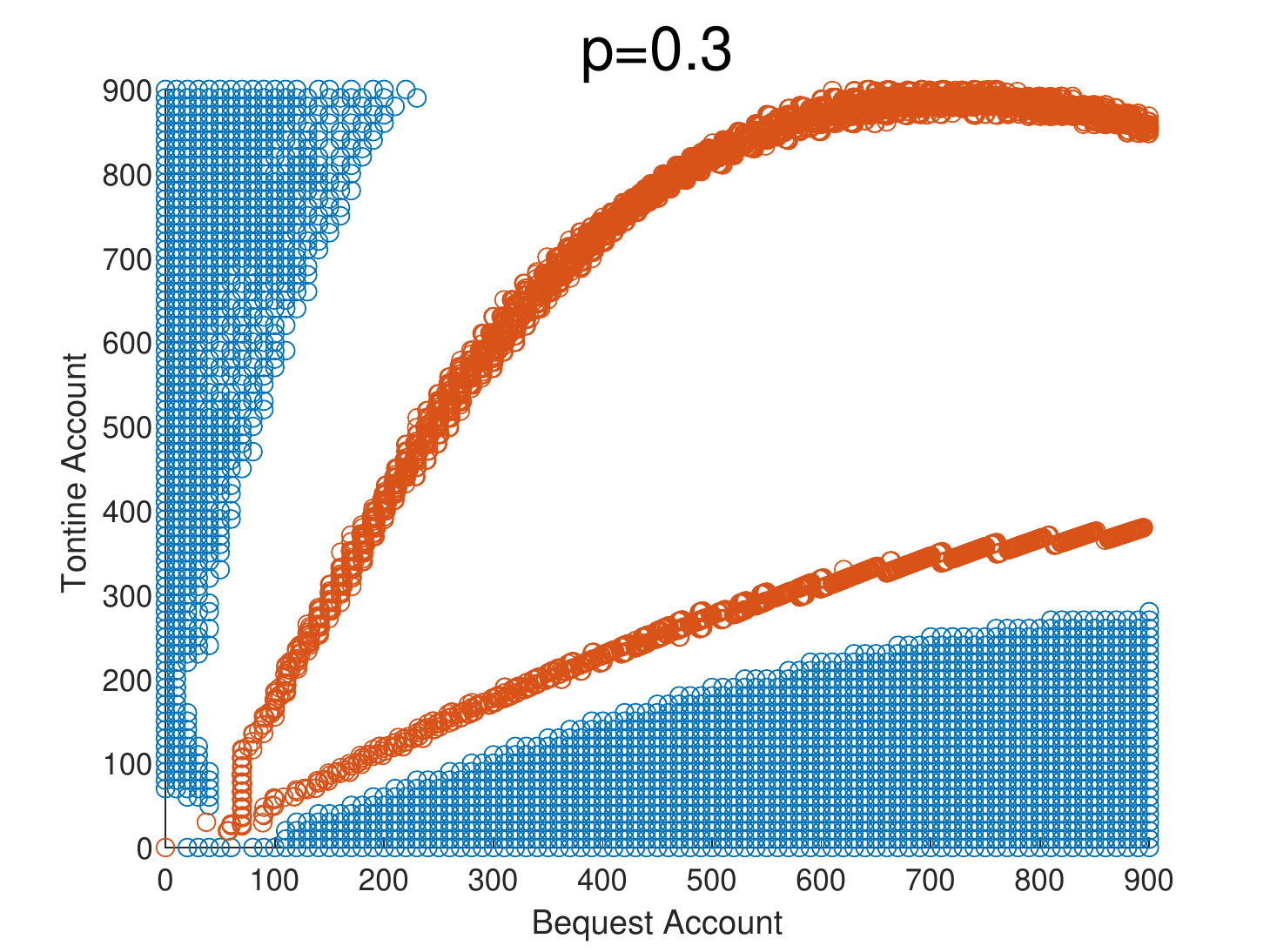}
	\end{subfigure}
	\begin{subfigure}[t]{0.22\textwidth}
		\centering
		\includegraphics[width=1\textwidth]{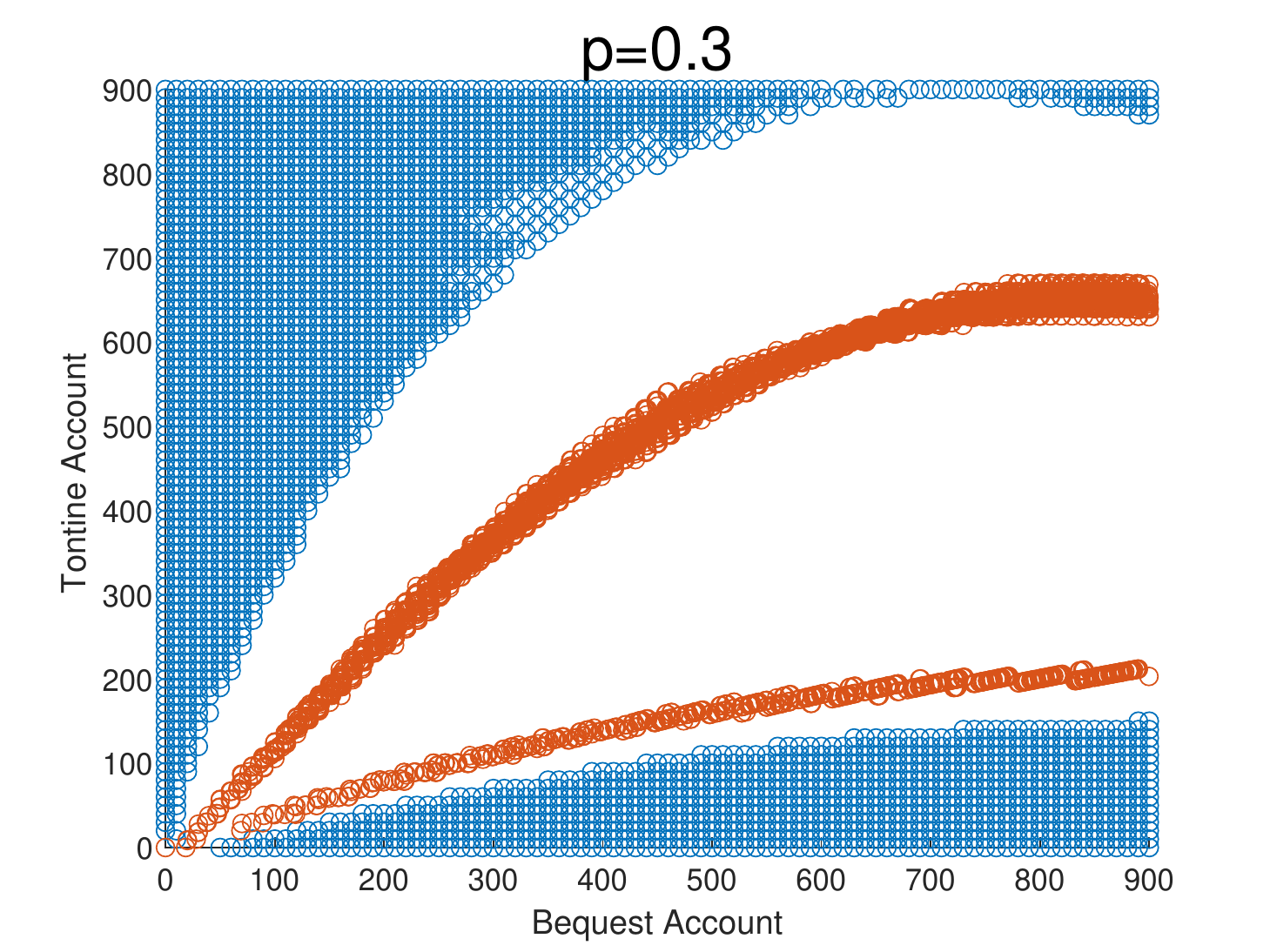}
	\end{subfigure}
	\begin{subfigure}[t]{0.22\textwidth}
		\centering
		\includegraphics[width=1\textwidth]{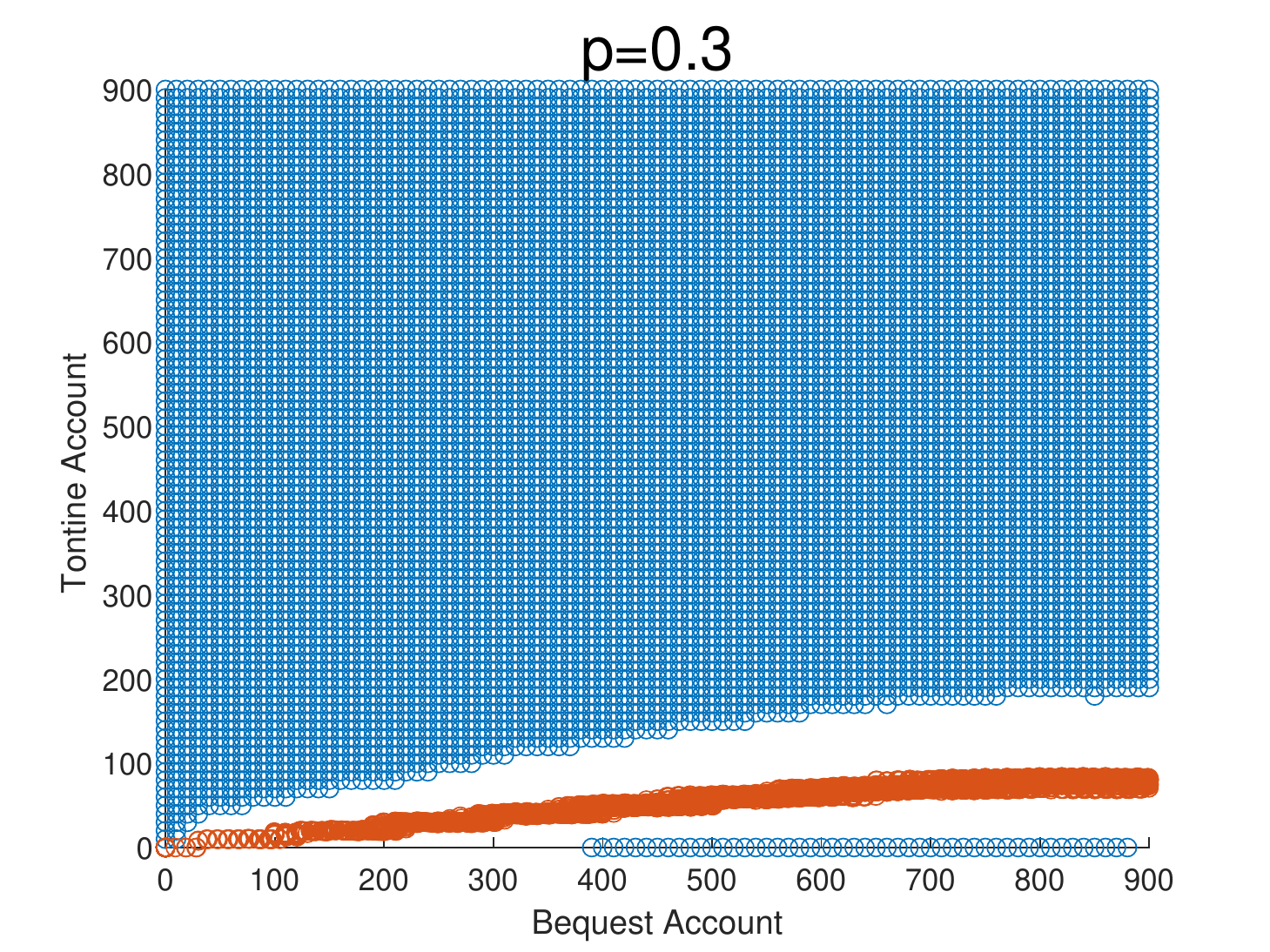}
	\end{subfigure}
	\begin{subfigure}[t]{0.22\textwidth}
		\centering
		\includegraphics[width=1\textwidth]{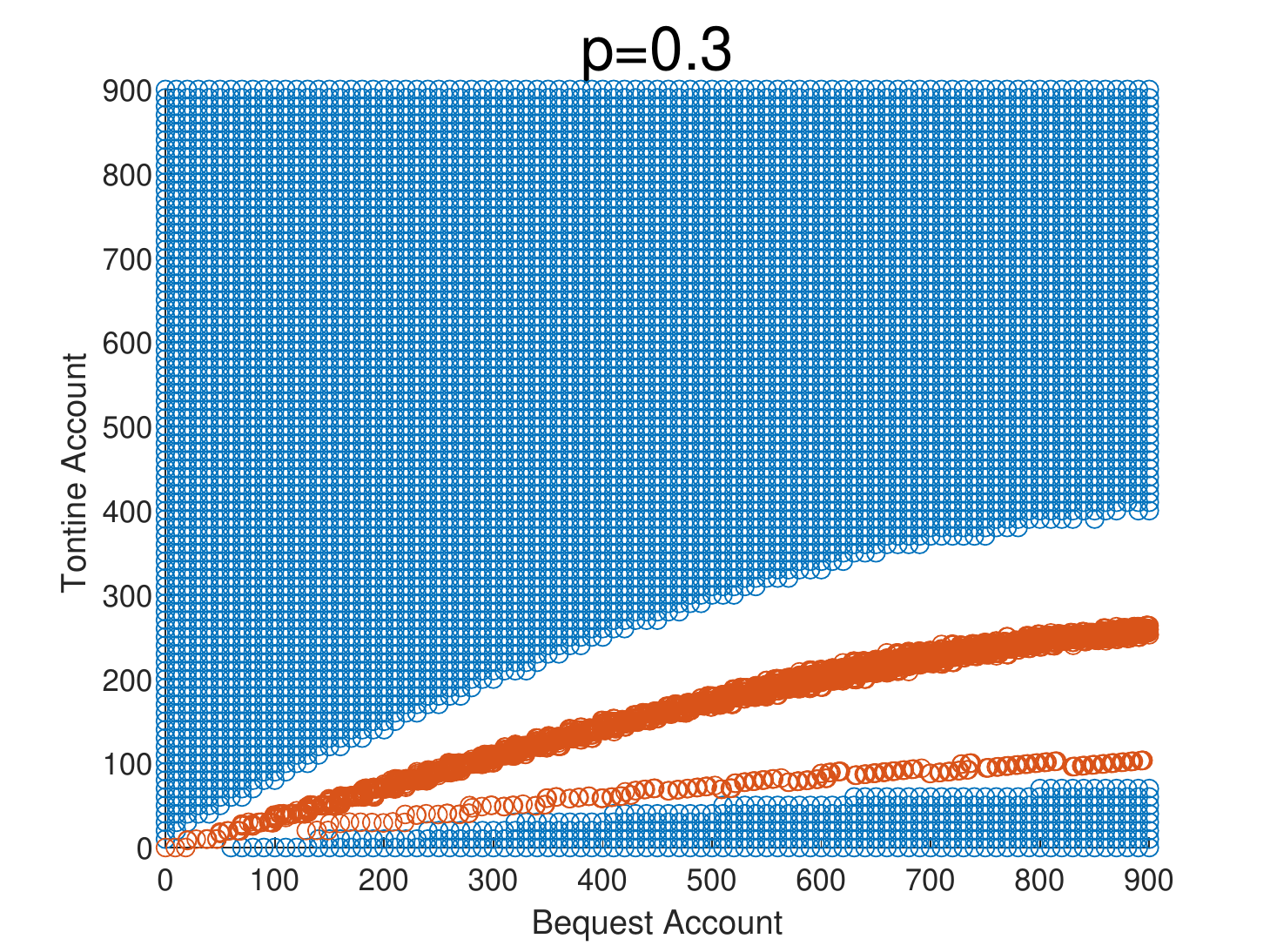}
	\end{subfigure}
	\begin{subfigure}[t]{0.22\textwidth}
		\centering
		\includegraphics[width=1\textwidth]{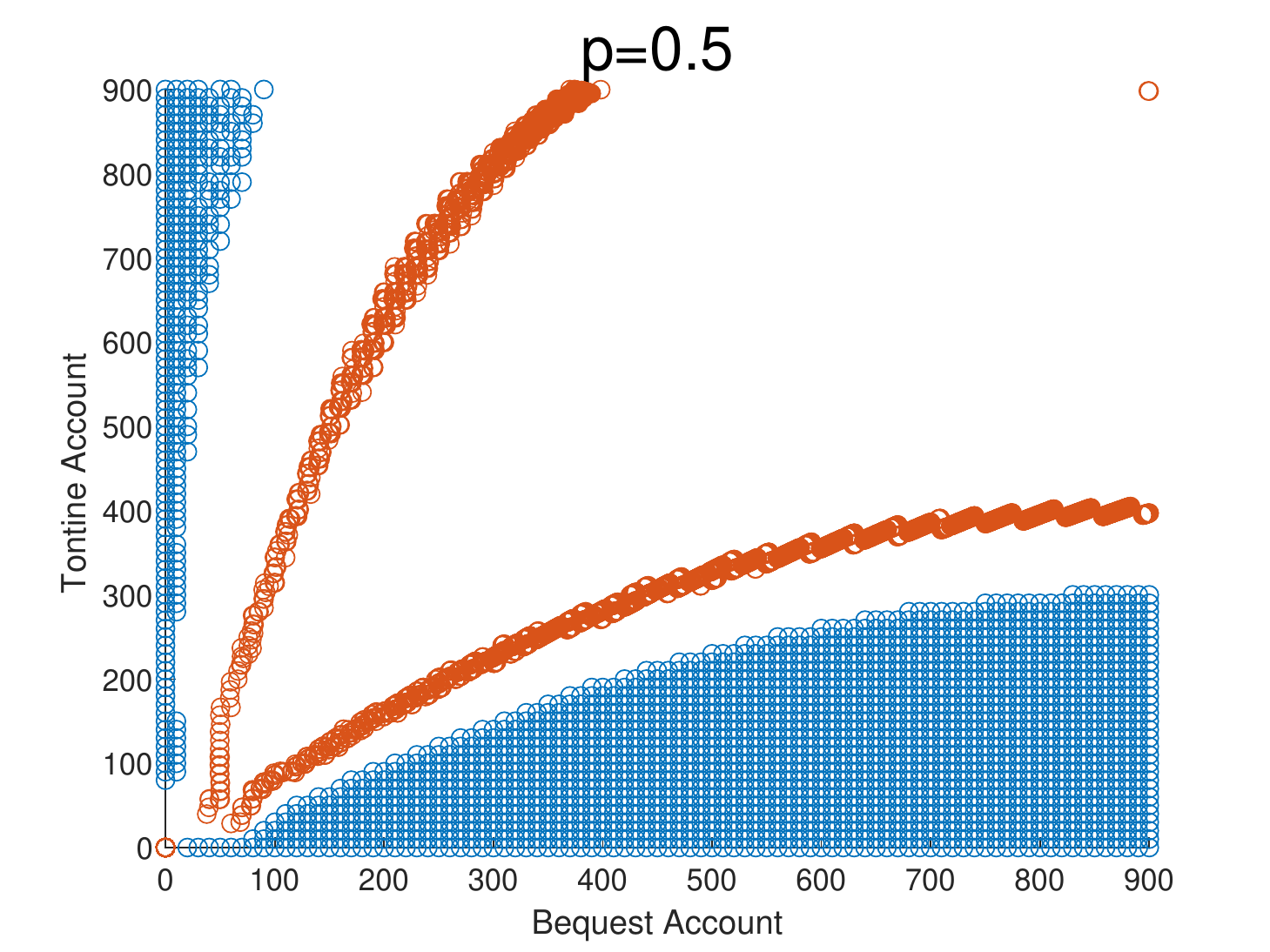}
	\end{subfigure}
	\begin{subfigure}[t]{0.22\textwidth}
		\centering
		\includegraphics[width=1\textwidth]{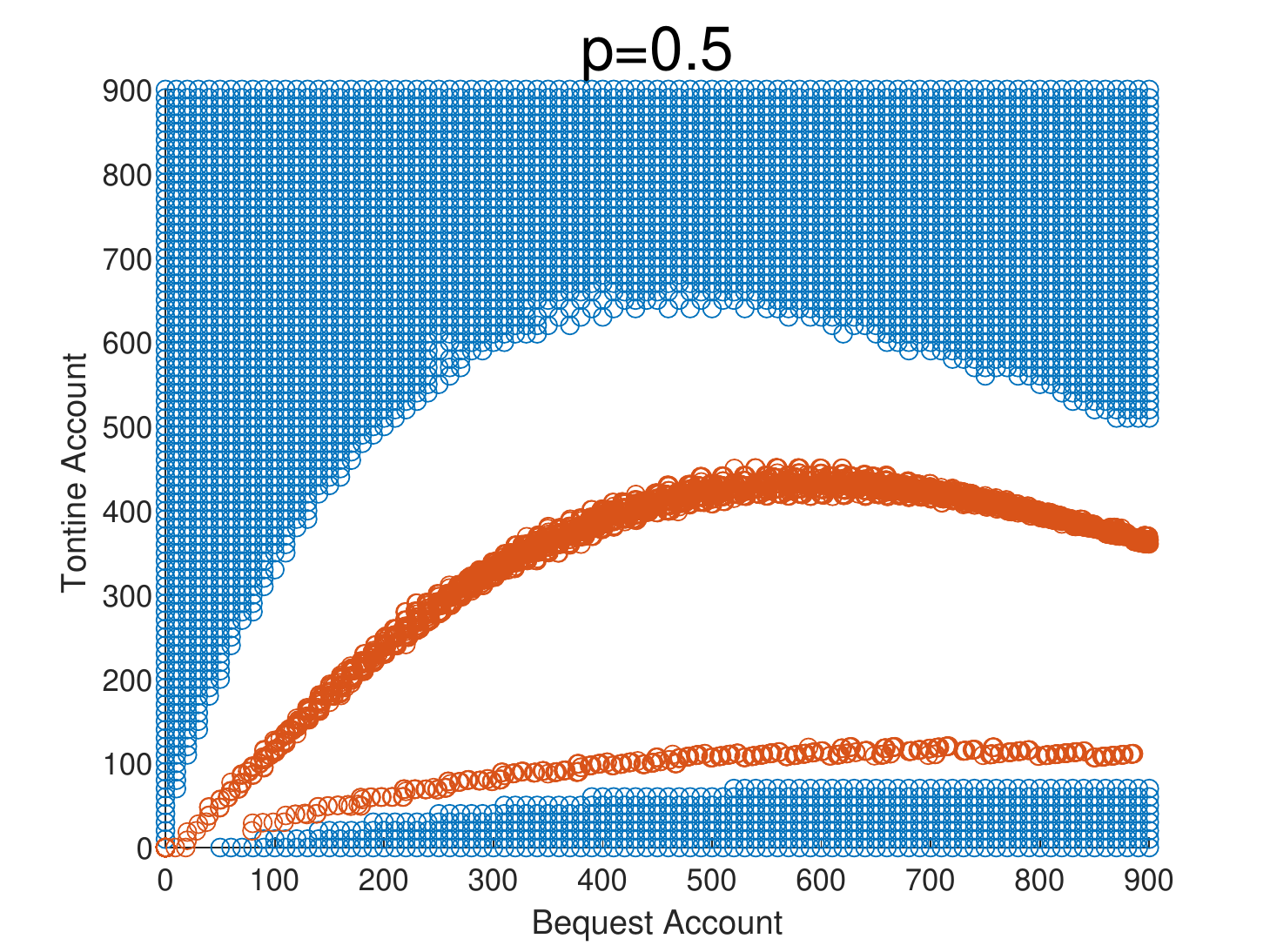}
	\end{subfigure}
	\begin{subfigure}[t]{0.22\textwidth}
		\centering
		\includegraphics[width=1\textwidth]{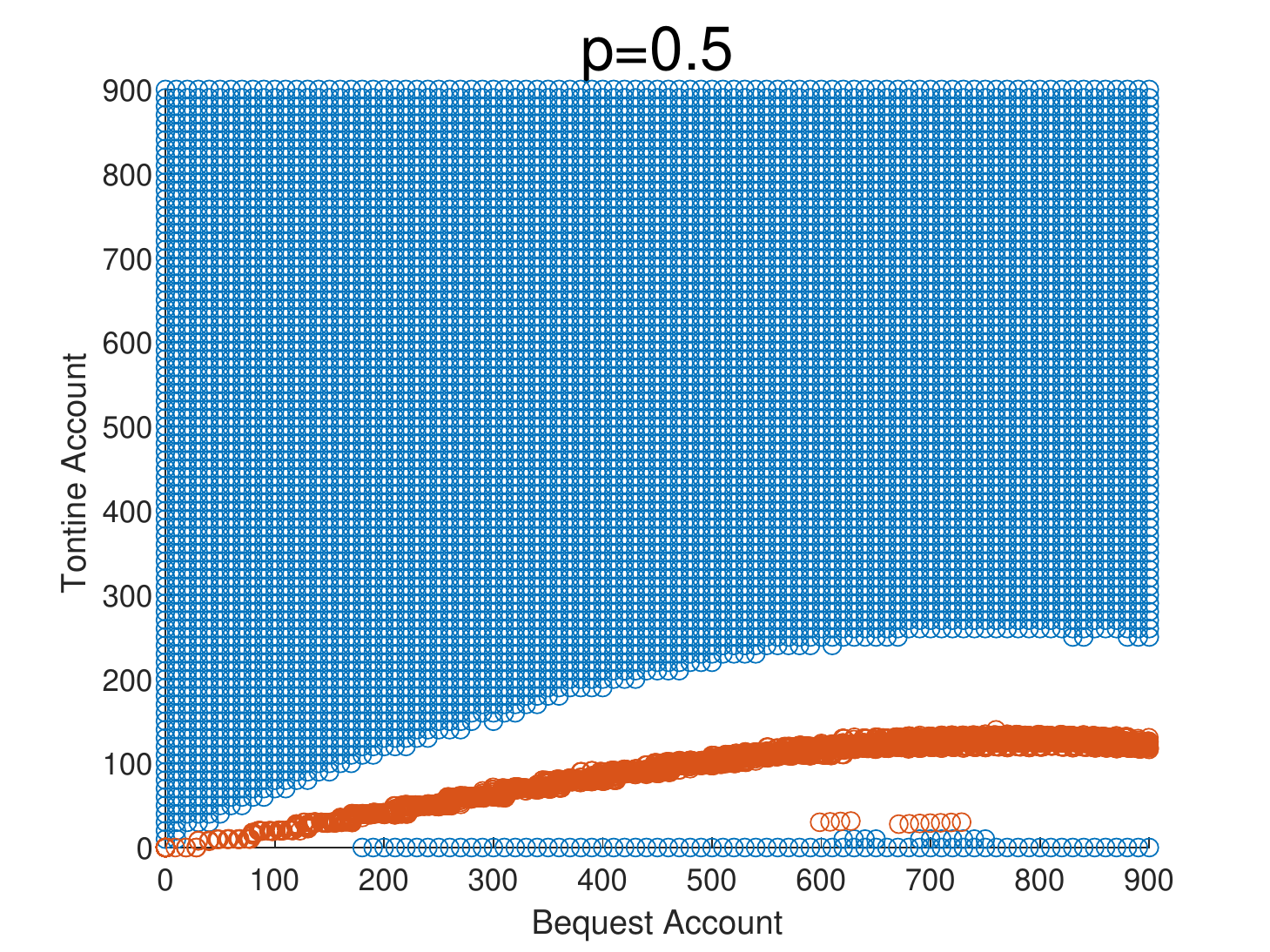}
	\end{subfigure}
	\begin{subfigure}[t]{0.22\textwidth}
		\centering
		\includegraphics[width=1\textwidth]{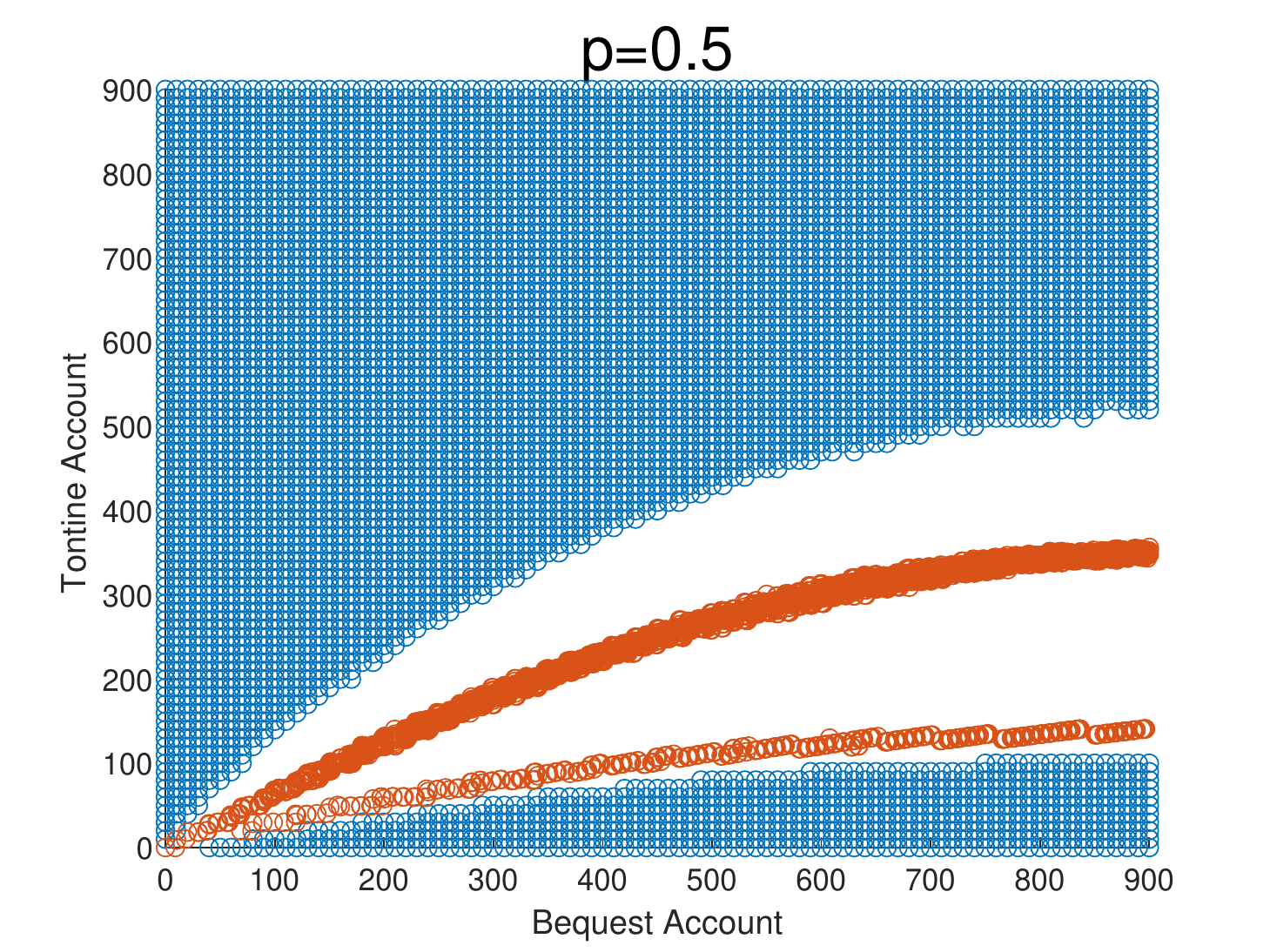}
	\end{subfigure}
	\begin{subfigure}[t]{0.22\textwidth}
		\centering
		\includegraphics[width=1\textwidth]{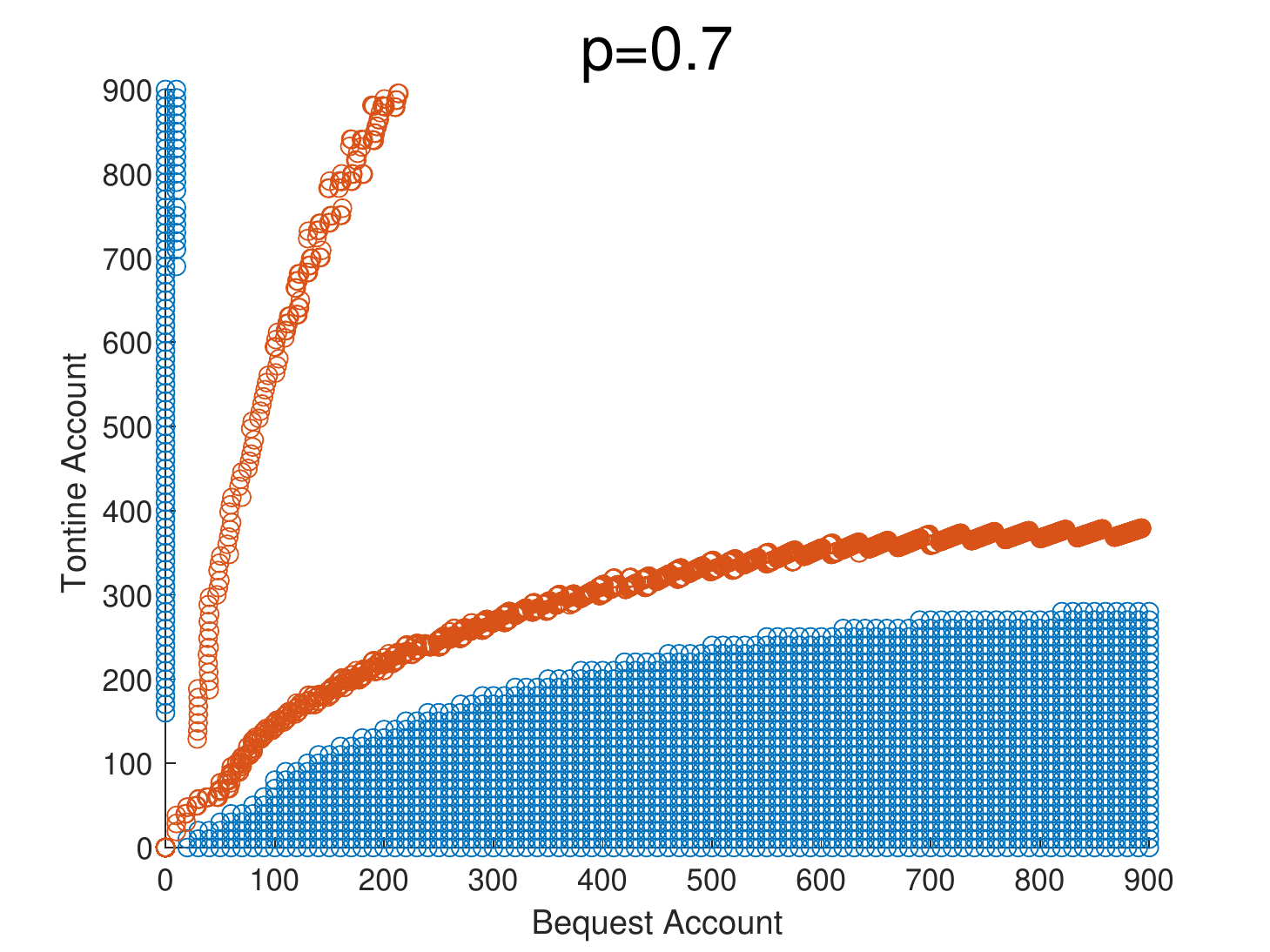}
		\subcaption*{65 years old}
	\end{subfigure}
	\begin{subfigure}[t]{0.22\textwidth}
		\centering
		\includegraphics[width=1\textwidth]{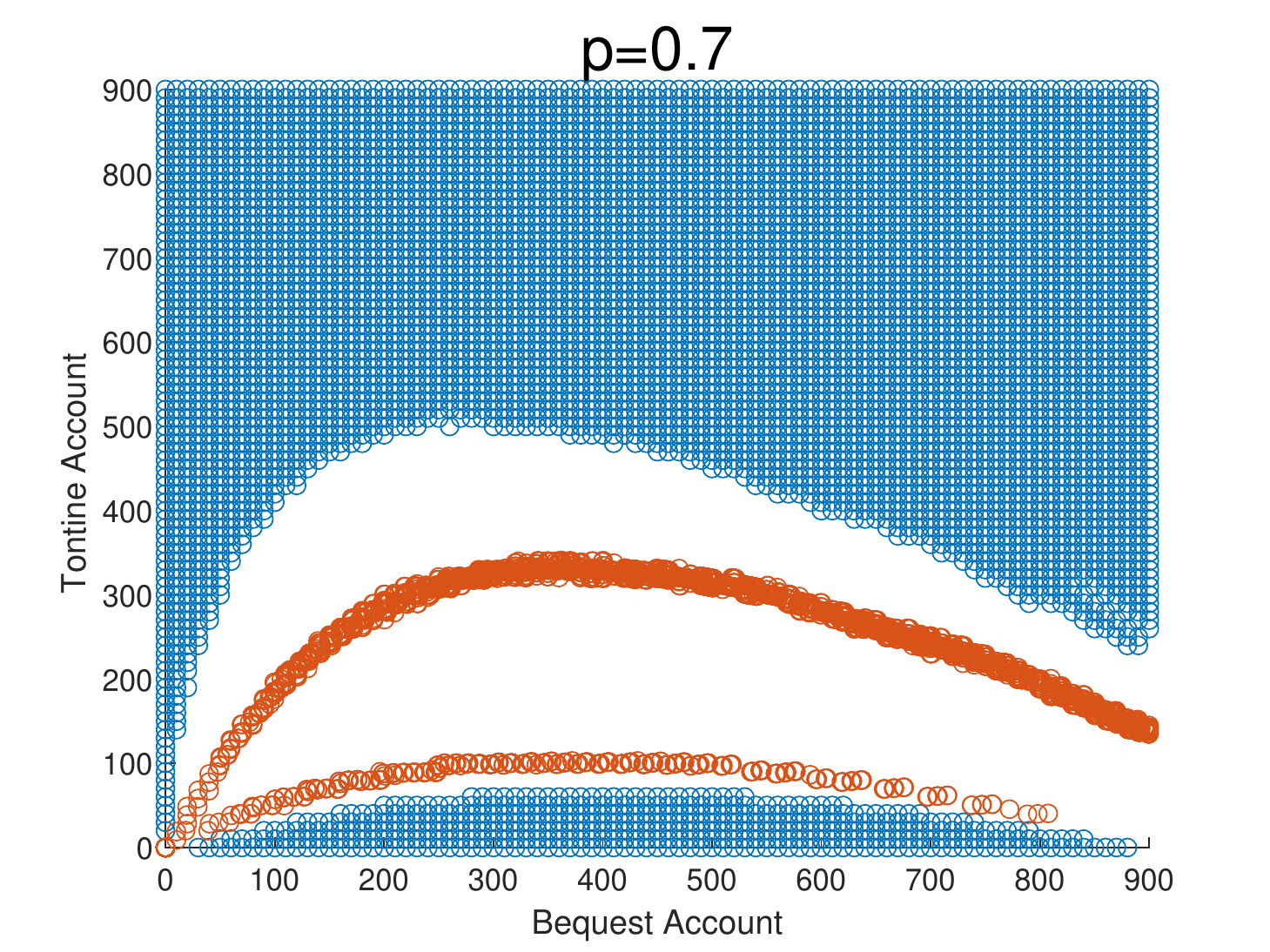}
		\subcaption*{80 years old}
	\end{subfigure}
	\begin{subfigure}[t]{0.22\textwidth}
		\centering
		\includegraphics[width=1\textwidth]{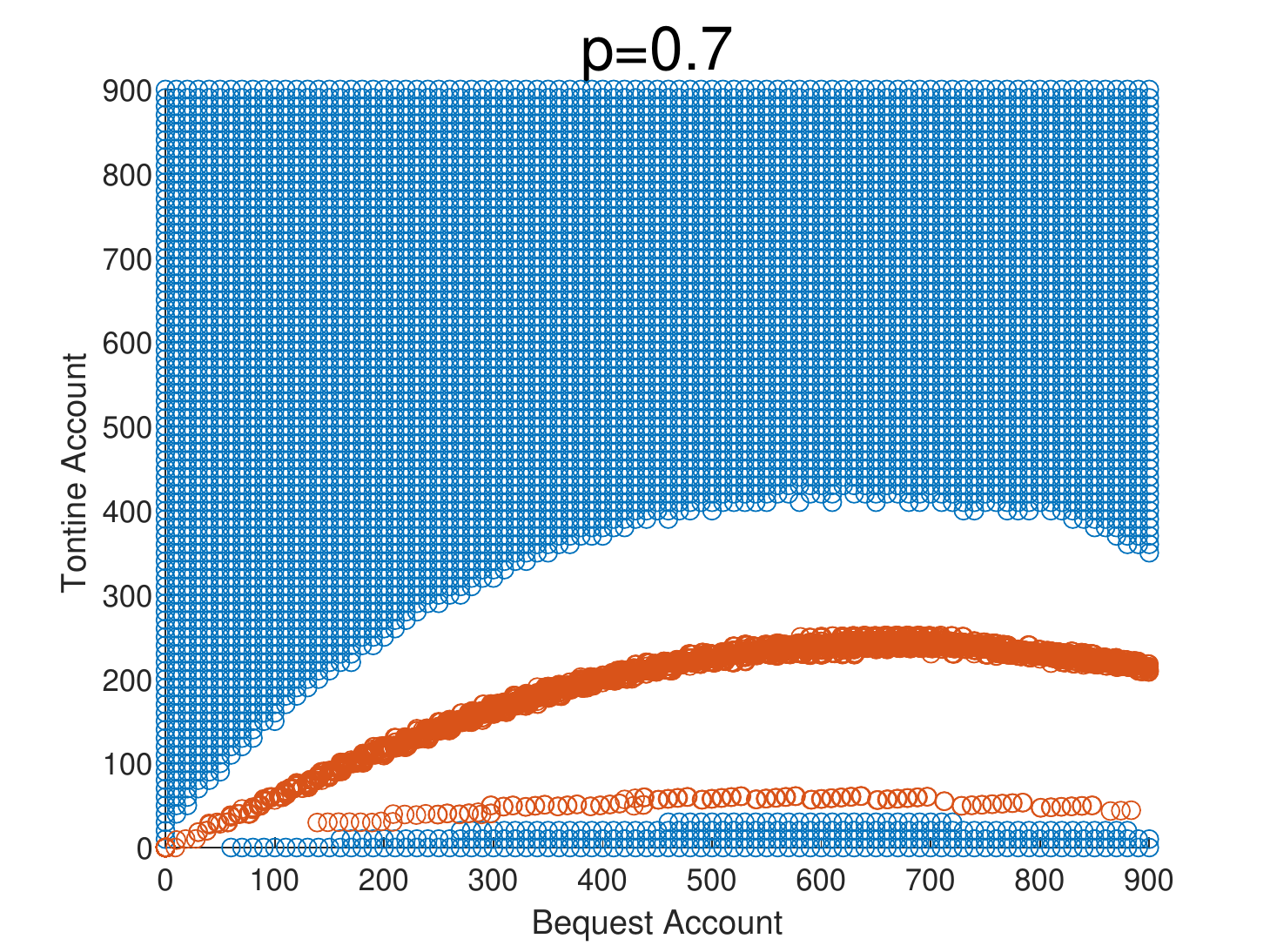}
		\subcaption*{95 years old}
	\end{subfigure}
	\begin{subfigure}[t]{0.22\textwidth}
		\centering
		\includegraphics[width=1\textwidth]{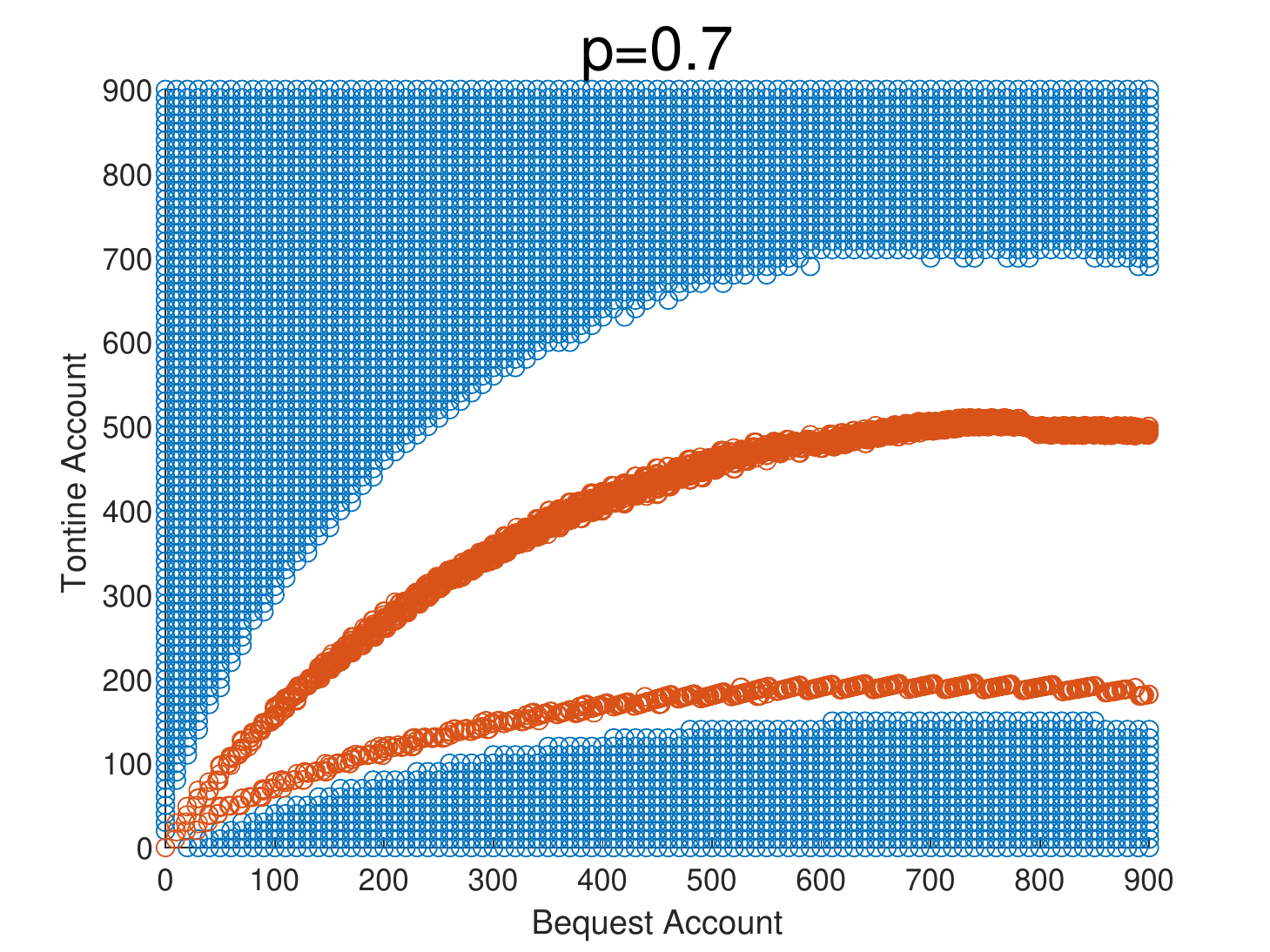}
		\subcaption*{110 years old}
	\end{subfigure}
	\caption{The impacts of relative risk aversion parameter on transaction regions}
	\label{figure_p}
\end{figure}
In Fig. \ref{figure_lambda}, we study the impacts of force of mortality on the transaction regions.
When the retiree is relatively young, she/he prefers to allocate less wealth in the tontine account in the case of higher mortality rate. In this circumstance, the longevity credit is higher. There is no need for the retiree to bear the risk in the tontine account. On the contrary, when the retiree is extremely old, she/he prefers to allocate more wealth in the tontine account to gamble for the huge longevity credit in the case of higher mortality rate.  It is noted that the force of mortality is also a decisive factor in determining the optimal tontine allocation policies.

%

\begin{figure}[]
	\centering
	\begin{subfigure}[t]{0.22\textwidth}
		\centering
		\includegraphics[width=1\textwidth]{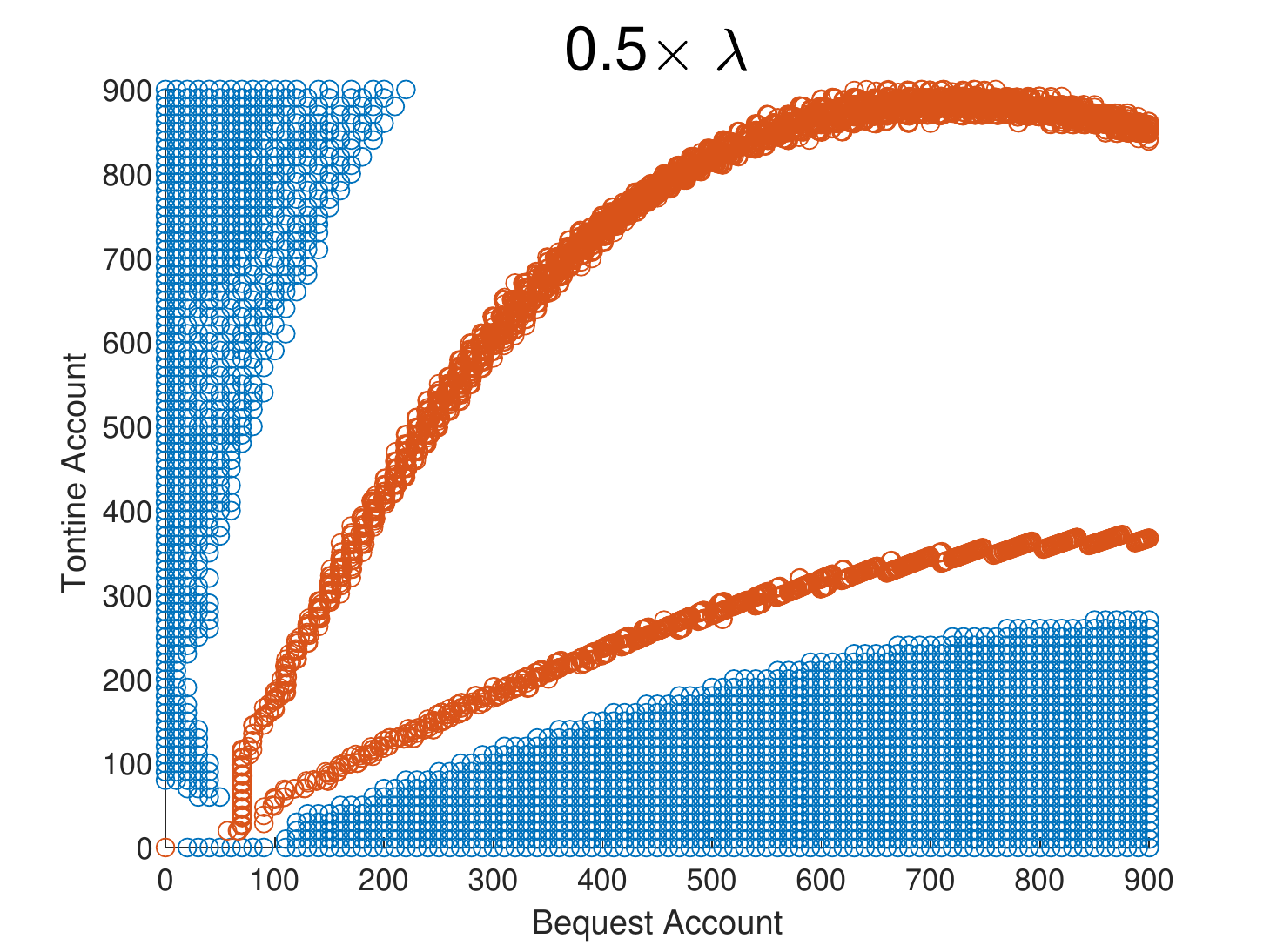}
	\end{subfigure}
	\begin{subfigure}[t]{0.22\textwidth}
		\centering
		\includegraphics[width=1\textwidth]{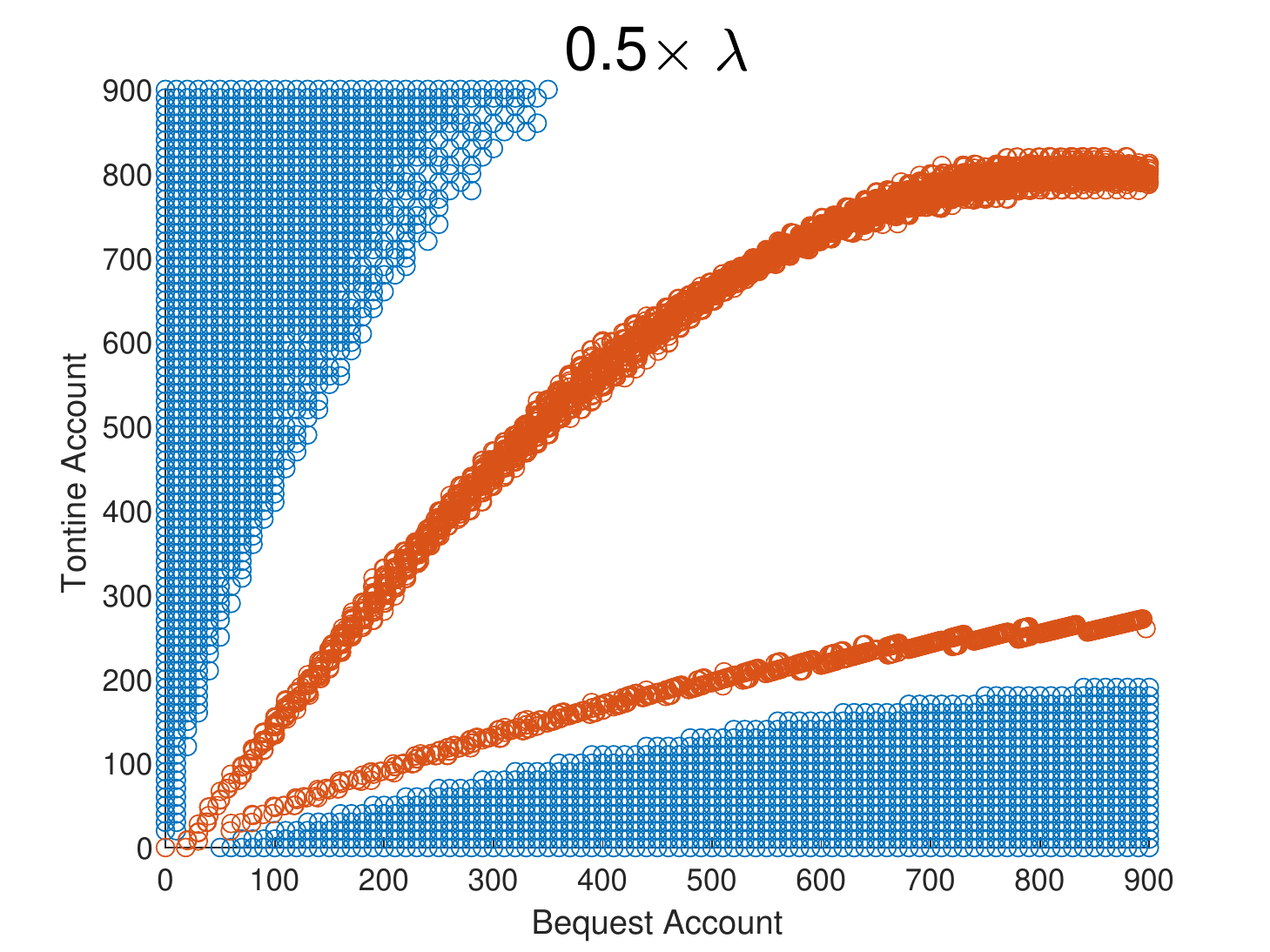}
	\end{subfigure}
	\begin{subfigure}[t]{0.22\textwidth}
		\centering
		\includegraphics[width=1\textwidth]{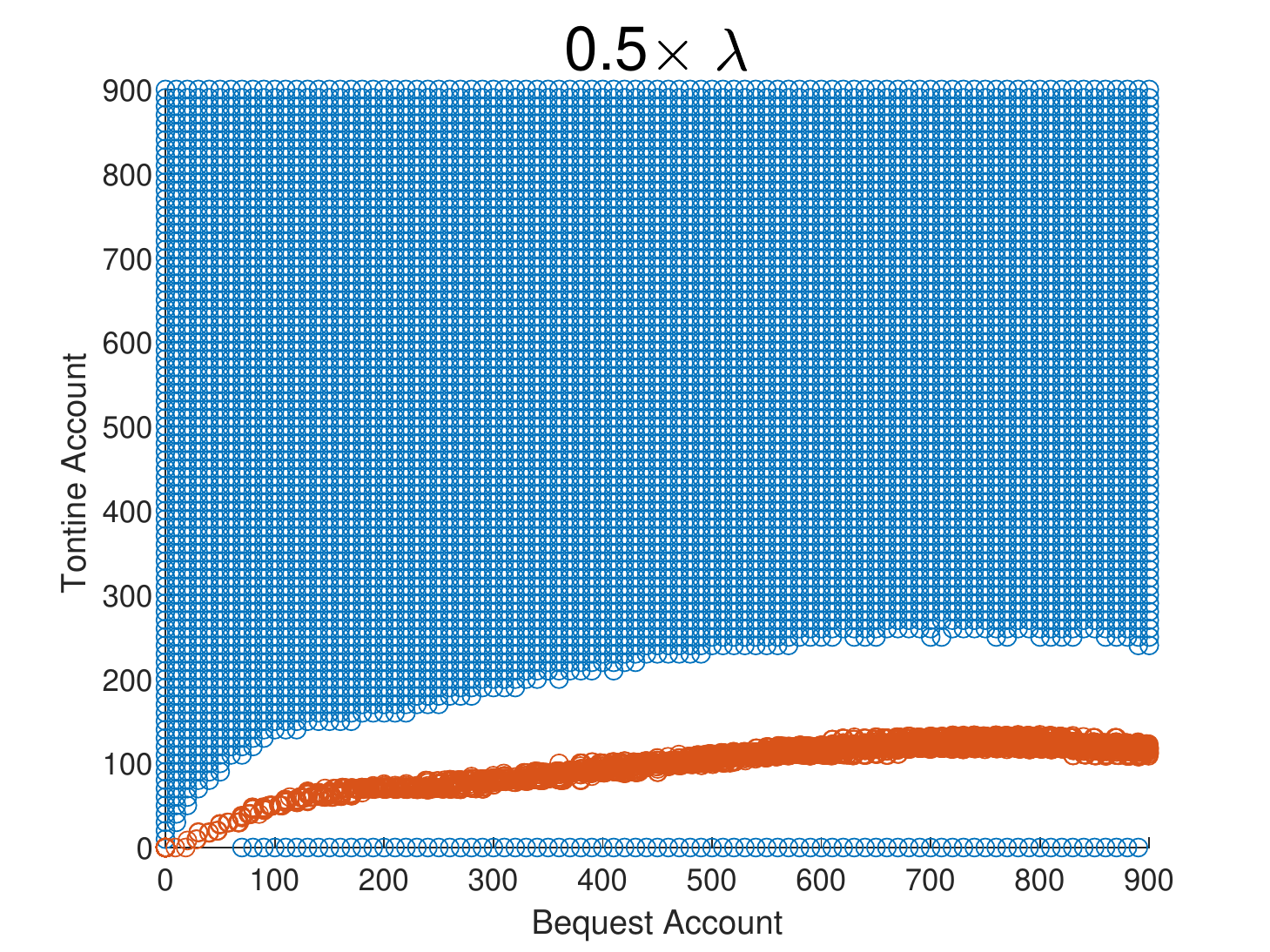}
	\end{subfigure}
	\begin{subfigure}[t]{0.22\textwidth}
		\centering
		\includegraphics[width=1\textwidth]{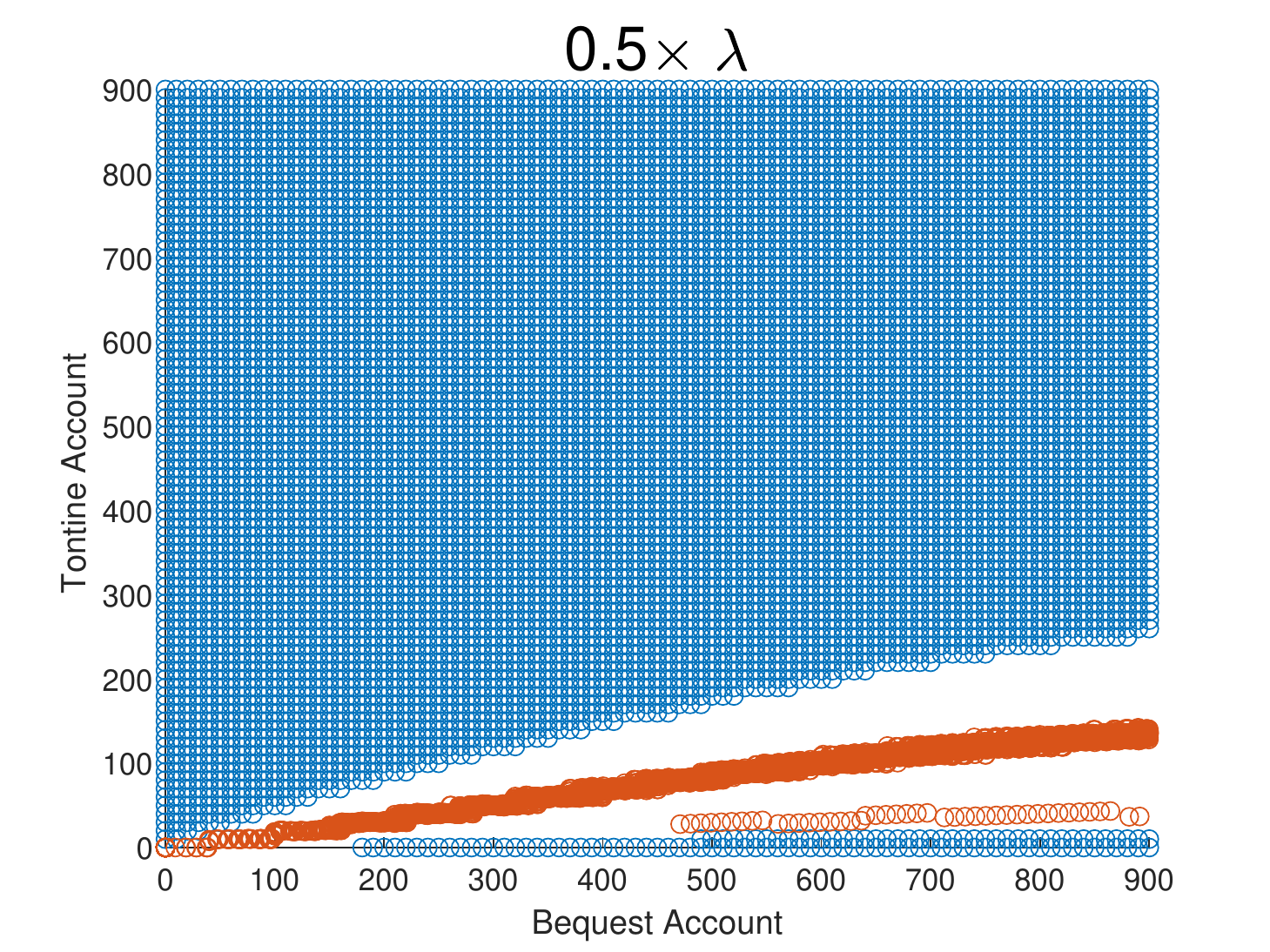}
	\end{subfigure}
	\begin{subfigure}[t]{0.22\textwidth}
		\centering
		\includegraphics[width=1\textwidth]{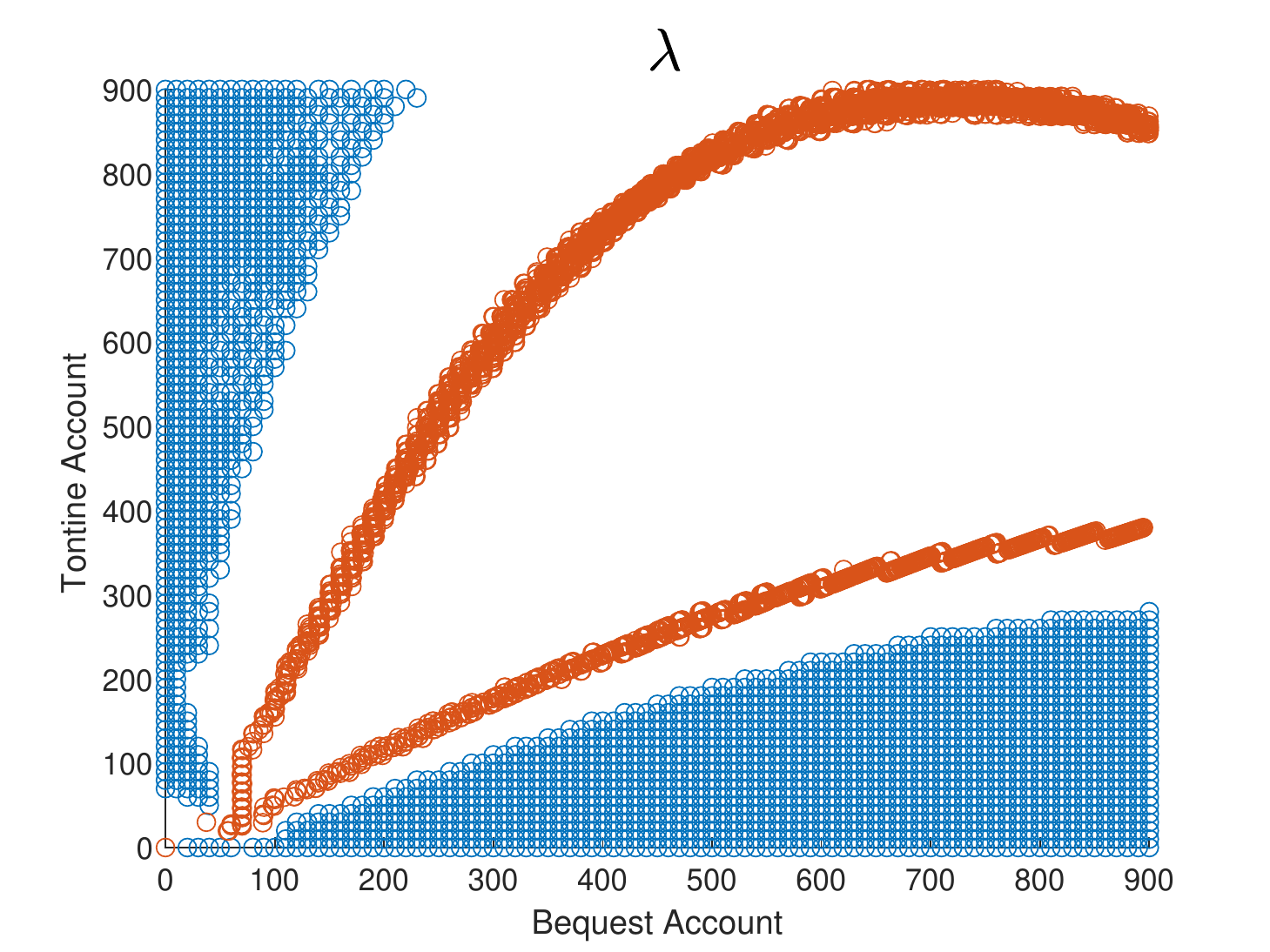}
	\end{subfigure}
	\begin{subfigure}[t]{0.22\textwidth}
		\centering
		\includegraphics[width=1\textwidth]{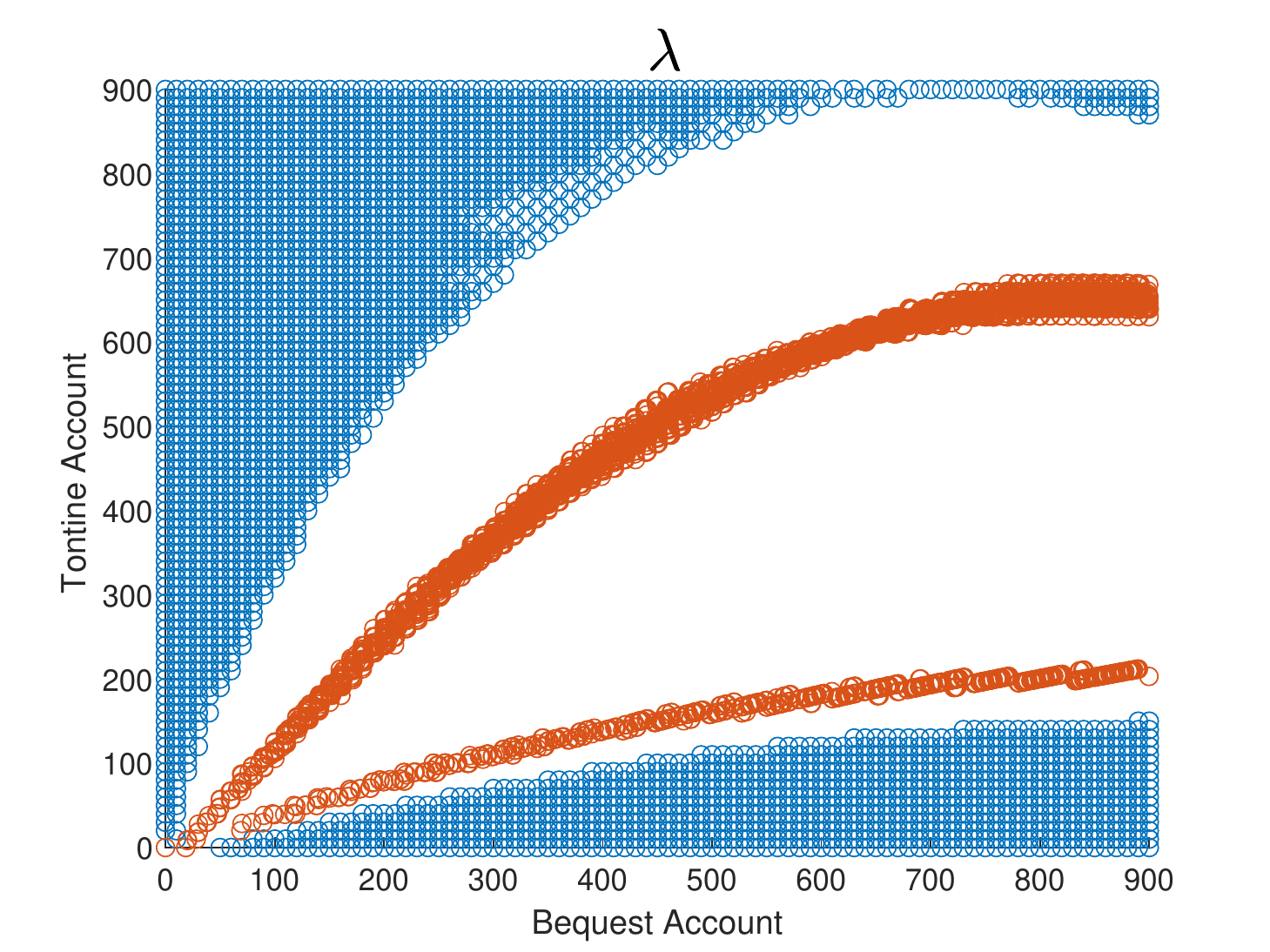}
	\end{subfigure}
	\begin{subfigure}[t]{0.22\textwidth}
		\centering
		\includegraphics[width=1\textwidth]{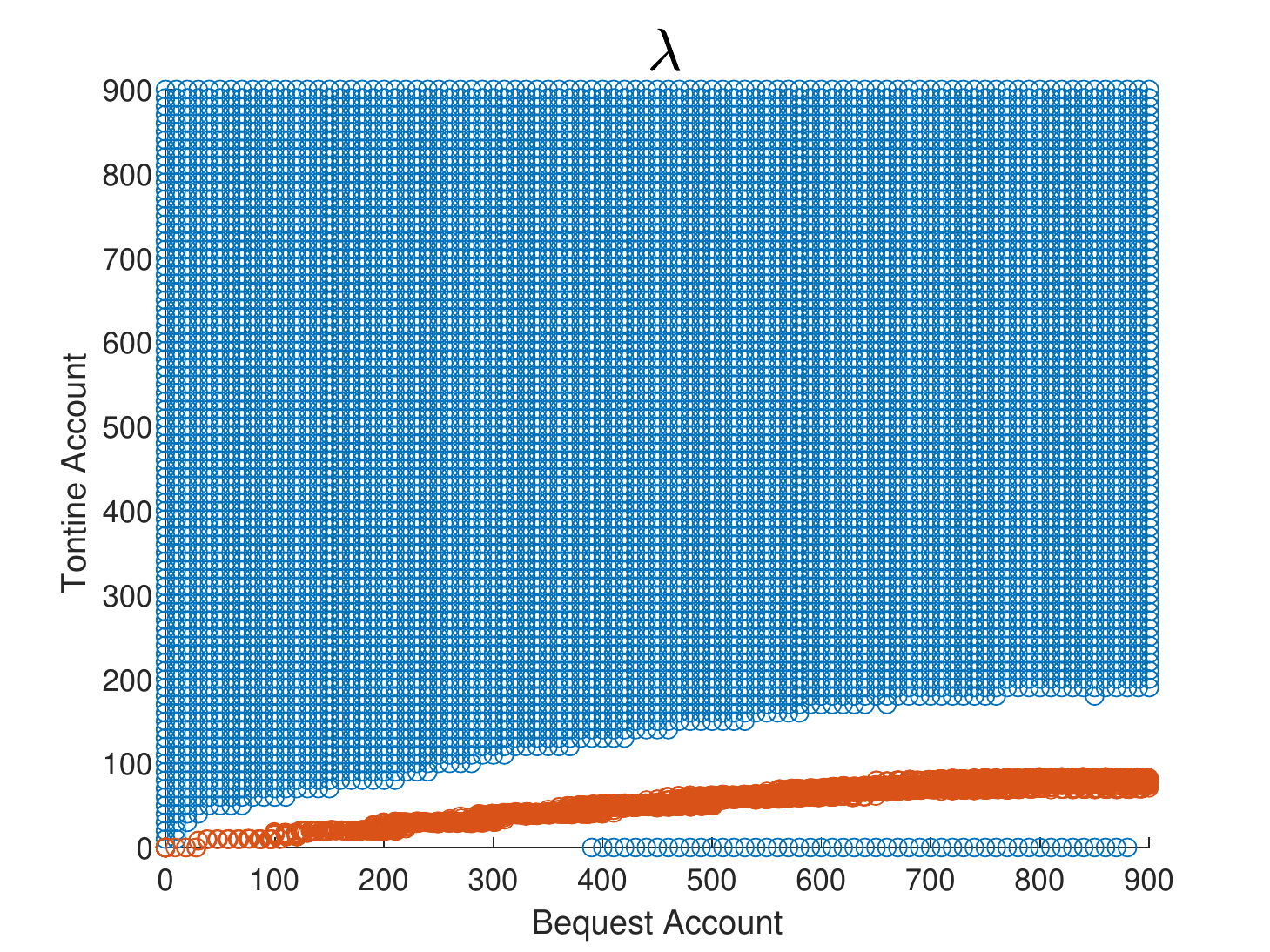}
	\end{subfigure}
	\begin{subfigure}[t]{0.22\textwidth}
		\centering
		\includegraphics[width=1\textwidth]{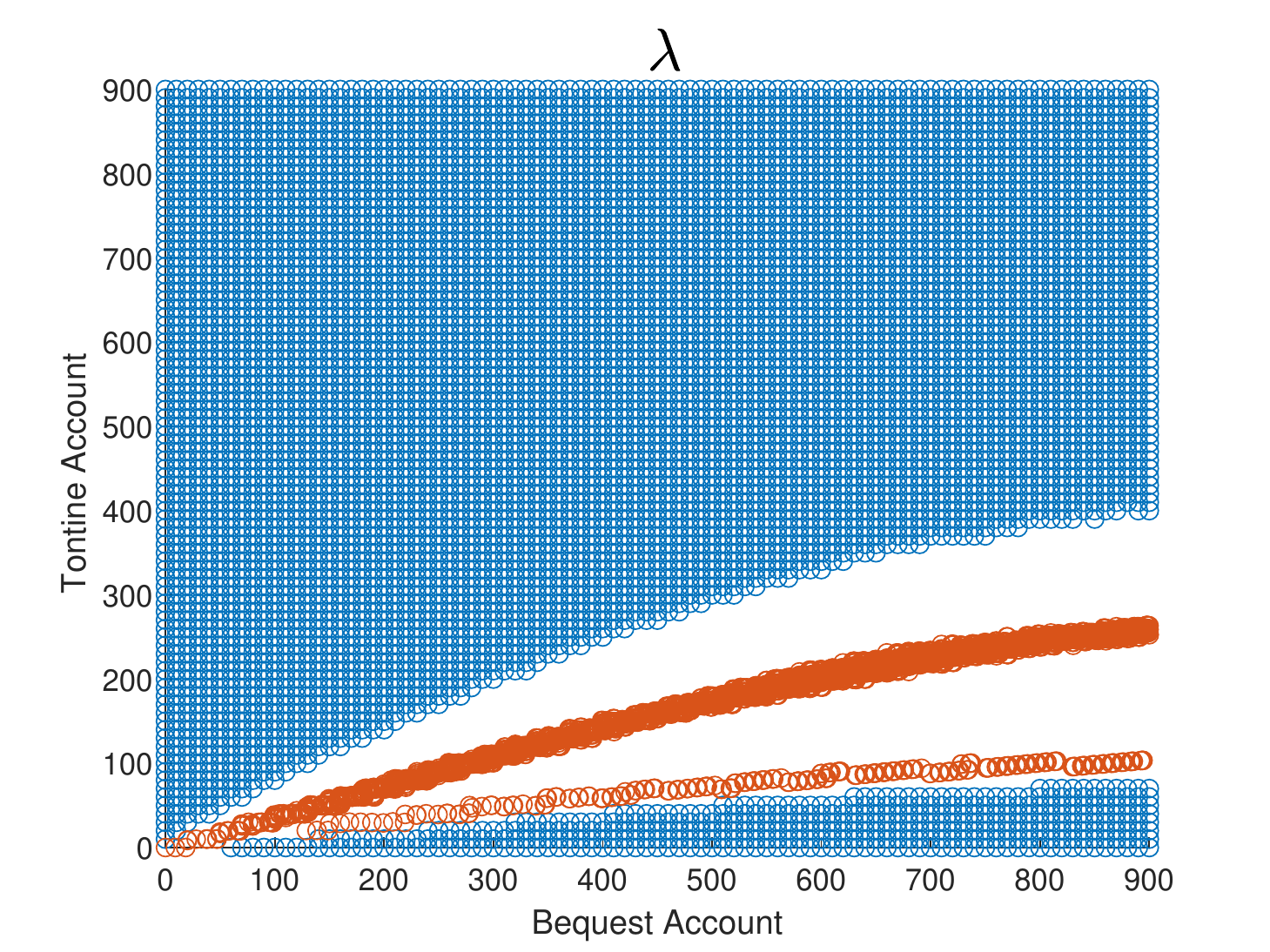}
	\end{subfigure}
	\begin{subfigure}[t]{0.22\textwidth}
		\centering
		\includegraphics[width=1\textwidth]{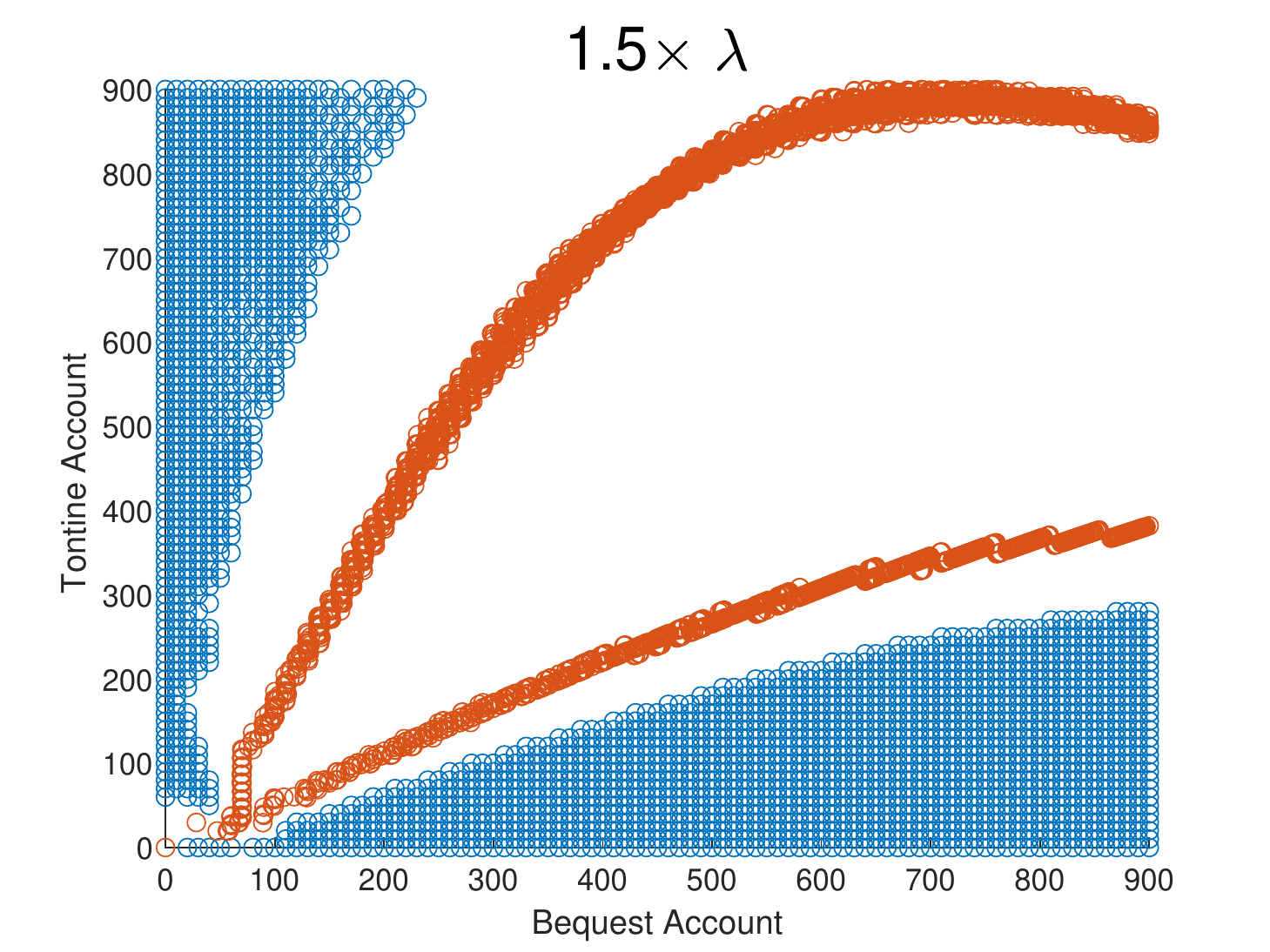}
		\subcaption*{65 years old}
	\end{subfigure}
	\begin{subfigure}[t]{0.22\textwidth}
		\centering
		\includegraphics[width=1\textwidth]{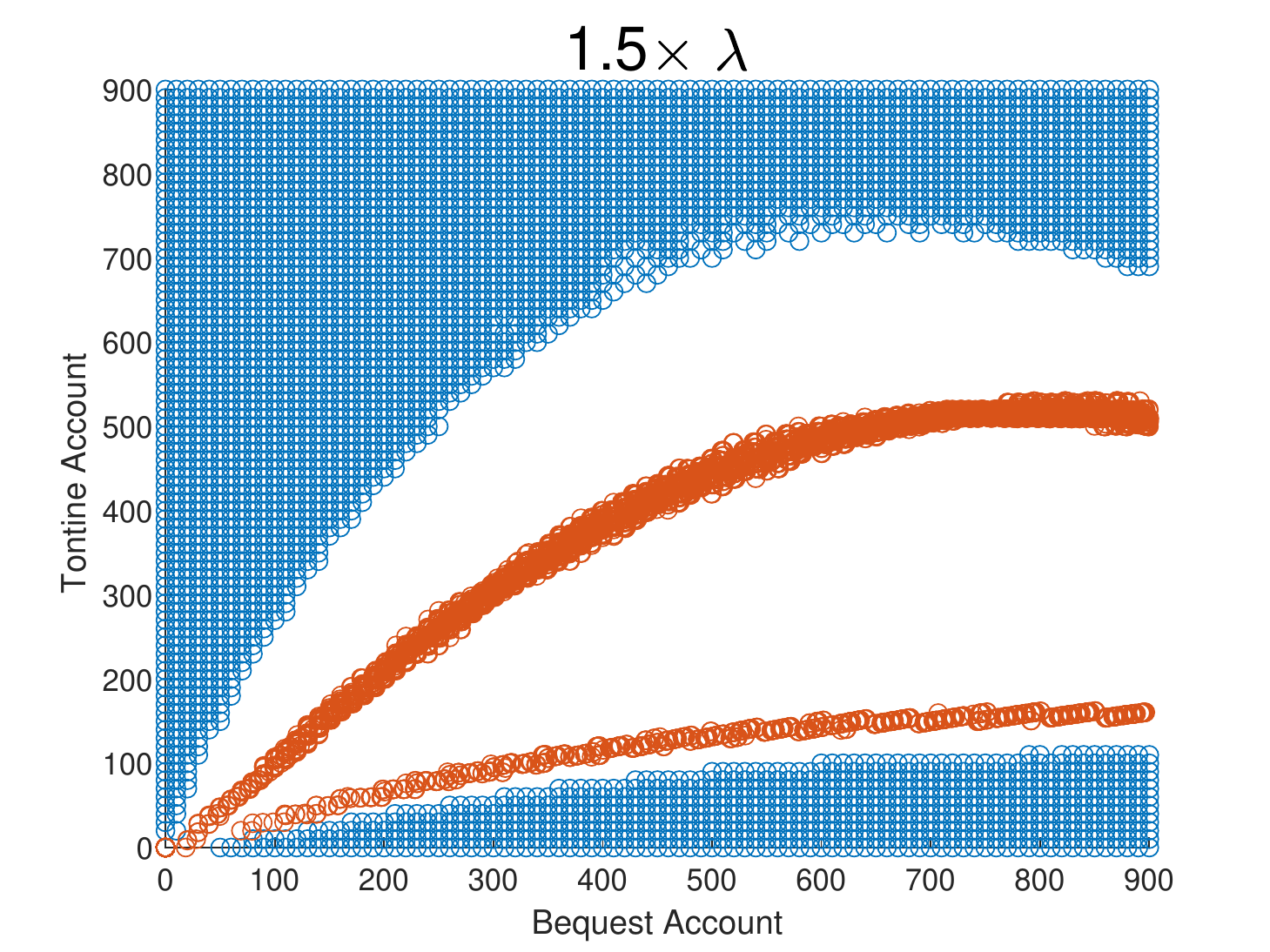}
		\subcaption*{80 years old}
	\end{subfigure}
	\begin{subfigure}[t]{0.22\textwidth}
		\centering
		\includegraphics[width=1\textwidth]{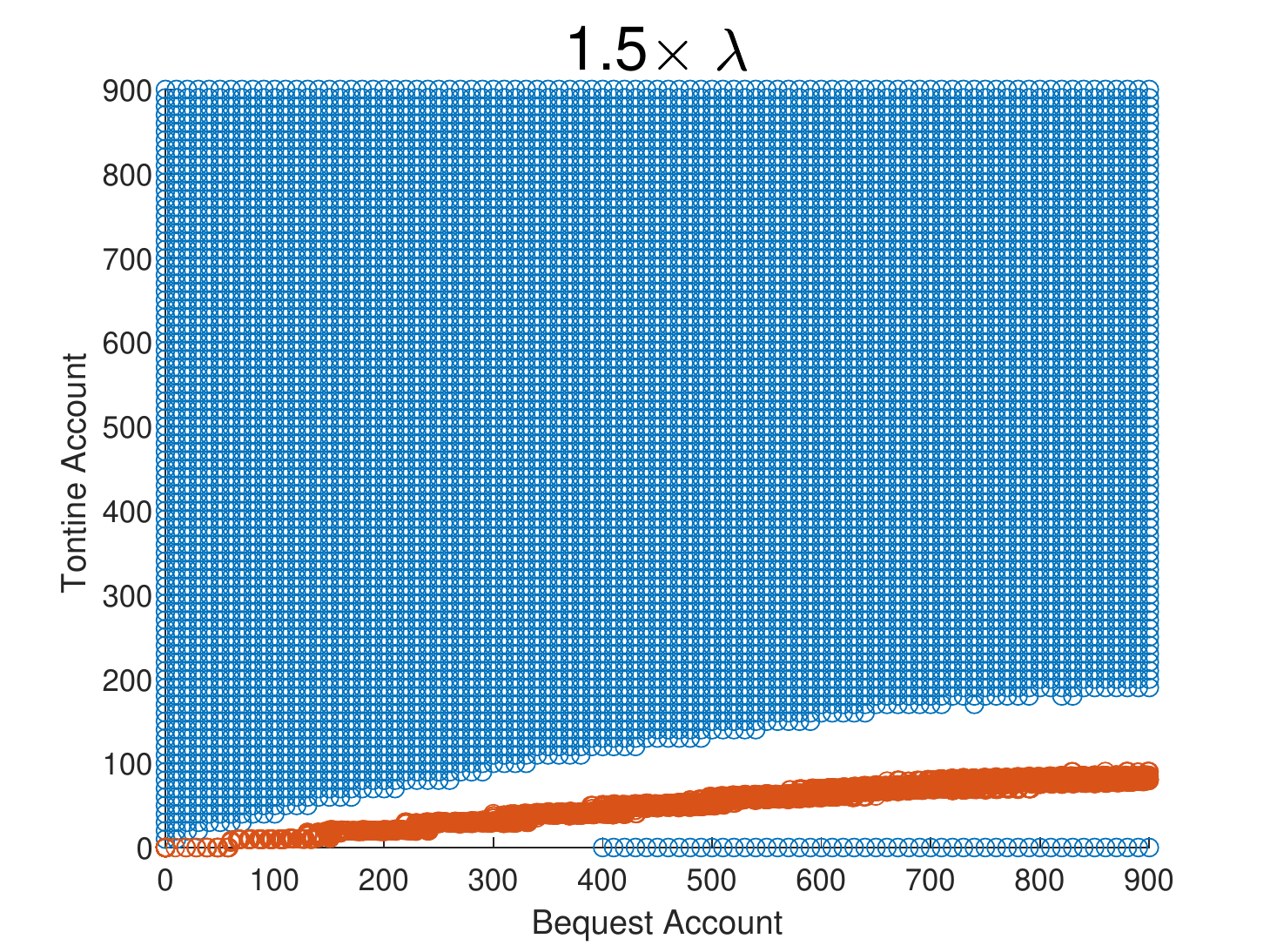}
		\subcaption*{95 years old}
	\end{subfigure}
	\begin{subfigure}[t]{0.22\textwidth}
		\centering
		\includegraphics[width=1\textwidth]{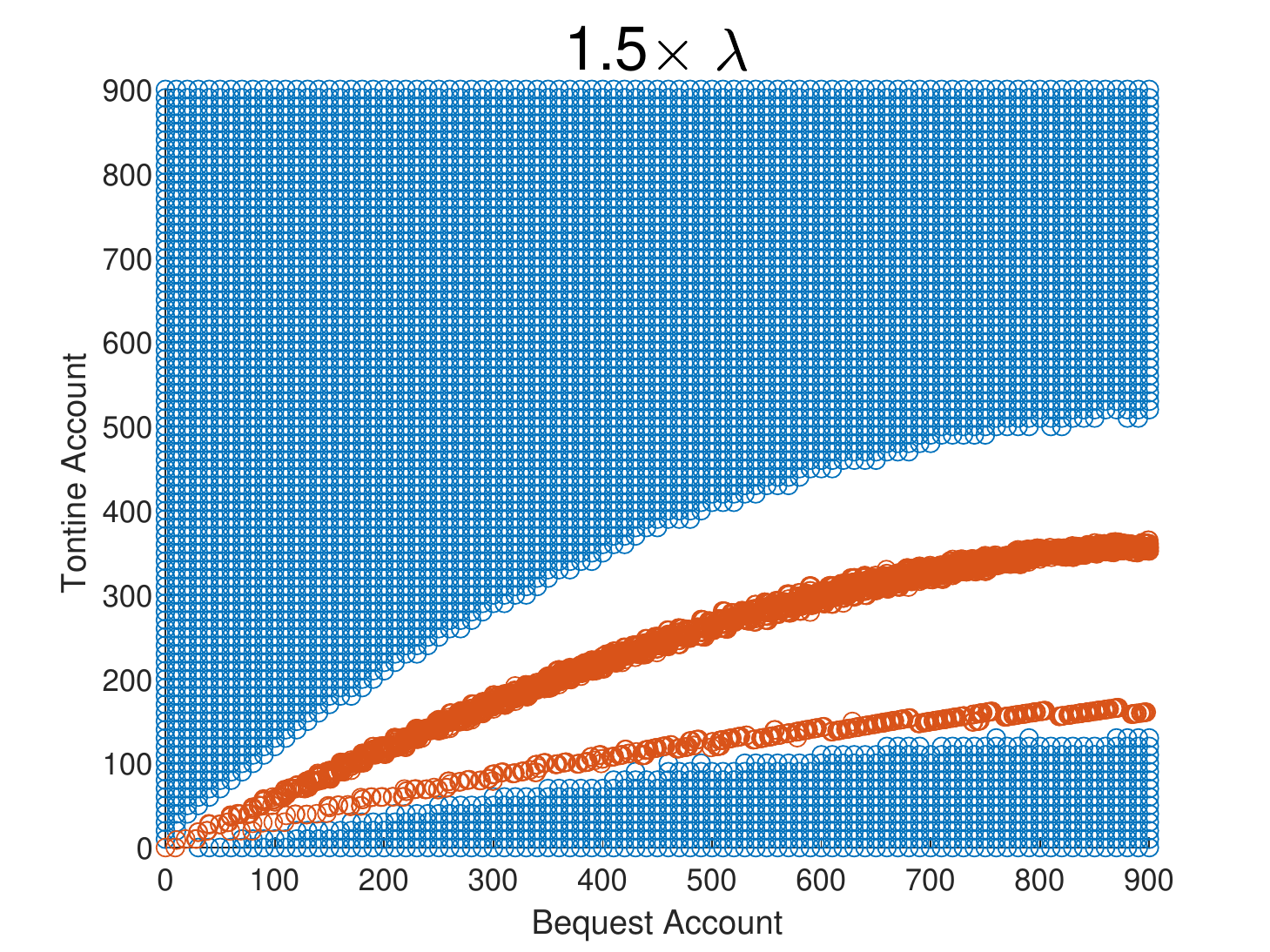}
		\subcaption*{110 years old}
	\end{subfigure}
	\caption{The impacts of force of mortality on transaction regions}
	\label{figure_lambda}
\end{figure}

\subsection{The Impacts of the Cost Parameter on Transaction Regions}
In this subsection, we study the impacts of the cost parameter on the transaction regions, and we take the proportional cost for an example.

In Fig. \ref{figure_xi}, the results show that when the proportional cost becomes larger, the No Transaction region slightly enlarges. That is, the Sell and Buy tontine regions both shrink. Naturally, the retiree will decrease the transaction frequency to reduce the transaction cost. Thus, the transaction only occurs when the wealth gap between the two accounts is relatively large. Moreover, the impacts of the cost parameter on the transaction policies are very small within its reasonable value range.
%
%
%
\subsection{The Impacts of the Expected Return and Proportional Transaction Cost on the Utility Improvement Ratio}
In this subsection, we study the utility improvement brought by participation in the tontine account. First, we compute the value function of the retiree when she/he cannot participate in the tontine account. The retiree's wealth process satisfies
\begin{equation*}
	X(s)=x+\int\limits_t^s[rX(u)-c(u)]\mathrm{d}u,
\end{equation*}
where $c(s)\geq 0$ represents the consumption rate at time $s$. Then, the value function can be defined as follows:
\begin{equation*}
	\tilde{V}(t,x)=\sup_{c(\cdot)}\mathbb{E}\!\left[\int\limits_t^{+\infty}\!\!\!\mathrm{e}^{-\int_t^s(\lambda(u)+\rho)\mathrm{d}u}\left\{U(c(s))\!+\!b\lambda(s)U(X(s))\right\}\mathrm{d}s\right].
\end{equation*}
We can easily obtain  $\tilde{V}(t,x)=h(t)x^p$, where $\tilde{h}(t)=D(0,t)h(t)$ satisfies the following ordinary differential equations
\begin{align*}
	\left \{
	\begin{array}{ll}
		\tilde{h}'(t)+(1-p)(\frac{D(0,t)}{p})^{\frac{1}{1-p}}\tilde{h}(t)^{\frac{p}{p-1}}+\frac{b\lambda(t)D(0,t)}{p}=0,\nonumber\\\\
		\lim\limits_{t\to +\infty}\tilde{h}(t)=0.
	\end{array}
	\right.
\end{align*}
Moreover, we define the utility improvement ratio as $\frac{V(0,x,0)-\tilde{V}(0,x)}{\tilde{V}(0,x)}$.
\begin{figure}[h]
	\centering
	\begin{subfigure}[t]{0.4\textwidth}
		\centering
		\includegraphics[width=1\textwidth]{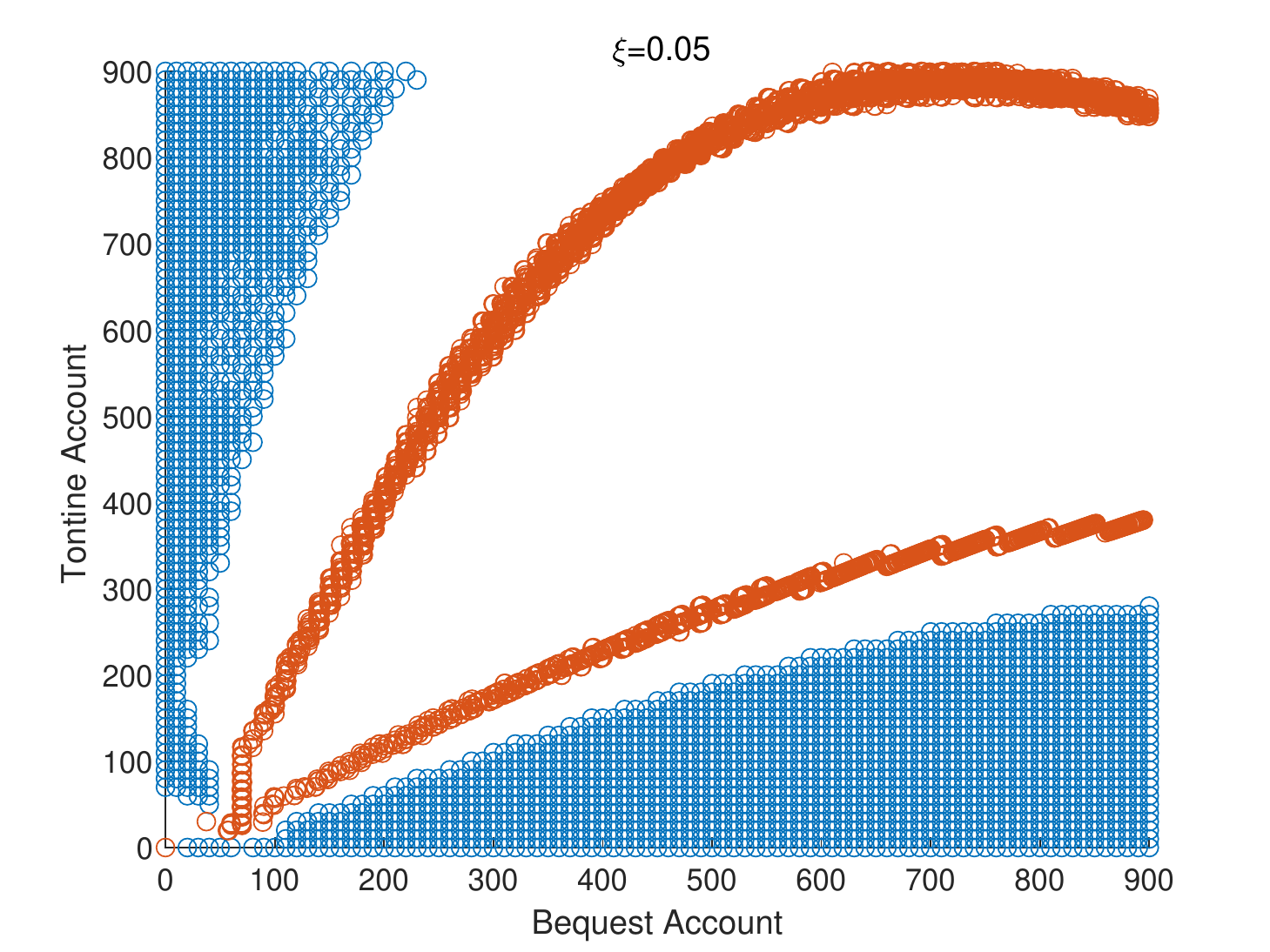}
	\end{subfigure}
	\begin{subfigure}[t]{0.4\textwidth}
		\centering
		\includegraphics[width=1\textwidth]{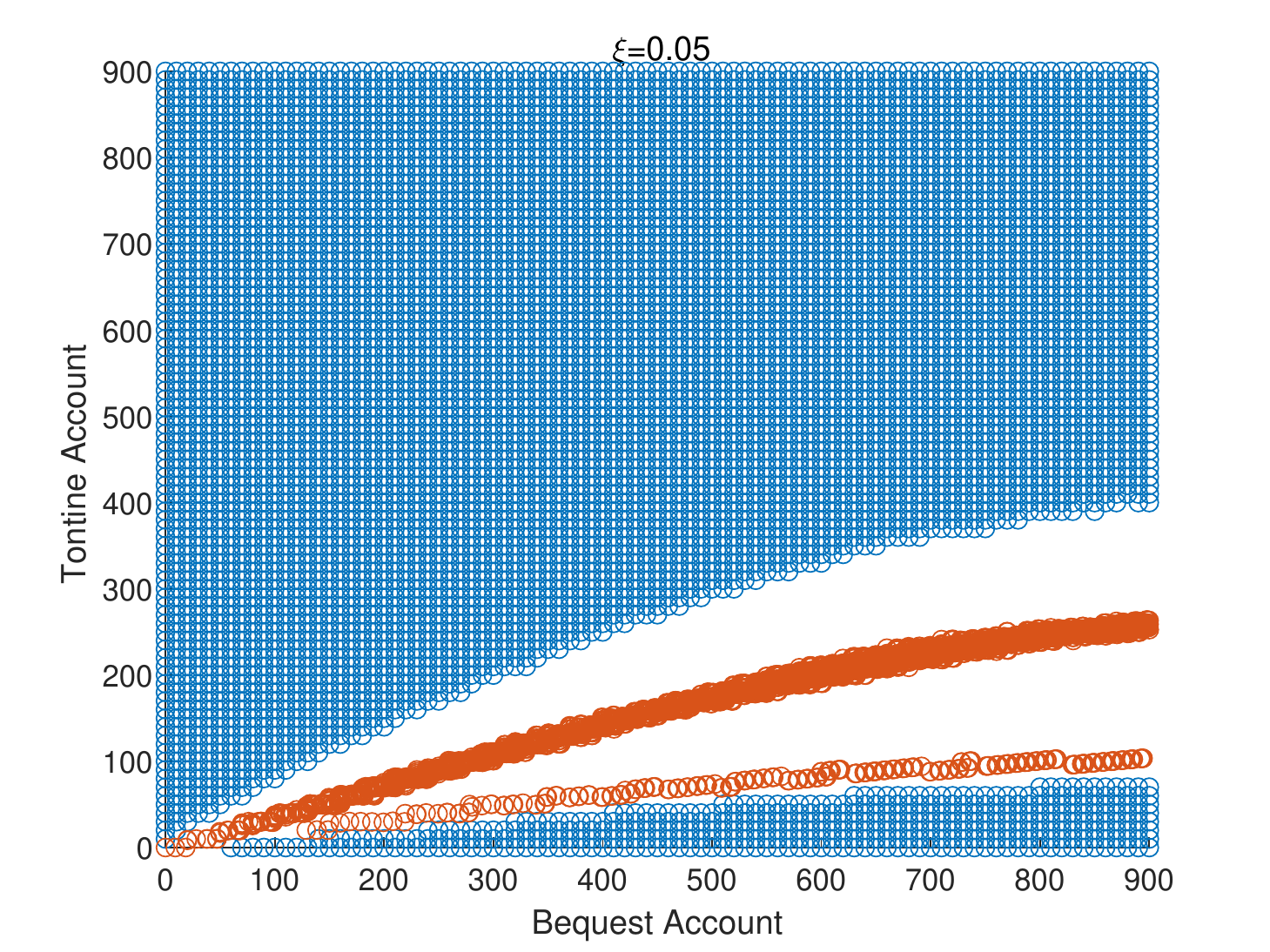}
	\end{subfigure}
	
	\begin{subfigure}[t]{0.4\textwidth}
		\centering
		\includegraphics[width=1\textwidth]{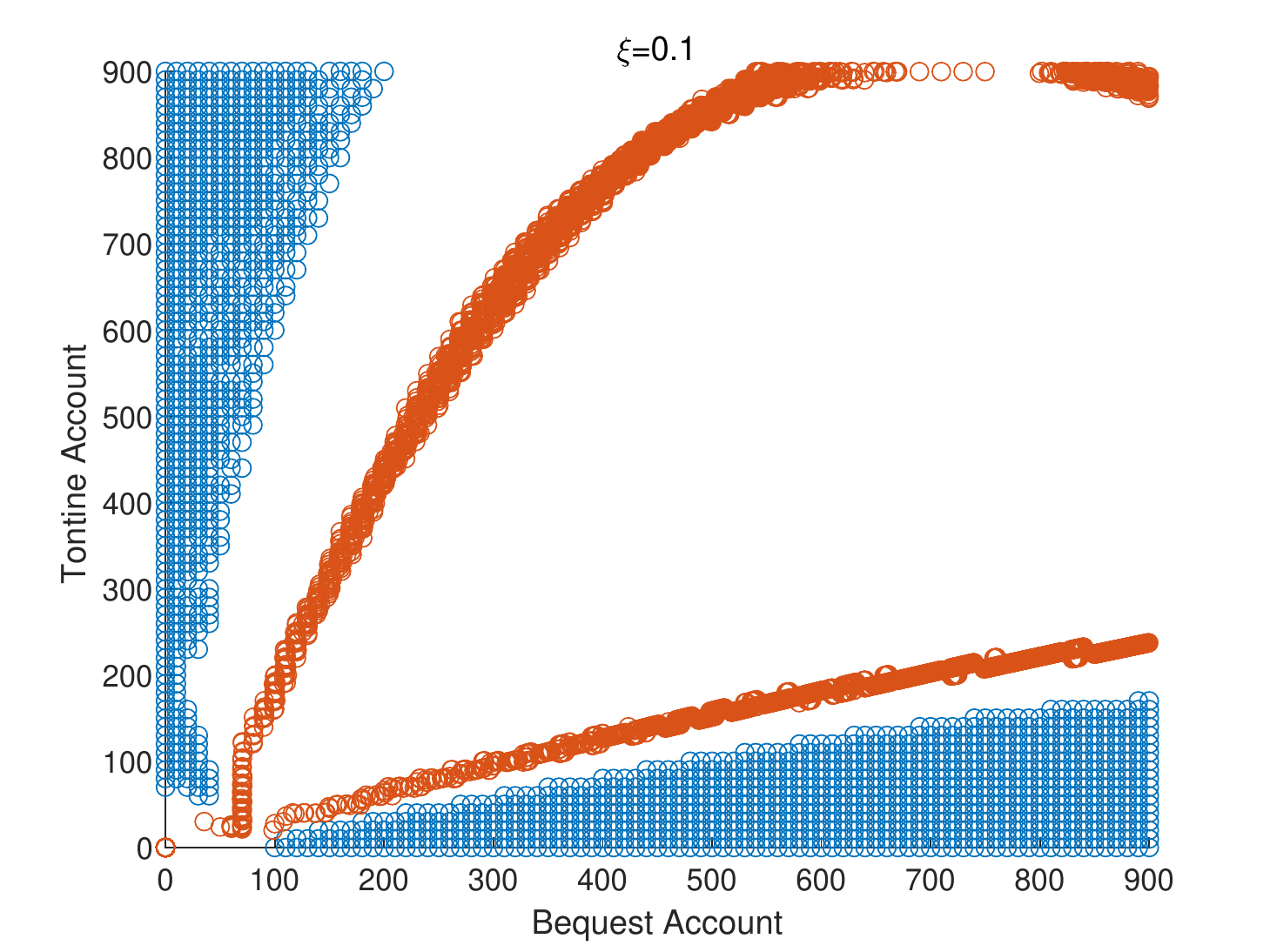}
		\caption*{65 years old}
	\end{subfigure}
	\begin{subfigure}[t]{0.4\textwidth}
		\centering
		\includegraphics[width=1\textwidth]{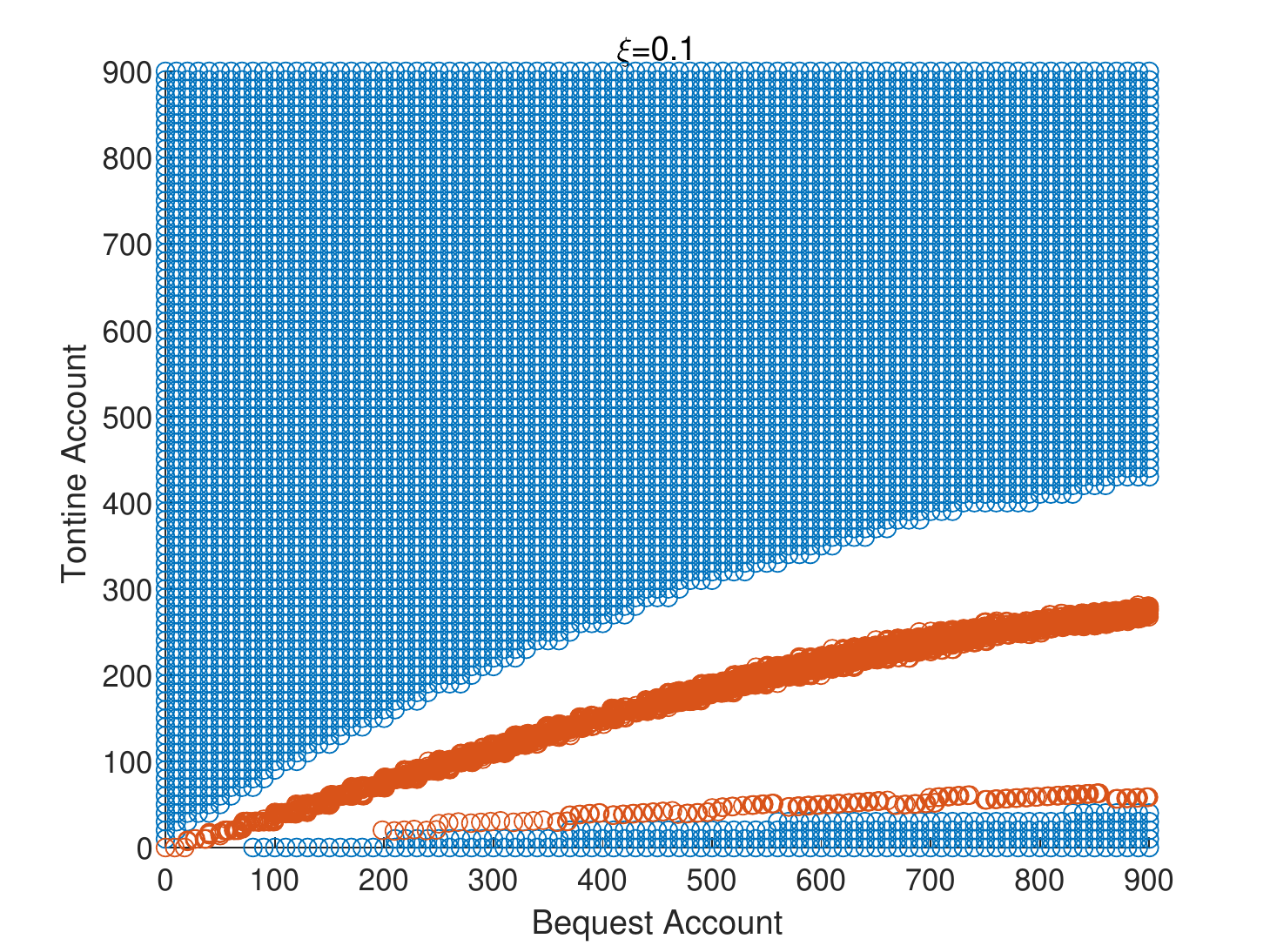}
		\caption*{110 years old}
	\end{subfigure}
	
	\caption{The impacts of the proportional cost parameter on transaction regions}  \label{figure_xi}
\end{figure}

\begin{figure}[h]
	\centering
	\includegraphics[width=0.7\textwidth]{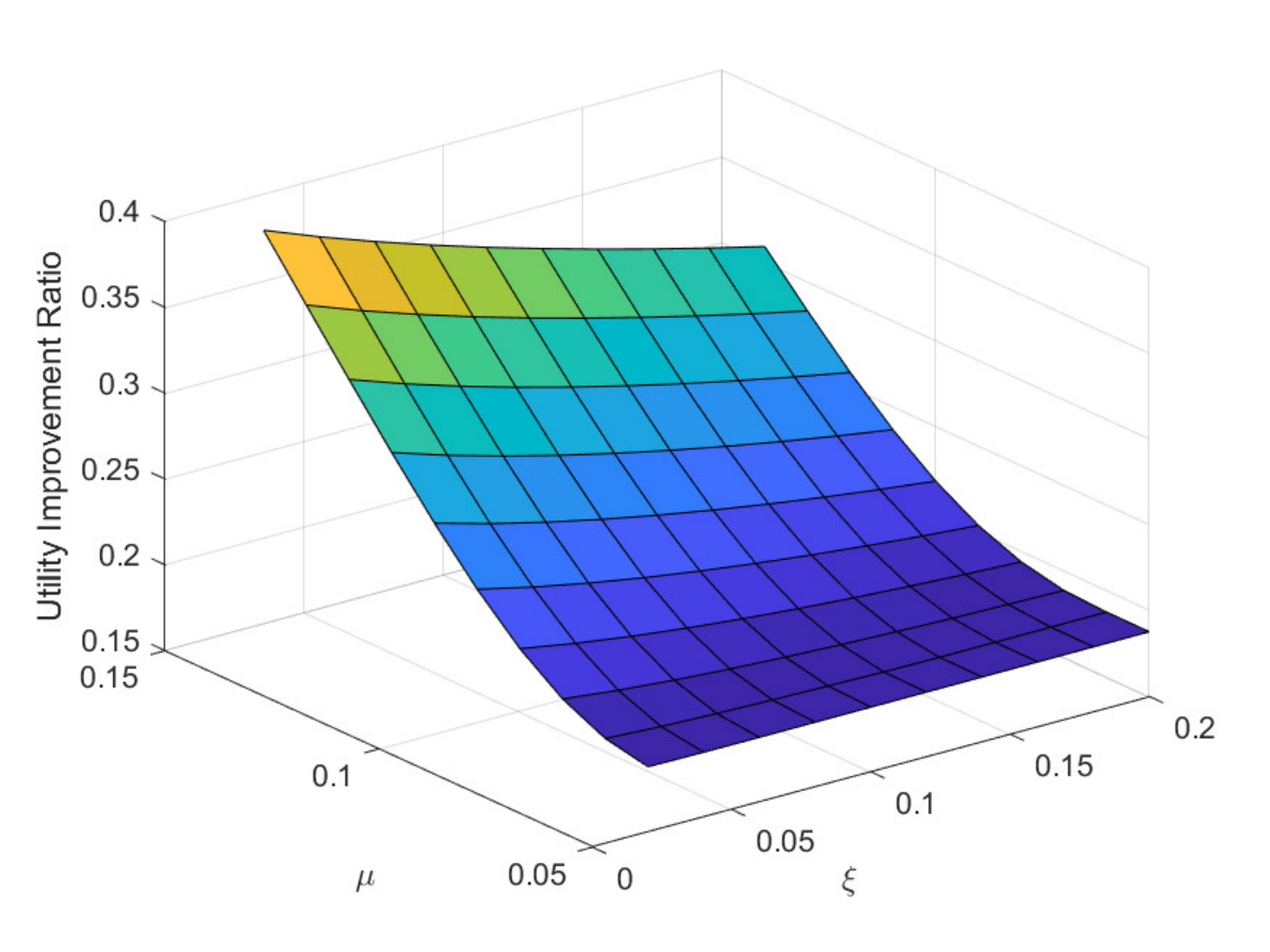}
	\caption{The impacts of expected return and transaction cost on utility improvement ratio}  \label{ratio}
\end{figure}
Fig. \ref{ratio} exhibits the impacts of the expected return of the tontine account  and proportional transaction cost on the utility improvement ratio when $x=600$.
The results show that the ratio is positively correlated with the expected return and negatively correlated with the proportional transaction cost, which is consistent with common sense. Moreover, when the expected return is small, this ratio is hardly affected by the changes of the transaction cost. In this circumstance, the tontine account is less attractive and the retiree will reduce  the tontine participation rate  as well as the transaction frequency. In the extreme case, when the expected return is $\mu=r=0.05$ and the proportional transaction cost is  $\xi=20\%$, the ratio is $18.77\%$. Even if the return rate of tontine account is not high and the transaction cost is large, the retirees will benefit considerably from the longevity credit by participating in the tontine account.


\vskip 15pt
\setcounter{equation}{0}
\section{{ {\bf Conclusions}}}\label{conclude}
\vskip 5pt
In this paper, we propose a new mechanism of modern tontine. The retiree dynamically allocates the wealth into the bequest account and the tontine account to achieve the consumption and bequest utilities. Particularly, each transaction between the two accounts incurs the transaction cost. The optimization problem is a combined stochastic and impulse control problem with infinite time horizon.  Using a WDPP, we prove that the value function is a viscosity solution of a HJBQVI.  Furthermore, the uniqueness of the HJBQVI and the continuity of the value function are characterized by a comparison theorem. The main findings of this paper are twofold. First, the V-shaped transaction region consists of two stages. The proportion allocated in the tontine account decreases in the former stage and increases in the latter stage. The gradual rise of the longevity credit first reduces the demand for risky tontine investment, and then increases the motivation of survival gambling. This fact confirms the rationalities we proposed. Second, in the latter stage, the transaction policies are hardly affected by the magnitude of bequest motive. Moreover, the risk averse attitude and the force of mortality rate become the decisive factors.

\vskip 15pt \setcounter{equation}{0}
\section*{Acknowledgments}
The authors acknowledge the
support from the National Natural Science Foundation of China (Nos.11871036,12271290).
The authors also thank the members of the group of
Actuarial Science and Mathematical Finance at the Department of
Mathematical Sciences, Tsinghua University for their feedbacks and
useful conversations.
\vskip 15pt \setcounter{equation}{0}

\vskip 15pt
\setcounter{equation}{0}
\appendix
\section{Details of the Proofs}
\subsection{Proof of Theorem \ref{WDPP}}\label{proof WDPP}
\vskip 5pt
Without loss of generality, we assume $\Omega=\{\omega|\omega\in C([0,\infty),\mathbb{R}),\omega(0)=0\}$. $\mathbb{P}$ is the Wiener measure, $W$ is the canonical Brownian motion,  and $\{\mathcal{F}_t,t\geq 0\}$ is the augment filtration of $W$.

The proof of \eqref{DPP1} is similar to \cite{AS2017}, we give its proof here for completeness.
Fixed $\nu=(c,\tau,\Delta)\in \mathcal{A}(t,x,y)$,  we denote  $(\theta^{\nu},c^{\nu}(s),X^{\nu}(s),Y^{\nu}(s))$ by $(\theta,c(s),X(s),Y(s))$ for  simplicity. Fixed $\omega\in \Omega$, we define $\nu^{\theta,\omega}(\omega',s)=\nu(\omega\overset{\theta}\oplus\omega',s)$ for $s\geq \theta(\omega)$, where
\begin{equation*}
	\omega\overset{\theta}\oplus\omega'=\begin{cases}\omega_z,&z\in[0,\theta(\omega)),\\
		\omega'_{z-\theta(\omega)}+\omega_{\theta(\omega)},&z\in[\theta(\omega),\infty).
	\end{cases}			
\end{equation*}
It is then clear that $\nu^{\theta,\omega}\in \mathcal{A}(\theta(\omega),X(\theta(\omega)-),Y(\theta(\omega)-))$, thus
\begin{align*}
	&	\mathbb{E}\left[\int\limits_t^{+\infty}D(t,s)\bar{U}(s,c(s),X(s))\mathrm{d}s\Bigg|\mathcal{F}_{\theta}\right](\omega)=\int\limits_t^{\theta(\omega)}D(t,s)\bar{U}(s,c(s),X(s))\mathrm{d}s\\
	&+D(t,\theta(\omega))\int_{\Omega}\int_{\theta(\omega)}^{\infty}D(\theta(\omega),s)\bar{U}(s,c^{\nu^{\theta,\omega}}(\omega',s),X^{\nu^{\theta,\omega}}(\omega',s))\mathrm{d}s\mathbb{P}(\mathrm{d}\omega')\\
	&\leq\int\limits_t^{\theta(\omega)}D(t,s)\bar{U}(s,c(s),X(s))\mathrm{d}s+D(t,\theta(\omega))V(\theta(\omega),X(\theta(\omega)-),Y(\theta(\omega)-))\\
	&\leq\int\limits_t^{\theta(\omega)}D(t,s)\bar{U}(s,c(s),X(s))\mathrm{d}s+D(t,\theta(\omega))V^*(\theta(\omega),X(\theta(\omega)-),Y(\theta(\omega)-)).
\end{align*}
Taking expectations on the two sides, and by the arbitrariness of $\nu$, we obtain \eqref{DPP1}.

To prove \eqref{DPP2}, we denote the second row of \eqref{DPP2} by $\underline{V}(t,x,y)$. For any $\epsilon>0$, we can choose a $\nu^{\epsilon}\in \mathcal{A}(t,x,y)$ such that
\begin{align}
	&\underline{V}(t,x,y)<\epsilon+\nonumber\\ &\mathbb{E}\left\{\int_t^{\theta^{\nu^{\epsilon}}}\!\!\!D(t,s)\bar{U}(s,c^{\nu^{\epsilon}}(s),X^{\nu^{\epsilon}}(s))\mathrm{d}s+D(t,\theta^{\nu^{\epsilon}})\varphi(\theta^{\nu^{\epsilon}},X^{\nu^{\epsilon}}(\theta^{\nu^{\epsilon}}),Y^{\nu^{\epsilon}}(\theta^{\nu^{\epsilon}}))\right\}.\label{condition1}
\end{align}
We assume that $\theta^{\epsilon}$ takes only countable values, i.e., $\theta^{\epsilon}(\omega)\in\{t_m,m\in \mathbb{N}\}$.
We proceed to find a countable cover of $\{t_m\}\times \mathbb{R}^2_+$ for some fixed $m\in\mathbb{N} $, and divide the cover into several cases:
\begin{itemize}
	\item $S_1=\{t_m\}\times(0,\infty)^2$
	
	For any $\zeta=(t_m,\bar{x},\bar{y})\in S_1$, we define a subset of $S_1$ as follows:
	\begin{equation*}
		\mathcal{R}^1_{\epsilon}(\zeta)=\{(t_m,x',y')\in S_1|x'>\bar{x},y'>\bar{y},\varphi(t_m,x',y')<\varphi(t_m,\bar{x},\bar{y})+\epsilon\}.
	\end{equation*}
	Since $\varphi$ is continuously differential,  $\mathcal{R}^1_{\epsilon}(\zeta)$ is relatively open in $S_1$ and  \\$\cup_{\zeta\in S_1}\mathcal{R}^1_{\epsilon}(\zeta)$ is an open cover of $S_1$. Hence, we can extract a countable subcover
	\begin{equation*}
		S_1=\cup_{n=1}^{\infty}\mathcal{R}^1_{\epsilon}(\zeta^1_n) ,
	\end{equation*}
	where $\zeta^1_n\in S_1$ is different from each other.
	\item$S_2=\{t_m\}\times(0,\infty)\times\{0\}$
	
	For any $\zeta=(t_m,\bar{x},0)\in S_2$, we define a subset of $S_2$ as follows:
	\begin{equation*}
		\mathcal{R}^2_{\epsilon}(\zeta)=\{(t_m,x',0)\in S_2|x'>\bar{x},\varphi(t_m,x',0)<\varphi(t_m,\bar{x},0)+\epsilon\}.
	\end{equation*}
	Similarly, we can extract a countable subcover
	\begin{equation*}
		S_2=\cup_{n=1}^{\infty}\mathcal{R}^2_{\epsilon}(\zeta^2_n) ,
	\end{equation*}
	where $\zeta^2_n\in S_2$ is different from each other.
	\item$S_3=\{t_m\}\times\{0\}\times(0,\infty)$
	
	For any $\zeta=(t_m,0,\bar{y})\in S_3$, we define a subset of $S_3$ as follows:
	\begin{equation*}
		\mathcal{R}^3_{\epsilon}(\zeta)=\{(t_m,0,y')\in S_3|y'>\bar{y},\varphi(t_m,0,y')<\varphi(t_m,0,\bar{y})+\epsilon\}.
	\end{equation*}
	Similarly, we can extract a countable subcover
	\begin{equation*}
		S_3=\cup_{n=1}^{\infty}\mathcal{R}^3_{\epsilon}(\zeta^3_n) ,
	\end{equation*}
	where $\zeta^3_n\in S_3$ is different from each other.
	\item$S_4=(t_m,0,0)$
	
	It	is only one point. We denote it by 	
	\begin{equation*}
		S_4=\mathcal{R}^4_{\epsilon}(\zeta^4_1).
	\end{equation*}
\end{itemize}

Considering all the $\{t_m\}$  and  rearranging the collections, we obtain a countable cover of $\{t_m,m\in\mathbb{N}\}\times \mathbb{R}^2_+=\cup_{n=1}^{\infty}\mathcal{R}_{\epsilon}(\zeta_n)$. To get a disjoint union, we define
\begin{equation*}
	A_1=\mathcal{R}_{\epsilon}(\zeta_1),A_2=\mathcal{R}_{\epsilon}(\zeta_2)\setminus A_1,
\end{equation*}
and iteratively define
\begin{equation*}
	A_n=\mathcal{R}_{\epsilon}(\zeta_n)\setminus (\cup_{k=1}^{n-1}A_k).
\end{equation*}
Then, $\{A_n\}_{n\in \mathbb{N}}$ satisfies:
\begin{itemize}
	\item $A_n$ is disjoint with each other and $\{t_m,m\in\mathbb{N}\}\times \mathbb{R}^2_+=\cup_{n=1}^{\infty}A_n$.
	\item For $\zeta_n=(\bar{t}_n,\bar{x}_n,\bar{y}_n)$ \footnote{$\zeta_n$ may be not in $A_n$.} and $\eta=(t',x',y')\in A_n$, we have $t'= \bar{t}_n $, $x'\geq\bar{x}_n$, $y'\geq\bar{y}_n$ and  $\varphi(\eta)\leq\varphi(\zeta_n)+\epsilon$.
	\item If $\nu_n\in \mathcal{A}(\zeta_n)$,
	then for all $\eta\in A_n$, we have $\nu_n\in \mathcal{A}(\eta)$.
\end{itemize}
By the definition of $V(\zeta_n)$, we can choose $\nu_n\in \mathcal{A}(\zeta_n)$ such that
\begin{equation}
	\varphi(\zeta_n)\leq V(\zeta_n)\leq \epsilon+
	\mathbb{E}\int\limits_{\bar{t}_n}^{+\infty}D(\bar{t}_n,s)\bar{U}(s,c^{\nu_n}(s),X^{\nu_n}(s))\mathrm{d}s.\label{c2}
\end{equation}
Then for all $\eta=(\bar{t}_n,x',y')\in A_n$, $\nu_n\in\mathcal{A}(\eta)$ and
\begin{equation}
	\varphi(\eta)-\epsilon \leq \epsilon+\mathbb{E}
	\int\limits_{\bar{t}_n}^{+\infty}D(\bar{t}_n,s)\bar{U}(s,c^{\nu_n}(s),X^{\nu_n}(s))\mathrm{d}s.\label{c1}
\end{equation}
Note that we generalize the symbol here. $X^{\nu_n}(\cdot)$ in the right hand of \eqref{c2} and \eqref{c1} have initial values  $\bar{x}_n$ and $x'$, respectively.
For $(t',x',y')\in\{t_m,m\in\mathbb{ N}\}\times\mathbb{R}^2_+)\cap A_n$, we define $I(t',x',y')=n$.  For $\omega\in \Omega$, we  denote $N(\omega)=I(\theta^{\nu^{\epsilon}}(\omega),X^{\nu^{\epsilon}}(\theta^{\nu^{\epsilon}}(\omega)),\\Y^{\nu^{\epsilon}}(\theta^{\nu^{\epsilon}}(\omega)))$. Moreover, we establish the following policy ($\theta=\theta^{\nu^{\epsilon}}$):
\begin{equation*}
	\nu^*(\omega\overset{\theta}\oplus\omega',s)=\begin{cases}
		v^{\epsilon}(\omega,s), &s\in [t,\theta(\omega)],\\
		v^{N(w)}(\omega',s),&s\in[\theta(\omega),\infty).
	\end{cases}
\end{equation*}
It follows that
\begin{align*}
	&\mathbb{E}\left[ \int\limits_{{t}}^{+\infty}D(t,s)\bar{U}(s,c^{\nu^{*}}(s),X^{\nu^{*}}(s))\mathrm{d}s\Bigg|\mathcal{F}_{\theta}\right](\omega)\\
	&= \int\limits_{{t}}^{\theta(\omega)}D(t,s)\bar{U}(s,c^{\nu^{\epsilon}}(s),X^{\nu^{\epsilon}}(s))\mathrm{d}s\\
	&+D(t,\theta(\omega))\int_{\Omega}\int_{\theta(\omega)}^{\infty}D(\theta(\omega),s)\bar{U}(s,c^{\nu^{N(\omega)}}(\omega',s),X^{\nu^{N(\omega)}}(\omega',s))\mathrm{d}s\mathbb{P}(\mathrm{d}\omega')
	\\
	&\geq \!\!\!\! \int\limits_{{t}}^{\theta(\omega)}D(t,s)\bar{U}(s,c^{\nu^{\epsilon}}(s),X^{\nu^{\epsilon}}(s))\mathrm{d}s+D(t,\theta(\omega))(\varphi(\theta(\omega),X^{\nu^{\epsilon}}(\theta(\omega)),Y^{\nu^{\epsilon}}(\theta(\omega)))-2\epsilon),
\end{align*}
where we employ the definition of $v^{*}$ and \eqref{c1}, respectively. Taking the expectation and using \eqref{condition1}, we obtain
\begin{equation*}
	\mathbb{E} \int\limits_{{t}}^{+\infty}D(t,s)\bar{U}(s,c^{\nu^*}(s),X^{\nu^*}(s))\mathrm{d}s\geq\underline{V}(t,x,y)-3\epsilon.
\end{equation*}
Since $\epsilon $ is arbitrary, \eqref{DPP2} is valid. Finally, if $\theta^{\nu^{\epsilon}}$ takes uncountable values, we choose $\theta_n$ which takes countable values and monotonically decreases to $\theta^{\nu^{\epsilon}}$. Then, using dominated convergence theorem and Fatou's lemma, we have
\begin{align*}
	&\varliminf_{n\to\infty}\mathbb{E}\left\{\int_t^{\theta_n}\!\!\!D(t,s)\bar{U}(s,c^{\nu^{\epsilon}}(s),X^{\nu^{\epsilon}}(s))\mathrm{d}s+D(t,\theta_n)\varphi(\theta_n,X^{\nu^{\epsilon}}(\theta_n),Y^{\nu^{\epsilon}}(\theta_n))\right\}	\\
	&=
	\mathbb{E}\left\{\int_t^{\theta^{\nu^{\epsilon}}}\!\!\!D(t,s)\bar{U}(s,c^{\nu^{\epsilon}}(s),X^{\nu^{\epsilon}}(s))\mathrm{d}s\right\}\!+\!\varliminf_{n\to\infty}\mathbb{E}[D(t,\theta_n)\varphi(\theta_n,X^{\nu^{\epsilon}}(\theta_n),Y^{\nu^{\epsilon}}(\theta_n))]\\
	&\geq
	\mathbb{E}\left\{\int_t^{\theta^{\nu^{\epsilon}}}\!\!\!D(t,s)\bar{U}(s,c^{\nu^{\epsilon}}(s),X^{\nu^{\epsilon}}(s))\mathrm{d}s\right\}\!+\!\mathbb{E}\varliminf_{n\to\infty}[D(t,\theta_n)\varphi(\theta_n,X^{\nu^{\epsilon}}(\theta_n),Y^{\nu^{\epsilon}}(\theta_n))]\\
	&= \mathbb{E}\left\{\int_t^{\theta^{\nu^{\epsilon}}}\!\!\!D(t,s)\bar{U}(s,c^{\nu^{\epsilon}}(s),X^{\nu^{\epsilon}}(s))\mathrm{d}s+D(t,\theta^{\nu^{\epsilon}})\varphi(\theta^{\nu^{\epsilon}},X^{\nu^{\epsilon}}(\theta^{\nu^{\epsilon}}),Y^{\nu^{\epsilon}}(\theta^{\nu^{\epsilon}}))\right\}.
\end{align*}
Therefore, without loss of generality, we can replace $\theta^{\nu^{\epsilon}}$ by some $\theta_N$ which takes only countable values.

\subsection{Proof of Theorem \ref{vis char}} \label{proof of vis char }
First, we  prove that $V_*$ is a viscosity supersolution of $\eqref{HJBQVI}$, and the proof is similar to the proof of Proposition 4.3.1 in \cite{Pham2009}. We take $(t,x,y)\in \mathbb{R}_+\times (\mathbb{R}^2_+\setminus{(0,0)})$ and $\varphi\in C^{1,2}(\mathbb{R}_+\times (\mathbb{R}^2_+\setminus{(0,0)}))$ such that $\varphi-V_*$ attains maximum $0$ at $(t,x,y)$. Using Point 9 of Proposition \ref{VM}, we have \begin{equation*}
	V_*(t,x,y)\geq \mathcal{M}[V_*](t,x,y)\geq \mathcal{M}[V_*]_*(t,x,y).
\end{equation*}
Thus, for $x\neq 0$, it is enough to prove
\begin{equation*}
	\mathcal{L}[\varphi](t,x,y)-b\lambda(t)U(x)-f(\varphi_x(t,x,y))\geq 0,
\end{equation*}
or equivalently
\begin{equation*}
	\mathcal{L}[\varphi](t,x,y)-b\lambda(t)U(x)-(U(c)-c\varphi_x(t,x,y))\geq 0,
\end{equation*}
for all $c\geq 0$. We choose $(t_n,x_n,y_n)\in \mathbb{R}_+\times (\mathbb{R}^2_+\setminus{(0,0)})$ such that $\lim\limits_{n\to \infty}(t_n,x_n,y_n,\\V(t_n,x_n,y_n))=(t,x,y,V_*(t,x,y))$ and $0<\frac{x}{2}<x_n$ for all $n$. Fix $c\geq 0$, processes $\{X^n(s)\}_{s\geq t_n}$ and $\{Y^n(s)\}_{s\geq t_n}$ are defined by
\begin{align*}
	&X^n(s)=x_n+\int_{t_n}^s(rX^n(u)-c)\mathrm{d}u,\\
	&Y^n(s)=y_n+\int_{t_n}^s[\mu+\lambda(u)]Y^n(u)\mathrm{d}u+\int_{t_n}^{s}\sigma Y^n(u)\mathrm{d}W(u),
\end{align*}
or equivalently
\begin{align*}
	&X^n(s)=\mathrm{e}^{r(s-t_n)}x_n-c\mathrm{e}^{rs}\int_{t_n}^s\mathrm{e}^{-ru}\mathrm{d}u,\\
	&Y^n(s)=y_n\mathrm{e}^{\int_{t_n}^s[\mu-\frac{1}{2}\sigma^2+\lambda(u)]\mathrm{d}u+\sigma[W(s)-W(t_n)]}.
\end{align*}
Now we define the stopping time $\theta_n$ as follows:
\begin{equation*}
	\theta_n=\inf\{s\geq t_n|(s,X^n(s),Y^n(s))\notin B_{\delta}(t_n,x_n,y_n)\},
\end{equation*}
where $B_{\delta}(t_n,x_n,y_n)=\{(\bar{t},\bar{x},\bar{y})\in\mathbb{R}_+\times (\mathbb{R}^2_+\setminus{(0,0)})|\|(t_n,x_n,y_n)-(\bar{t},\bar{x},\bar{y})\|<\delta\}$ and $0<\delta<\frac{x}{4}$. Here and after $\|\cdot\|$ represents the $L^{\infty}-norm$, that is,  $\|({t},{x},{y})\|=\max\{|t|,|x|,|y|\}$.
We define $\alpha_n=V(t_n,x_n,y_n)-\varphi(t_n,x_n,y_n)\geq0$. According to the selection of $(t_n,x_n,y_n)$, we have $\lim\limits_{n\to \infty}\alpha_n=0$. We then choose $\beta_n>0$ such that $ \lim\limits_{n\to \infty}\beta_n=0$ and $\lim\limits_{n\to \infty}\alpha_n/\beta_n=0$. Moreover, we define the stopping time $\tau_n=\theta_n\wedge(t_n+\beta_n)$. According to the WDPP \eqref{DPP2}, we have
\begin{equation*}
	V(t_n,x_n,y_n)\geq\mathbb{E}\{\int_{t_n}^{\tau_n}D(t_n,s)\bar{U}(s,c,X^n(s))\mathrm{d}s+D(t_n,\tau_n)\varphi(\tau_n,X^n(\tau_n),Y^n(\tau_n))\}.
\end{equation*}
And then
\begin{align*}
	&\varphi(t_n,x_n,y_n)\!+\!\alpha_n\!\geq\!\mathbb{E}\{\int_{t_n}^{\tau_n}\!\!D(t_n,s)\bar{U}(s,c,X^n(s))\mathrm{d}s\!+\!D(t_n,\tau_n)\varphi(\tau_n,X^n(\tau_n),Y^n(\tau_n))\}\\
	&=\mathbb{E}\{\int_{t_n}^{\tau_n}\!\!D(t_n,s)[\bar{U}(s,c,X^n(s))-\mathcal{L}[\varphi](s,X^n(s),Y^n(s))-c\varphi_x(s,X^n(s),Y^n(s))]\mathrm{d}s\}\\
	&\phantom{eeeeeeeeeeeeeeeeeeeeeeeeeeeeeeeeeeeeedddddddddddddddddddddddd}+\varphi(t_n,x_n,y_n).
\end{align*}
According to the definition of $X^n$ and $Y^n$ and $\tau_n$, we know  $\tau_n(\omega)=t_n+\beta_n$ when $n\geq N(\omega)$ for some $N(\omega)\in \mathbb{N}$ large enough.  Hence, using the mean value theorem, we get
\begin{align*}
	&\lim\limits_{n\to \infty}	\frac{\int_{t_n}^{\tau_n}D(t_n,s)[\bar{U}(s,c,X^n(s))-\mathcal{L}[\varphi](s,X^n(s),Y^n(s))-c\varphi_x(s,X^n(s),Y^n(s))]\mathrm{d}s}{\beta_n}\\
	&=-\mathcal{L}[\varphi](t,x,y)-c\varphi_x(t,x,y)+U(c)+b\lambda(t)U(x),
\end{align*}
and using the dominated convergence theorem, we have
\begin{equation*}
	-\mathcal{L}[\varphi](t,x,y)-c\varphi_x(t,x,y)+U(c)+b\lambda(t)U(x)\leq 0.
\end{equation*}
Therefore, we complete the proof when $x>0$. Letting $c=0$ in the above proof, we can obtain a similar proof for the case of $x=0$.

Second, we prove that $V^*$ is a viscosity subsolution of \eqref{HJBQVI}. We take $(t,x,y)\in \mathbb{R}_+\times (\mathbb{R}^2_+\setminus{(0,0)})$ and $\varphi\in C^{1,2}(\mathbb{R}_+\times (\mathbb{R}^2_+\setminus{(0,0)}))$ such that $\varphi-V^*$ attains minimum $0$ at $(t,x,y)$. If  $V^*(t,x,y)\leq \mathcal{M}[V^*](t,x,y)$, the subsolution inequality holds trivially.
Thus, we assume
\begin{equation}
	V^*(t,x,y)>\mathcal{M}[V^*](t,x,y).\label{contradiction1}
\end{equation}
We want to show
\begin{equation*}
	\mathcal{L}[\varphi](t,x,y)-b\lambda(t)U(x)-f(\varphi_x(t,x,y))\leq 0.
\end{equation*}
We argue by contradiction and assume that there exists a $\eta>0$ such that
\begin{align*}
	\mathcal{L}[\varphi](t,x,y)-b\lambda(t)U(x)-f(\varphi_x(t,x,y))>\eta.
\end{align*}
A by-product is $\varphi_x(t,x,y)>0$.
Inspired by \cite{AS2017}, we can choose a $\alpha>0$ such that
\begin{equation*}
	\mathcal{L}[\varphi](t,x,y)-b\lambda(t)U(x)-f(\varphi_x(t,x,y)-\alpha)>\eta.
\end{equation*}
Then, there exists a $\delta>0$ such that $\overline{B_{2\delta}(t,x,y)}\subset \mathbb{R}_+\times (\mathbb{R}^2_+\setminus{(0,0)}) $. For all $(\bar{t},\bar{x},\bar{y})\in B_{2\delta}(t,x,y)$, we have
\begin{equation*}
	\mathcal{L}[\varphi](\bar{t},\bar{x},\bar{y})-b\lambda(\bar{t})U(\bar{x})-f(\varphi_x(\bar{t},\bar{x},\bar{y})-\alpha)>\eta,
\end{equation*}
or equivalently
\begin{equation*}
	\mathcal{L}[\varphi](\bar{t},\bar{x},\bar{y})-b\lambda(\bar{t})U(\bar{x})-\sup_{c\geq 0}\{U(c)+c\alpha-c\varphi_x(\bar{t},\bar{x},\bar{y})\}>\eta.
\end{equation*}
Therefore, for all $c\geq0 $ and  $(\bar{t},\bar{x},\bar{y})\in \overline{B_{2\delta}(t,x,y)}$, we have
\begin{equation*}
	\mathcal{L}[\varphi](\bar{t},\bar{x},\bar{y})-b\lambda(\bar{t})U(\bar{x})-U(c)+c\varphi_x(\bar{t},\bar{x},\bar{y})>\eta+c\alpha.
\end{equation*}
Now we choose $\{(t_n,x_n,y_n)\}_{n\geq 1}$ such that  $\lim\limits_{n\to \infty}(t_n,x_n,y_n,V(t_n,x_n,y_n))=(t,x,y,\\V^*(t,x,y))$ and $  B_{\delta}(t_n,x_n,y_n)\subset B_{2\delta}(t,x,y)$  for all $n\in\mathbb{N}$. We also choose $\{\epsilon_n\}_{n\geq 1}$ such that $0<\epsilon_n<\delta$ and $\epsilon_n\downarrow 0$. For $n\geq 1$, we choose $\nu^n=(c^n,\tau^n,\Delta^n)\in \mathcal{A}(t_n,x_n,y_n)$ such that
\begin{equation*}
	V(t_n,x_n,y_n)\leq \mathbb{E}\!\left[\int\limits_{t_n}^{+\infty}D(t_n,s)\bar{U}(s,c^n(s),X^{\nu^n}(s))\mathrm{d}s\right]+\epsilon_n.
\end{equation*}
Define
\begin{equation*}
	\theta^n=\tau^n_1\wedge\inf\{s\geq t_n|(s,X^{\nu^n}(s),Y^{\nu^n}(s))\notin B_{\delta} (t_n,x_n,y_n)\}.
\end{equation*}
Based on the proof of \eqref{DPP1}, we have
\begin{align*}
	&V(t_n,x_n,y_n)\\&\leq\!\epsilon_n+\mathbb{E}\left[\int\limits_{t_n}^{\theta^n}\!\!D(t_n,s)\bar{U}(s,c^{n}(s),X^{\nu^n}(s))\mathrm{d}s\!+\!D(t_n,\theta^n)\varphi(\theta^n,X^{\nu^n}(\theta^n-),Y^{\nu^n}(\theta^n-))\right]\!\\
	&=\!\epsilon_n+\varphi(t_n,x_n,y_n)\\
	&\!+\!\mathbb{E}\!\!\int\limits_{t_n}^{\theta^n}\!\!D(t_n,s)[\bar{U}(s,c^{n},X^{\nu^n}(s))\!-\!\mathcal{L}[\varphi](s,X^{\nu^n}(s),Y^{\nu^n}(s))\!-\!c\varphi_x(s,X^{\nu^n}(s),Y^{\nu^n}(s))]\mathrm{d}s\\
	&\leq \epsilon_n+ \varphi(t_n,x_n,y_n)+\mathbb{E}\int_{t_n}^{\theta^n}D(t_n,s)[-\eta-\alpha c^n(s)]\mathrm{d}s\\&\leq\epsilon_n+\varphi(t_n,x_n,y_n)+D(t_n,\theta^n)\mathbb{E}\int_{t_n}^{\theta^n}[-\eta-\alpha c^n(s)]\mathrm{d}s.
\end{align*}
Thus,
\begin{align*}
	&	\lim\limits_{n\to\infty}\mathbb{E}[\theta^n-t_n]=0,\\&\lim\limits_{n\to\infty}\mathbb{E}\int_{t_n}^{\theta^n}c^n(s)\mathrm{d}s=0.
\end{align*}
To find an upper bound for $V(t_n,x_n,y_n)$, we consider two disjoint sets $A_n=\{\tau^n_1>\theta^n\}$ and $A_n^c=\{\tau^n_1\leq\theta^n\}$. We proceed to prove  $\lim\limits_{n\to\infty}\mathbb{P}(A_n)=0$. We consider the processes before the first impulse happens:
\begin{align*}
	&X^n(s)=x_n+\int_{t_n}^s(rX^n(u)-c^n(u))\mathrm{d}u,\\
	&Y^n(s)=y_n+\int_{t_n}^s[\mu+\lambda(u)]Y^n(u)\mathrm{d}u+\int_{t_n}^{s}\sigma Y^n(u)\mathrm{d}W(u),
\end{align*}
or equivalently
\begin{align*}
	&X^n(s)=\mathrm{e}^{r(s-t_n)}x_n-\mathrm{e}^{rs}\int_{t_n}^s\mathrm{e}^{-ru}c^n(u)\mathrm{d}u,\\
	&Y^n(s)=y_n\mathrm{e}^{\int_{t_n}^s[\mu-\frac{1}{2}\sigma^2+\lambda(u)]\mathrm{d}u+\sigma[W(s)-W(t_n)]}.
\end{align*}
Define
\begin{equation*}
	\tilde{\theta}^n=\inf\{s\geq t_n|(s,X^{n}(s),Y^{n}(s))\notin B_{\delta} (t_n,x_n,y_n)\},
\end{equation*}
and then $A_n\subset \{\omega\in\Omega|\tau_1^n(\omega)>\tilde{\theta}^n(\omega)\}$. Now we turn to proving
\begin{equation*}
	\lim\limits_{n\to\infty}\mathbb{P}(A_n)=0.
\end{equation*}
If it is not valid, then there will exist a $0<\epsilon<1$ such that
\begin{equation*}
	\varlimsup_{n\to\infty}\mathbb{P}(A_n)=2\epsilon.
\end{equation*}
Without loss of generality, we assume  $\mathbb{P}(A_n)>\epsilon$ for all $n\in \mathbb{N}$. As $\theta^n$ is bounded  and $ \lim\limits_{n\to\infty}\mathbb{E}\int_{t_n}^{\theta^n}c^n(s)\mathrm{d}s=0$ , $ \mathrm{e}^{r\theta^n}\int_{t_n}^{\theta^n}\mathrm{e}^{-ru}c^n(u)\mathrm{d}u$ converges to $0$ in probability.  Thus, for $\frac{\delta}{2}>0$,
\begin{equation*}
	\lim\limits_{n\to \infty}\mathbb{P}(\mathrm{e}^{r\theta^n}\int_{t_n}^{\theta^n}\mathrm{e}^{-ru}c^n(u)\mathrm{d}u\geq\frac{\delta}{2})=0.
\end{equation*}
Letting $B_n=\{\mathrm{e}^{r\theta^n}\int_{t_n}^{\theta^n}\mathrm{e}^{-ru}c^n(u)\mathrm{d}u>\frac{\delta}{2}\}$, we assume  $\mathbb{P}(B_n)<\frac{\epsilon}{4}$ for all $n\in\mathbb{N}$. By the definition of $Y^n$  and the fact that  continuous function is uniformly continuous on a compact set, it is easy to see
\begin{equation*}
	\lim\limits_{0<\tilde{\delta}\downarrow 0}\sup_{n\in\mathbb{N}}\sup_{t_n\leq s\leq t_n+\tilde{\delta}}|Y^n(s)-y_n|=0\,\, a.s.\, \mathbb{P}.
\end{equation*}
Thus, there exists a $0<\tilde{\delta}<\delta$ such that
\begin{equation*}
	\mathbb{P}(\sup_{t_n\leq s\leq t_n+\tilde{\delta}}|Y^n(s)-y_n|\geq\delta)<\frac{\epsilon}{4},\, \forall n\in \mathbb{N}.
\end{equation*}
Letting $C_n=\{\sup\limits_{t_n\leq s\leq t_n+\tilde{\delta}}|Y^n(s)-y_n|\geq{\delta}\}$, we assume  $\sup\limits_{n}|\mathrm{e}^{r\tilde{\delta}}-1|x_n<\frac{\delta}{2}$. Then,
\begin{equation*}
	A_n\cap B_n^c\cap C_n^c\subset \{\tau_1^n>\tilde{\theta}^n,\tilde{\theta}^n-t_n>\tilde{\delta}\}=\{\tau_1^n>{\theta}^n,{\theta}^n-t_n>\tilde{\delta}\}\subset \{\theta^n-t_n>\tilde{\delta}\},
\end{equation*}
that is,
\begin{equation*}
	\mathbb{P}(\theta^n-t_n>\tilde{\delta})\geq \mathbb{P}(A_n\cap B_n^c\cap C_n^c)\geq\mathbb{P}(A_n)-\mathbb{P}(B_n)-\mathbb{P}(C_n)\geq \frac{\epsilon}{2},
\end{equation*}
which is in contradiction with $\lim\limits_{n\to\infty}\mathbb{E}[\theta^n-t_n]=0$. Therefore, $\lim\limits_{n\to\infty}\mathbb{P}(A_n)=0$.  But then,
\begin{align*}
	&V(t_n,x_n,y_n)\leq\epsilon_n+\mathbb{E}\left[\int\limits_{t_n}^{\theta^n}D(t_n,s)\bar{U}(s,c^{n}(s),X^{\nu^n}(s))\mathrm{d}s\right]+\\&\mathbb{E}\!\!\left\{D(t_n,\theta^n)\left[V^*(\theta^n,X^{\nu^n}(\theta^n\!-\!),Y^{\nu^n}(\theta^n\!-\!))1_{A_n}\!\!\!+\!\mathrm{M}[V^*](\theta^n,X^{\nu^n}(\theta^n\!-\!),Y^{\nu^n}(\theta^n\!-\!))1_{A_n^c}\right]\right\}\\
	&\leq\epsilon_n+\mathbb{E}\left[\int\limits_{t_n}^{\theta^n}D(t_n,s)\bar{U}(s,c^{n}(s),X^{\nu^n}(s))\mathrm{d}s\right]+\\&\mathbb{E}\!\left\{\!\!D(t_n,\theta^n)\!\!\left[V^*(\theta^n,X^{\nu^n}(\theta^n-),Y^{\nu^n}(\theta^n-))1_{A_n}\!+\!\!\sup_{(\bar{t},\bar{x},\bar{y})\in B_{2\delta}(t,x,y)}\mathrm{M}[V^*](\bar{t},\bar{x},\bar{y})1_{A_n^c}\right]\right\}.
\end{align*}
Letting $n\to \infty$, using $c^p\leq 1+c$, $V^*(\bar{t},\bar{x},\bar{y})\leq C(1+\bar{x}+\bar{y})^p$ and the dominated convergence theorem, we have
\begin{equation*}
	V^*(t,x,y)\leq \sup_{(\bar{t},\bar{x},\bar{y})\in B_{2\delta}(t,x,y)}\mathrm{M}[V^*](\bar{t},\bar{x},\bar{y}).
\end{equation*}
Letting $\delta\to 0$, we obtain
\begin{equation*}
	V^*(t,x,y)\leq \mathrm{M}[V^*]({t},{x},{y}),
\end{equation*}
which is in contradiction with \eqref{contradiction1}. Hence, we complete the proof.

\subsection{Proof of Theorem \ref{compare}}\label{proof of compare}
We re-define $v$ on $\mathbb{R}_+\times\{0\}\times\mathbb{R}_+$ by \eqref{v}, and $v$ is still a viscosity supersolution of \eqref{HJBQVI2} in any (relative) open subset of $\mathbb{R}_+\times(0,+\infty)\times\mathbb{R}_+$.

Using Lemma \ref{supsolution}, we first choose $p<q<1$ and $C>0$, and then define
$\Phi(t,x,y)=D(0,t)\Psi(t,x,y)=CD(0,t)(1+x+y)^q$ and $K(t,x,y)=D(0,t)K(x,y)$ such that
\begin{align*}
	\min\{\mathcal{\tilde{L}}&[\Phi](t,x,y)-bD(0,t)\lambda(t)U(x)-[D(0,t)]^{\frac{1}{1-p}}f(\Phi_x(t,x,y)),\nonumber\\ &\phantom{eeeeeeeeeeeeeeeeee}\Phi(t,x,y)-\mathcal{M}[\Phi](t,x,y)\}\geq K(t,x,y).
\end{align*}
Choosing $C$	large enough, we can assume  $0\leq u,v\leq \Phi$. Fixing $0<\epsilon<1$, we define
\begin{equation*}
	v_{\epsilon}=(1-\epsilon)v+\epsilon\Phi.
\end{equation*}
Now we will prove $u\leq v_{\epsilon}$ for all $0<\epsilon<1$. Then, letting $\epsilon\to 0$, we obtain $u\leq v$. We argue by contradiction as assuming
\begin{equation*}
	\Gamma=\sup_{(t,x,y)\in\mathbb{R}_+^3}(u(t,x,y)-v_{\epsilon}(t,x,y))>0.
\end{equation*}
Denote
\begin{align*}
	&F=\argmax_{(t,x,y)\in\mathbb{R}_+^3}(u(t,x,y)-v_{\epsilon}(t,x,y)),\\
	&\tilde{F}=\{(t,x,y)\in\mathbb{R}_+^3|(u(t,x,y)-v_{\epsilon}(t,x,y))>\frac{\Gamma}{2}\}.
\end{align*}
From \eqref{uv}, we deduce that $\tilde{F}$ is bounded in $(x,y)$ and then bounded in $t$.  Because $u-v_{\epsilon}$ is upper semi-continuous,    $F\subset\tilde{F}$ is nonempty and compact. That is why we consider \eqref{HJBQVI2} rather than \eqref{HJBQVI}.
Thus, we can choose an open (relative)  subset of $\mathbb{R}_+^3$ such that $F\subset G$ and $G$ is bounded. We also assume  $\bar{G}\cap(\mathbb{R}_+\times\{0\}\times\mathbb{R}_+)=\emptyset$ when $F\cap (\mathbb{R}_+\times\{0\}\times\mathbb{R}_+)=\emptyset$.  And we further assume  $(\mathbb{R}_+\times(0,0))\cap \bar{G}=\emptyset$ due to $(\mathbb{R}_+\times(0,0))\cap F=\emptyset$.

For simplicity, we denote  $(x,y)$ by $z$, and $|\cdot|$ represents the $2$-norm. That is, $|z|=\sqrt{x^2+y^2}$. For each $n\in \mathbb{N}\cup\{0\}$, we define
\begin{equation*}
	h_n(t,z,\hat{t},\hat{z})=u(t,z)-v_{\epsilon}(\hat{t},\hat{z})-\frac{n}{2}((t-\hat{t})^2+|z-\hat{z}|^2),
\end{equation*}
and
\begin{equation*}
	\Gamma_n=\sup_{(t,z,\hat{t},\hat{z})\in \bar{G}\times \bar{G}}	h_n(t,z,\hat{t},\hat{z}).
\end{equation*}
Thus, $\Gamma_0\geq \Gamma_1\geq\cdots\geq \Gamma_n\geq\cdots\geq \Gamma$.
As $h_n\in$ USC for all $n\in\mathbb{N}\cup\{0\}$,  we can choose $(t_n,z_n,\hat{t}_n,\hat{z}_n)\in \bar{G}\times\bar{G}$ such that
\begin{equation*}
	\Gamma_n=h_n(t_n,z_n,\hat{t}_n,\hat{z}_n),\quad \forall n\in\mathbb{N}\cup\{0\}.
\end{equation*}
Considering a subsequence, we assume  that $(t_n,z_n,\hat{t}_n,\hat{z}_n)$ converges to some point in $\bar{G}\times\bar{G}$.
Using
\begin{equation*}
	\frac{n}{2}((t_n-\hat{t}_n)^2+|z_n-\hat{z}_n|^2)\leq \Gamma_0,
\end{equation*}
we  have
\begin{equation*}
	\lim\limits_{n\to\infty}(t_n,z_n,\hat{t}_n,\hat{z}_n)=(\bar{t},\bar{z},\bar{t},\bar{z})\in \bar{G}\times\bar{G}.
\end{equation*}
Then
\begin{align*}
	&0\leq \varlimsup_{n\to \infty}\frac{n}{2}((t_n-\hat{t}_n)^2+|z_n-\hat{z}_n|^2)\\&=\varlimsup_{n\to \infty}( u(t_n,z_n)-v_{\epsilon}(\hat{t}_n,\hat{z}_n)-\Gamma_n)
	\leq u(\bar{t},\bar{z})-v_{\epsilon}(\bar{t},\bar{z})-\Gamma\leq 0.
\end{align*}
Therefore,
\begin{align*}
	&\lim\limits_{n\to\infty}(t_n,z_n,\hat{t}_n,\hat{z}_n)=(\bar{t},\bar{z},\bar{t},\bar{z}),\quad \lim\limits_{n\to \infty}n((t_n-\hat{t}_n)^2+|z_n-\hat{z}_n|^2)=0,\\
	&\lim\limits_{n\to \infty}u(t_n,z_n)=u(\bar{t},\bar{z}),\quad \lim\limits_{n\to \infty}v_{\epsilon}(\hat{t}_n,\hat{z}_n)=v_{\epsilon}(\bar{t},\bar{z}),\quad \lim\limits_{n\to\infty}\Gamma_n\!=\!\Gamma\!=\!u(\bar{t},\bar{z})-v_{\epsilon}(\bar{t},\bar{z}).
\end{align*}
Thus, $(\bar{t},\bar{z})\in F$ and $(t_n,z_n,\hat{t}_n,\hat{z}_n)\in G\times G$	when $n$ is large enough.
Without loss of generality, we assume  $(t_n,z_n,\hat{t}_n,\hat{z}_n)\in G\times G$ for all $n\in \mathbb{N}$. 	

Case 1: $\bar{x}>0$.

We employ Ishii's lemma (see \cite{GIL(1992)} and \cite{FS2005}). There are $X_n \in\mathbb{S}^3$ and $Y_n\in\mathbb{S}^3$ with
\begin{equation*}
	\begin{pmatrix}
		X_n  & 0\\
		0 & -Y_n
	\end{pmatrix}\leq 3n
	\begin{pmatrix}
		1 & 0 & 0&0\\
		0 & I_2 & 0&-I_2\\
		0 & 0 & 1&0\\
		0 &- I_2 & 0&I_2\\
	\end{pmatrix},
\end{equation*}
such that
\begin{align*}
	&(n(t_n-\hat{t}_n),n(z_n-\hat{z}_n),X_n)\in \overline{\mathcal{J}}^{2,+}u(t_n,z_n),\\
	&	(n(t_n-\hat{t}_n),n(z_n-\hat{z}_n),Y_n)\in \overline{\mathcal{J}}^{2,-}v_{\epsilon}(\hat{t}_n,\hat{z}_n).
\end{align*}
Let $M_n(N_n)\in\mathbb{S}^2$ be the matrix obtained by removing the first row and the first column of  $X_n(Y_n)$, then
\begin{equation}
	\begin{pmatrix}
		M_n  & 0\\
		0 & -N_n
	\end{pmatrix}\leq 3n
	\begin{pmatrix}
		I_2&-I_2\\
		-I_2&I_2
	\end{pmatrix},\label{ine}
\end{equation}
and
\begin{align*}
	&(n(t_n-\hat{t}_n),n(z_n-\hat{z}_n),M_n)\in \overline{\mathcal{P}}^{2,+}u(t_n,z_n),\\
	&	(n(t_n-\hat{t}_n),n(z_n-\hat{z}_n),N_n)\in \overline{\mathcal{P}}^{2,-}v_{\epsilon}(\hat{t}_n,\hat{z}_n).
\end{align*}
Because $u$ is a viscosity subsolution of \eqref{HJBQVI2}, we have
\begin{equation*}
	\min\{F(t_n,z_n,n(t_n-\hat{t}_n),n(z_n-\hat{z}_n),M_n),u(t_n,z_n)-\mathcal{M}[u](t_n,z_n)\}\leq 0.
\end{equation*}
Using Lemma \eqref{Ishi}, we have
\begin{equation}
	\min\{F(\hat{t}_n,\hat{z}_n,n(t_n-\hat{t}_n),n(z_n-\hat{z}_n),N_n),v_{\epsilon}(\hat{t}_n,\hat{z}_n)-\mathcal{M}[v_{\epsilon}]_*(\hat{t}_n,\hat{z}_n)\}\geq \epsilon\tilde{k}.\label{contra2}
\end{equation}
After dropping a subsequence, we can assume
\begin{equation}
	F(t_n,z_n,n(t_n-\hat{t}_n),n(z_n-\hat{z}_n),M_n)\leq 0.\label{contra1}
\end{equation}
The proof is inspired by the proof of Step 3 of Theorem 5.4 in \cite{BS2021} and the proof of Theorem 3.8 in \cite{OS2002}. Throughout the proof, $N\in\mathbb{N}$ is a constant that may vary from line to line. We argue by contradiction and assume
\begin{equation}\label{Q1}
	u(t_n,z_n)-\mathcal{M}[u](t_n,z_n)\leq 0, \quad \forall n>N.
\end{equation}
We can then deduce  $z_n\notin S_{\emptyset}$ for all $n>N$.  According to the definition of $\Gamma_n$, we have
\begin{equation}\label{Q2}
	\Gamma=u(\bar{t},\bar{z})-v_{\epsilon}(\bar{t},\bar{z})\leq u(t_n,z_n)-v_{\epsilon}(\hat{t}_n,\hat{z}_n),\quad \forall n>N.
\end{equation}
Moreover, because $\Phi(\hat{t}_n,\hat{z}_n)$ converges to $\Phi(\bar{t},\bar{z})>0$, we suppose
\begin{equation}\label{Q3}
	\Phi(\hat{t}_n,\hat{z}_n)>\frac 34 \Phi(\bar{t},\bar{z}),\quad \forall n>N.
\end{equation}
We now choose $\zeta=\frac 12 \min\{\epsilon\Phi(\bar{t},\bar{z}), \frac 12\epsilon \tilde{k}\}>0$. The upper semi-continuity of $\mathcal{M}[u]$ yields
\begin{equation}\label{Q4}
	\mathcal{M}[u](t_n,z_n)\leq \mathcal{M}[u](\bar{t},\bar{z})+\zeta,\quad  \forall n>N,
\end{equation}
and \eqref{Q1}-\eqref{Q4} give
\begin{align}
	\Gamma\leq \mathcal{M}[u](\bar{t},\bar{z})+\zeta-v_{\epsilon}(\hat{t}_n,\hat{z}_n)\label{Q7}\\
	\leq \mathcal{M}[u](\bar{t},\bar{z})-\frac 14 \epsilon\Phi(\bar{t},\bar{z})\nonumber.
\end{align}
If $z\in \partial S_{\emptyset}$, $\mathcal{M}[u](\bar{t},\bar{z})=0$. Thus we have $\Gamma\leq-\frac 14 \epsilon\Phi(\bar{t},\bar{z})<0$, which is in contradiction with $\Gamma>0$. Therefore, $z\in \mathbb{R}^2_+\setminus\bar{S_{\emptyset}}$. We then assume
\begin{equation}\label{Q5}
	\mathcal{M}[v_{\epsilon}]_*(\hat{t}_n,\hat{z}_n)>\mathcal{M}[v_{\epsilon}]_*(\bar{t},\bar{z})-\frac 12 \epsilon\tilde{k}=\mathcal{M}[v_{\epsilon}](\bar{t},\bar{z})-\frac 12 \epsilon\tilde{k},\quad \forall n>N.
\end{equation}
Together with \eqref{contra2}, \eqref{Q7} and \eqref{Q5}, we have
\begin{equation*}
	\Gamma\leq \mathcal{M}[u](\bar{t},\bar{z})-\mathcal{M}[v_{\epsilon}](\bar{t},\bar{z})+\zeta-\frac 12 \epsilon \tilde{k}\leq \Gamma-\frac 14 \epsilon \tilde{k},
\end{equation*}
which is a contradiction. Moreover, \eqref{contra2} and \eqref{contra1} give
\begin{equation*}
	F(\hat{t}_n,\hat{z}_n,n(t_n-\hat{t}_n),n(z_n-\hat{z}_n),N_n)-		F(t_n,z_n,n(t_n-\hat{t}_n),n(z_n-\hat{z}_n),M_n)\geq \epsilon\tilde{k}.
\end{equation*}
Similar to the proof of Example 3.6 in \cite{GIL(1992)},   using \eqref{ine}, we can prove that there is a constant $L>0$ such that
\begin{align*}
	F(\hat{t}_n,\hat{z}_n,n(t_n-\hat{t}_n),n(z_n-\hat{z}_n),N_n)-		F(t_n,z_n,n(t_n-\hat{t}_n),n(z_n-\hat{z}_n),M_n)
	\\ \leq Ln((t_n-\hat{t}_n)^2+|z_n-\hat{z}_n|^2)+bD(0,t_n)\lambda(t_n)U(x_n)-bD(0,\hat{t}_n)\lambda(\hat{t}_n)U(\hat{x}_n)\\
	+\left(D(0,t_n)^{\frac{1}{1-p}}-D(0,\hat{t}_n)^{\frac{1}{1-p}}\right)f(n(x_n-\hat{x}_n)).
\end{align*}
If we can show $\varliminf\limits_{n\to \infty}n(x_n-\hat{x}_n)>\eta$ for some $\eta>0$, then $n\to \infty$ will give a contradiction.
Indeed, we consider
\begin{equation*}
	g_n(x)\triangleq  h_n(t_n,z_n,\hat{t}_n,x,\hat{y}_n)=C_n -(1-\epsilon)v(\hat{t}_n,x,\hat{y}_n)-\epsilon\Phi(\hat{t}_n,x,\hat{y}_n)-\frac n2 (x-x_n)^2,
\end{equation*}
where $C_n$ is a constant. Then
\begin{equation*}
	\sup_{\{x|(\hat{t}_n,x,\hat{y}_n)\in G\}}g_n(x)=g_n(\hat{x}_n).
\end{equation*}
Note that the function
\begin{equation*}
	m_n(x)=-\epsilon\Phi(\hat{t}_n,x,\hat{y}_n)-\frac n2 (x-x_n)^2
\end{equation*}
has the first and second derivatives as
\begin{align*}
	&m'_n(x)=-\epsilon q CD(0,\hat{t}_n)(1+x+\hat{y}_n)^{q-1}-n (x-x_n),\\
	&m''_n(x)=-\epsilon q(q-1) CD(0,\hat{t}_n)(1+x+\hat{y}_n)^{q-2}-n.
\end{align*}
When $n$ is large enough,  $m''_n<0$. Since $\bar{x}>0$, when $n$ is large enough, $m_n(x)$ attains its maximum over $\mathbb{R}_+$ at a unique point $ \tilde{x}_n>0$, where $\tilde{x}_n$ satisfies $(\hat{t}_n,\tilde{x}_n,\hat{y}_n)\in G$ and $\tilde{x}_n$ solves
\begin{equation*}
	-\epsilon q C(1+\tilde{x}_n+\hat{y}_n)^{q-1}-n (\tilde{x}_n-x_n)=0.
\end{equation*}
Using the monotonicity of $v$ in $x$, we have $\hat{x}_n\leq\tilde{x}_n$. Thus, we obtain $$\varliminf_{n\to \infty}n(x_n-\hat{x}_n)>\eta, \ \ \mbox{ for some $\eta>0$}.$$

Case 2: $\bar{x}=0$.

We choose $(t_n,x_n,y_n)\in G$, $x_n>0$ such that $\lim\limits_{n\to\infty}(t_n,x_n,y_n,v(t_n,x_n,y_n))=(\bar{t},0,\bar{y},v(\bar{t},0,\bar{y}))$. Letting $\epsilon_n=\sqrt{(t_n-\bar{t})^2+x_n^2+(y_n-\bar{y})^2}$, we define
\begin{equation*}
	Q_n(t,z,\hat{t},\hat{z})=u(t,z)-v_{\epsilon}(\hat{t},\hat{z})-\varphi(t,z,\hat{t},\hat{z}),
\end{equation*}
where
\begin{equation*}
	\varphi(t,z,\hat{t},\hat{z})=\frac 14 (t-\bar{t})^4\!+\!\frac 14|x-\bar{x}|^4\!+\!\frac 14|y-\bar{y}|^4\!+\!\frac{(t-\hat{t})^2+|z-\hat{z}|^2}{2\epsilon_n}\!+\!\frac 14\left(\frac{\hat{x}-x}{x_n}-1\right)^4.
\end{equation*}
We then choose $(t^n,z^n,\hat{t}^n,\hat{z}^n)\in \bar {G}\times \bar{G}$ such that
\begin{equation*}
	\Sigma_n\triangleq\sup_{(t,z,\hat{t},\hat{z})\in \bar {G}\times \bar{G}}Q_n(t,z,\hat{t},\hat{z})=Q_n(t^n,z^n,\hat{t}^n,\hat{z}^n).
\end{equation*}
Without loss of generality, we assume $\lim\limits_{n\to \infty}(t^n,z^n,\hat{t}^n,\hat{z}^n)=(t^*,z^*,\hat{t}^*,\hat{z}^*)\in \bar{G}\times\bar{G}$. We have
\begin{equation*}
	\Sigma_n\geq Q_n(\bar{t},0,\bar{y},t_n,x_n,y_n)=u(\bar{t},0,\bar{y})-v_{\epsilon}(t_n,x_n,y_n)-\frac{\epsilon_n}{2}\rightarrow \Gamma (n\to \infty).
\end{equation*}
As $u$ and $-v_{\epsilon}$ are upper bounded in $\bar{G}$, $\frac{(t^n-\hat{t}^n)^2+|z^n-\hat{z}^n|^2}{\epsilon_n}$ is bounded. Therefore, $(\hat{t}^*,\hat{z}^*)=(t^*,z^*)$ and then
\begin{equation*}
	\varlimsup_{n\to \infty}\varphi(t^n,z^n,\hat{t}^n,\hat{z}^n)\!=\!\varlimsup_{n\to \infty}(u(t^n,z^n)\!-\!v_{\epsilon}(\hat{t}^n,\hat{z}^n)\!-\!\Sigma_n)\leq u(t^*,z^*)\!-\!v_{\epsilon}(t^*,z^*)\!-\!\Gamma\leq 0,
\end{equation*}
which gives
\begin{equation*}
	(t^*,z^*)=(\bar{t},0,\bar{y}),\lim\limits_{n\to \infty}\frac{(t^n-\hat{t}^n)^2+|z^n-\hat{z}^n|^2}{\epsilon_n}=0,\lim\limits_{n\to \infty}\frac{\hat{x}^n-x^n}{x_n}=1.
\end{equation*}
Thus, we can assume $\hat{x}^n>x^n\geq0$ for all $n\in \mathbb{N}$, and $(t^n,z^n,\hat{t}^n,\hat{z}^n)\in G\times G$ for all $n\in \mathbb{N}$. By Ishii's lemma, for some $\eta_n\downarrow 0$, there are $M_n\in\mathbb{S}^2$ and $N_n\in\mathbb{S}^2$ such that
\begin{align*}
	&(q_n,p_n,r_n,M_n)\in \overline{\mathcal{P}}^{2,+}u(t^n,z^n),\\
	&	(q'_n,p'_n,r'_n,N_n)\in \overline{\mathcal{P}}^{2,-}v_{\epsilon}(\hat{t}^n,\hat{z}^n),\\
	&	\begin{pmatrix}
		M_n  & 0\\
		0 & -N_n
	\end{pmatrix}\leq D^2\varphi_{z,\hat{z}}(t^n,z^n,\hat{t}^n,\hat{z}^n)+\eta_n [D^2\varphi_{z,\hat{z}}(t^n,z^n,\hat{t}^n,\hat{z}^n)]^2,
\end{align*}
where
\begin{align*}
	q_n=\frac{t^n-\hat{t}^n}{\epsilon_n}+(t^n-\bar{t})^3,\quad & q'_n=\frac{t^n-\hat{t}^n}{\epsilon_n},\\
	p_n=(x^n-\bar{x})^3+  \frac{x^n-\hat{x}^n}{\epsilon_n}\!-\!\frac{1}{x_n}(\frac{\hat{x}^n-x^n}{x_n}-1)^3,\!\quad& \!
	p'_n=  \frac{x^n-\hat{x}^n}{\epsilon_n}\!-\!\frac{1}{x_n}(\frac{\hat{x}^n-x^n}{x_n}-1)^3,\\ r_n=(y^n-\bar{y})^3+\frac{y^n-\hat{y}^n}{\epsilon},\quad &r'_n=\frac{y^n-\hat{y}^n}{\epsilon},
\end{align*}
and
\begin{equation*}
	D^2\varphi_{z,\hat{z}}(t^n,z^n,\hat{t}^n,\hat{z}^n)=	\begin{pmatrix}
		A_n+B_n  & -A_n\\
		-A_n & A_n
	\end{pmatrix},
\end{equation*}
where
\begin{equation*}
	A_n=\begin{pmatrix}
		\frac{1}{\epsilon_n} + \frac{1}{x^2_n}(\frac{\hat{x}^n-x^n}{x_n}-1)^2& 0\\
		0 & \frac{1}{\epsilon_n}
	\end{pmatrix},\quad B_n=\begin{pmatrix}
		3(x^n-\bar{x})^2  & 0\\
		0 & 3(y^n-\bar{y})^2
	\end{pmatrix}.
\end{equation*}
We choose $\eta_n\downarrow 0$ such that $\eta_n [D^2\varphi_{z,\hat{z}}(t^n,z^n,\hat{t}^n,\hat{z}^n)]^2\rightarrow 0(n\to \infty)$. As $\hat{x}^n>0$, we apply Lemma \ref{Ishi} and  treat $\mathcal{M}$ by the same way as in the first part to establish
\begin{equation*}
	F(\hat{t}^n,\hat{z}^n,q'_n,q'_n,r'_n,N_n)-F(t^n,z^n,q_n,p_n,r_n,M_n)\geq \epsilon\tilde{k}.
\end{equation*}
We proceed to prove that there exists a $\bar{\eta}>0$ such that
$\varliminf\limits_{n\to \infty}p'_n\geq \bar{\eta}$. We can then get a contradiction by the  same way as in the first part, which completes the proof. Indeed, we define
\begin{align*}
	f_n(x)&\triangleq  Q_n(t^n,z^n,\hat{t}^n,x,\hat{y}^n)\\&=C_n -(1-\epsilon)v(\hat{t}^n,x,\hat{y}^n)-\epsilon\Phi(\hat{t}^n,x,\hat{y}^n)-\frac { (x-x^n)^2}{2\epsilon_n}-\frac 14\left(\frac{x-x^n}{x_n}-1\right)^4.
\end{align*}
Then
\begin{equation*}
	\Sigma_n=\sup_{(\hat{t}^n,x,\hat{y}^n)\in G} f_n(x)=f_n(\hat{x}^n).
\end{equation*}
Note that the function
\begin{equation*}
	\bar{m}_n(x)=-\epsilon\Phi(\hat{t}^n,x,\hat{y}^n)-\frac { (x-x^n)^2}{2\epsilon_n}-\frac 14\left(\frac{x-x^n}{x_n}-1\right)^4
\end{equation*}
has the following first and second derivatives:
\begin{align*}
	&\bar{m}'_n(x)=-\epsilon q CD(0,\hat{t}^n)(1+x+\hat{y}_n)^{q-1}-\frac{(x-x^n)}{\epsilon_n}-\frac{1}{x_n}\left(\frac{x-x^n}{x_n}-1\right)^3,\\
	&\bar{m}''_n(x)=-\epsilon q(q-1) CD(0,\hat{t}^n)(1+x+\hat{y}_n)^{q-2}-\frac{1}{\epsilon_n}-\frac{1}{x_n^2}\left(\frac{x-x^n}{x_n}-1\right)^2.
\end{align*}
When $n$ is large enough,  $\bar{m}''_n<0$ and  $\bar{m}_n'(x^n)>0$. As $\hat{x}^n>x^n$ and $v$ is non-decreasing in $x$, we  have $m'_n(\hat{x}^n)\geq 0$ which means $\varliminf\limits_{n\to \infty}p'_n\geq \bar{\eta}$ for some $\bar{\eta}>0$.

In the above proof, we employ the following lemma which is used in the literature. However, in \cite{OS2002}, the paper does not give the proof. A similar lemma appears in \cite{BC(2019)}, but the paper does not consider the stochastic control. Therefore, we give a detailed proof here. Particularly, the validity of the lemma depends on the concavity of $-f$.

\begin{lemma}\label{Ishi}
	Let $(t_0,z_0)\in (\mathbb{R}_+\times(0,+\infty)\times \mathbb{R}_+)\cap B_{\delta}(\bar{t},\bar{z}) $ for some $\bar{z}\neq (0,0)$ and $0<\delta<{|\bar{z}|}$. If $(p_0,q_0,X_0)\in \overline{\mathcal{P}}^{2,-}v_{\epsilon}(t_0,z_0)$, then
	\begin{equation*}
		\min\{F(t_0,z_0,q_0,p_0,X_0),v_{\eta}(t_0,z_0)-\mathcal{M}[v_{\eta}]_*(t_0,z_0)\}\geq \epsilon{\tilde{k}},
	\end{equation*}
	where
	\begin{equation*}
		\tilde{k}=\min_{(t,z)\in\overline{B_{\delta}(\bar{t},\bar{z})}}K(t,z)>0.
	\end{equation*}
\end{lemma}
\begin{proof}
	As $(p_0,q_0,X_0)\in \overline{\mathcal{P}}^{2,-}v_{\eta}(t_0,z_0)$, we can choose $(t_n,z_n)\in B_{\delta}(\bar{t},\bar{z})$ and $(p_n,q_n,X_n)\in \mathcal{P}^{2,-}v_{\epsilon}(t_n,z_n)$ such that
	\begin{equation*}
		\lim\limits_{n\to\infty}(t_n,z_n,q_n,p_n,X_n,v_{\epsilon}(t_n,z_n))=(t_0,z_0,q_0,p_0,X_0,v_{\epsilon}(t_0,z_0)).
	\end{equation*}
	For each $n\in\mathbb{N}$, we choose $\varphi_n\in C^{1,2}(S)$ such that $\varphi_n-v_{\epsilon}$ attains its maximum $0$ at $(t_n,z_n)$ and $(\varphi_t(t_n,z_n),\varphi_z(t_n,z_n),\mathrm{D}^2_z\varphi(t_n,z_n))=(q_n,p_n,X_n)$. $\varphi_n\leq v_{\epsilon}$ gives  $v\geq \frac{1}{1-\epsilon}\varphi_n-\frac{\epsilon}{1-\epsilon}\Phi$. As $v$ is a viscosity supersolution, we have
	\begin{equation}
		\min\{F\left[\frac{1}{1-\epsilon}\varphi_n-\frac{\epsilon}{1-\epsilon}\Phi\right](t_n,z_n),v(t_n,z_n)-\mathcal{M}[v]_*(t_n,z_n)\}\geq 0.\label{concave}
	\end{equation}
	By the concavity of $-f$, for all $a$, $b\in \mathbb{R}$,
	\begin{align*}
		&
		-(1-\epsilon)f(\frac{1}{1-\epsilon}a-\frac{\epsilon}{1-\epsilon}b)-\epsilon f(b)\leq -f(a)\\		\Rightarrow &
		-f(\frac{1}{1-\epsilon}a-\frac{\epsilon}{1-\epsilon}b)\leq -\frac{1}{1-\epsilon}f(a)+\frac{\epsilon}{1-\epsilon}f(b).
	\end{align*}
	Combining with \eqref{concave}, we have
	\begin{equation*}
		\min\{\frac{1}{1-\epsilon}F[\varphi_n](t_n,z_n)-\frac{\epsilon}{1-\epsilon}F[\Phi](t_n,z_n)(t_n,z_n),v(t_n,z_n)-\mathcal{M}[v]_*(t_n,z_n)\}\geq 0,
	\end{equation*}
	or equivalently
	\begin{equation*}
		\!\min\{F[\varphi_n](t_n,z_n)\!-\!\epsilon F[\Phi](t_n,z_n)(t_n,z_n),(v_{\eta}\!-\!\epsilon\Phi)(t_n,z_n)\!-\!\mathcal{M}[v_{\eta}\!-\!\epsilon\Phi]_*(t_n,z_n)\}\!\geq\! 0.
	\end{equation*}
	Therefore,
	\begin{equation}
		F[\varphi_n](t_n,z_n)\geq \epsilon{\tilde{k}},\label{limit1}
	\end{equation}
	and
	\begin{equation*}
		(v_{\eta}-\epsilon\Phi)(t_n,z_n)-\mathcal{M}[v_{\eta}-\epsilon\Phi]_*(t_n,z_n)\geq 0.
	\end{equation*}
	Using Point 7 of  Proposition \ref{VM}, we have
	\begin{equation*}
		(v_{\eta}-\epsilon\Phi)(t_n,z_n)\geq \mathcal{M}[v_{\eta}-\epsilon\Phi]_*(t_n,z_n)\geq \mathcal{M}[v_{\eta}]_*(t_n,z_n)-\epsilon\mathcal{M}[\Phi](t_n,z_n),
	\end{equation*}
	and then
	\begin{equation}
		v_{\eta}(t_n,z_n)-\mathcal{M}[v_{\eta}]_*(t_n,z_n)\geq	\epsilon{\tilde{k}}.\label{limit2}
	\end{equation}
	Taking $\varlimsup\limits_{n\to\infty}$ in \eqref{limit1} and \eqref{limit2}, we complete the proof.
\end{proof}

\begin{thebibliography}{99}
	\baselineskip 17pt
	\bibitem[Akian, Menaldi and Sulem(1996)]{AMS1996}Akian, M., Menaldi, J. L., Sulem, A. (1996). On an investment-consumption model with transaction costs. SIAM Journal on Control and Optimization, 34(1), 329-364.
	\bibitem[Akian, Sulem and Taksar(2001)]{AST2001}Akian, M., Sulem, A.,  Taksar, M. I. (2001). Dynamic optimization of long‐term growth rate for a portfolio with transaction costs and logarithmic utility. Mathematical Finance, 11(2), 153-188.
	\bibitem[Altarovici, Reppen and Soner(2017)]{AS2017}Altarovici, A., Reppen, M.,  Soner, H. M. (2017). Optimal consumption and investment with fixed and proportional transaction costs. SIAM Journal on Control and Optimization, 55(3), 1673-1710.
	\bibitem[Azimzadeh and Forsyth(2016)]{AFP2016}Azimzadeh, P.,  Forsyth, P. A. (2016). Weakly chained matrices, policy iteration, and impulse control. SIAM Journal on Numerical Analysis, 54(3), 1341-1364.
	\bibitem[Barles(1994)]{B1994}Barles, G. (1994). Solutions de Viscosité des {\'E}quations de Hamilton-Jacobi (Vol. 17). Berlin: Springer.
	\bibitem[Belak and Seifried(2021)]{BS2021}Belak, C., Mich, L., Seifried, F. T. (2021). Optimal investment for retail investors. Mathematical Finance, 32(2), 555-594.
	\bibitem[Bernhardt and Donnelly(2019)]{BD2019} Bernhardt, T., Donnelly, C. (2019). Modern tontine with bequest: Innovation in pooled annuity products. Insurance: Mathematics and Economics, 86, 168-188.
	\bibitem[Belak and Christensen(2019)]{BC(2019)}Belak, C.,  Christensen, S. (2019). Utility maximisation in a factor model with constant and proportional transaction costs. Finance and Stochastics, 23(1), 29-96.
	\bibitem[Bouchard and Touzi(2011)]{BT(2011)}Bouchard, B.,  Touzi, N. (2011). Weak dynamic programming principle for viscosity solutions. SIAM Journal on Control and Optimization, 49(3), 948-962.
	\bibitem[Chen and Rach(2022)]{CR2022}Chen, A.,  Rach, M. (2022). Bequest-embedded annuities and tontines. Asia-Pacific Journal of Risk and Insurance, 16(1), 1-46.
	\bibitem[Crandall, Ishii and Lions(1992)]{GIL(1992)}Crandall, M. G., Ishii, H.,  Lions, P. L. (1992). User’s guide to viscosity solutions of second order partial differential equations. Bulletin of the American Mathematical Society, 27(1), 1-67.
	\bibitem[Dagpunar(2021)]{D2021}Dagpunar, J. (2021). Closed-form solutions for an explicit modern ideal tontine with bequest motive. Insurance: Mathematics and Economics, 100, 261-273.
	\bibitem[Fleming and Soner(2005)]{FS2005}Fleming, W. H.,  Soner, H. M. (2006). Controlled Markov Processes and Viscosity Solutions (Vol. 25). Springer Science  Business Media.
	\bibitem[Framstad, Øksendal and Sulem(2001)  ]{FOS2001}Framstad, N. C., Øksendal, B.,  Sulem, A. (2001). Optimal consumption and portfolio in a jump diffusion market with proportional transaction costs. Journal of Mathematical Economics, 35(2), 233-257.
	\bibitem[Hainaut and Devolder(2006)]{HD2006}Hainaut, D., Devolder, P. (2006). Life annuitization: Why and how much? Astin Bulletin, 36(2), 629-654.
	\bibitem[Horneff, Maurer, Mitchell and Stamos(2010)]{HMMS2010}Horneff, W., Maurer, R., Mitchell, O., Stamos, M. (2010). Variable payout annuities and dynamic portfolio choice in retirement. Journal of Pension Economics and Finance, 9, 163-183.
	\bibitem[Horneff, Maurer and Stamos(2008)]{HMS2008}Horneff, W., Maurer, R., Stamos, M. (2008). Optimal gradual annuitization: Quantifying the costs of switching to annuities. Journal of Risk and Insurance, 75(4), 1019-1038.
	\bibitem[Ishii(1993)]{I1993}Ishii, K. (1993). Viscosity solutions of nonlinear second order elliptic PDEs associated with impulse control problems. Funkcial. Ekvac, 36(1), 123-141.
	\bibitem[Jacod and Shiryaev(2013)]{JS2013}Jacod, J.,  Shiryaev, A. (2013). Limit theorems for stochastic processes (Vol. 288). Springer Science \& Business Media.
	\bibitem[Kabanov and Klüppelberg(2004)]{KK2004}Kabanov, Y.,  Klüppelberg, C. (2004). A geometric approach to portfolio optimization in models with transaction costs. Finance and Stochastics, 8(2), 207-227.
	\bibitem[Karoui and Tan(2013)]{KT2003}Karoui, N. E.,  Tan, X. (2013). Capacities, measurable selection and dynamic programming part II: Application in stochastic control problems. arXiv preprint arXiv:1310.3364.
	\bibitem[McKeever(2009)]{M2009}McKeever, K. (2009). A short history of tontines. Fordham J. Corp.  Fin. L., 15, 491.
	\bibitem[Merton(1969)]{M1969}Merton, R. C. (1969). Lifetime portfolio selection under uncertainty: The continuous-time case. The review of Economics and Statistics, 51(3), 247-257.
	\bibitem[Merton(1971)]{M1971}Merton, R. (1971). Optimum consumption and portfolio rules in a continuous-time model. Journal of Economic Theory, 3, 373-413.
	\bibitem[Milevsky and Salisbury(2015)]{MS2015}Milevsky, M. A., Salisbury, T. S. (2015). Optimal retirement income tontines. Insurance: Mathematics and Economics, 64, 91-105.
	\bibitem[Pham(2009)]{Pham2009}Pham, H. (2009). Continuous-time stochastic control and optimization with financial applications (Vol. 61). Springer Science  Business Media.
	\bibitem[Seydel(2009)]{Seydel2009}Seydel, R. C. (2009). Existence and uniqueness of viscosity solutions for QVI associated with impulse control of jump-diffusions. Stochastic Processes and their Applications, 119(10), 3719-3748.
	\bibitem[Shreve and Soner(1994)]{SS(1994)}Shreve, S. E.,  Soner, H. M. (1994). Optimal investment and consumption with transaction costs. The Annals of Applied Probability, 4, 609-692.
	\bibitem[Vath, Mnif and Pham(2007)]{VMP(2007)}Vath, V. L., Mnif, M.,  Pham, H. (2007). A model of optimal portfolio selection under liquidity risk and price impact. Finance and Stochastics, 11(1), 51-90.
	\bibitem[Weinert(2017)]{Weinert2017}Weinert, J. H. (2017). The fair surrender value of a tontine. ICIR Working Paper Series (2017).
	\bibitem[Weinert and  Gründl(2021)]{WG2021}Weinert, J. H.,  Gründl, H. (2021). The modern tontine: An innovative instrument for longevity risk management in an aging society. European Actuarial Journal, 11(1), 49-86.
	\bibitem[{\O}ksendal and  Sulem(2002)]{OS2002}{\O}ksendal, B.,  Sulem, A. (2002). Optimal consumption and portfolio with both fixed and proportional transaction costs. SIAM Journal on Control and Optimization, 40(6), 1765-1790.
\end{thebibliography}
\end{document}